\documentclass{IEEEtran}
\usepackage{amsmath,amssymb,amsfonts}
\usepackage{amsthm}
\usepackage{algorithm}
\usepackage{algorithmic}
\usepackage{graphicx}
\usepackage{textcomp}
\usepackage{nicefrac}

\usepackage{booktabs}
\usepackage{array}
\usepackage{xcolor}
\usepackage{pifont}
\usepackage{caption}
\usepackage{graphicx}
\usepackage{amsmath}
\usepackage{amsfonts}
\usepackage{amssymb}
\usepackage{cases}
\usepackage{amsthm}
\usepackage[colorlinks=true,
            linkcolor=blue,
            citecolor=blue,
            urlcolor=blue]{hyperref}

\newcommand{\lf}{\left}
\newcommand{\rg}{\right}
\DeclareMathOperator*{\argmax}{\arg\max}

\newtheorem{theorem}{Theorem}
\newtheorem{observation}{Observation}
\newtheorem{lemma}{Lemma}
\newtheorem{corollary}{Corollary}

\theoremstyle{definition}
\newtheorem{definition}{Definition}

\newtheorem{assumption}{Assumption}
\newtheorem{example}{Counter-Example}

\theoremstyle{remark}
\newtheorem{remark}{Remark}

\renewcommand{\phi}{\varphi}
\renewcommand{\epsilon}{\varepsilon}

\newcommand{\R}{\mathbb{R}}
\newcommand{\extended}{\bar{\R}}
\newcommand{\N}{\mathbb{N}}
\newcommand{\Rn}{\R^n}
\newcommand{\Ru}{\R^{p_1}}
\newcommand{\Rd}{\R^{p_2}}

\newcommand{\Rm}{\R^m}

\newcommand{\sbs}{\subseteq}
\newcommand{\ssbs}{\subset}
\newcommand{\sps}{\supseteq}

\newcommand{\tms}{\times}
\newcommand{\ii}{\infty}

\newcommand{\dynamics}{f}

\newcommand{\tVal}{t}
\newcommand{\tNot}{\tVal_0}
\newcommand{\tDum}{s}
\newcommand{\tDumDum}{\tau}
\newcommand{\tInit}{S}
\newcommand{\tFin}{T}
\newcommand{\tFinDum}{S}

\newcommand{\xVal}{x}
\newcommand{\xNot}{\xVal_0}

\newcommand{\uVal}{u}

\newcommand{\dVal}{d}

\newcommand{\uVals}{\mathcal{U}}
\newcommand{\dVals}{\mathcal{D}}

\newcommand{\uSig}{\mathrm{u}}
\newcommand{\dSig}{\mathrm{d}}

\newcommand{\uSigs}{\mathrm{U}}
\newcommand{\dSigs}{\mathrm{D}}

\newcommand{\dStrat}{\delta}
\newcommand{\dStrats}{\Delta}

\newcommand{\xSig}{\mathrm{x}}
\newcommand{\xSigs}{\mathrm{X}}

\newcommand{\traj}[2]{\xSig_{#1}^{#2}}
\newcommand{\naughtTraj}{\traj{x_0,t_0}{\uSig,\dSig}}
\newcommand{\naughtTrajDerivative}{\dot{\xSig}_{x_0,t_0}^{\uSig,\dSig}}
\newcommand{\standardTraj}{\traj{x,t}{\uSig,\dSig}}
\newcommand{\stratTraj}{\traj{x,t}{\uSig,\dStrat(\uSig)}}

\newcommand{\evalTimes}{\mathcal{T}}
\newcommand{\rob}{J}

\newcommand{\supu}{\sup_{\uSig \in \uSigs}}
\newcommand{\supd}{\sup_{\dSig \in \dSigs}}
\newcommand{\infd}{\inf_{\dStrat \in \dStrats}}

\newcommand{\val}{V}

\newcommand{\tarFn}{r}
\newcommand{\conFn}{q}
\newcommand{\terFn}{\ell}

\newcommand{\tarSet}{\mathcal{R}}
\newcommand{\conSet}{\mathcal{Q}}
\newcommand{\terSet}{\mathcal{L}}

\newcommand{\tarEval}[2]{\tarFn \lf( #1(#2), #2 \rg)}
\newcommand{\sconeval}[2]{\conFn \lf( #1(#2), #2 \rg)}
\newcommand{\terEval}[2]{\terFn \lf( #1(#2)\rg)}

\newcommand{\graData}{\tarFn, \conFn, \terFn; \tFin}

\newcommand{\robGRA}{\rob_\mathrm{GRA}}
\newcommand{\robRS}{\rob_\mathrm{T}}
\newcommand{\robRT}{\rob_\mathrm{R}}
\newcommand{\robAT}{\rob_\mathrm{A}}
\newcommand{\robRSAT}{\rob_\mathrm{CT}}
\newcommand{\robRTAT}{\rob_\mathrm{RA}}

\newcommand{\valGRA}{\val_\mathrm{GRA}}

\newcommand{\valAT}{\val_\mathrm{A}}
\newcommand{\valRSAT}{\val_\mathrm{CT}}
\newcommand{\valRTAT}{\val_\mathrm{RA}}

\newcommand{\hjDep}{v}
\newcommand{\hjInd}{z}
\newcommand{\hjSol}{v}
\newcommand{\hjFn}{F}
\newcommand{\hjData}{g}
\newcommand{\hjDomain}{\Omega}
\newcommand{\hjTest}{\varphi}
\newcommand{\costate}{\lambda}
\newcommand{\tcostate}{p}
\newcommand{\hamiltonian}{H}
\newcommand{\hjInit}{S}
\newcommand{\hjFin}{T}

\newcommand{\ppt}{\partial_\tVal}
\newcommand{\pps}{\partial_\tDum}
\newcommand{\compact}{\mathcal{K}}
\newcommand{\RleT}{\R_{\le \tFin}}
\newcommand{\RleS}{\R_{\le \tFinDum}}

\newcommand{\xDum}{y}
\newcommand{\xDumDum}{z}
\newcommand{\coDum}{\mu}

\newcommand{\blackbox}{\hfill\rule{6pt}{6pt}}

\newcommand{\homeo}{\sigma} 
\DeclareMathOperator{\lip}{Lip}

\newcommand{\cmark}{\textcolor{green!60!black}{\ding{51}}} 
\newcommand{\xmark}{\textcolor{red}{\ding{55}}}             

\newcommand{\comment}{$//~$}

\title{Extending and Unifying the Fundamental Tasks of Hamilton-Jacobi Reachability Analysis
\author{Dylan Hirsch, William Sharpless, Donggun Lee, and Sylvia Herbert}
\thanks{Research reported in this publication was supported by the National Institutes of Health under award number T32EB009380. The content is solely the responsibility of the authors and does not necessarily represent the official views of the National Institutes of Health.}
\thanks{Dylan Hirsch (corresponding author), William Sharpless, and Sylvia Herbert are with the Department of Mechanical and Aerospace Engineering, University of California, San Diego, 9500 Gilman Drive, La Jolla, CA 92093.
Donggun Lee is with the Department of Mechanical and Aerospace Engineering, North Carolina State University, 1840 Entrepreneur Drive, Raleigh, NC 27695.
        {\tt\small dhirsch@ucsd.edu,
        whsarpless@ucsd.edu,
        dlee48@ncsu.edu,
        sherbert@ucsd.edu.}}%
}

\begin{document}
\maketitle

\begin{abstract}
In this work, we introduce the generalized reach-avoid (GRA) task, which both extends and unifies the canonical tasks of Hamilton-Jacobi Reachability (HJR).
We show that the GRA not only serves as a common primitive in this class of fundamental tasks, but also strictly extends the fundamental tasks that can be solved with HJR.
Moreover, the GRA formulation enables one to compute the value functions of certain composite tasks, including ones from timed temporal logic, by decomposing the value function of the composite task into value functions of GRA tasks.
We additionally show that the GRA is also a natural primitive to consider from a PDE perspective, as it can be used to represent all sufficiently regular solutions of the HJ-PDE that is canonical to HJR.
Collectively, the results in this work show the theoretical and practical utility of this task within the increasingly important framework of HJR.
\end{abstract}

\section{Introduction}
Hamilton-Jacobi reachability analysis (HJR) is a mathematical and algorithmic framework that is central to safety-critical control \cite{HJR-Survey}.
The theory behind this framework is rooted in the classical theory of differential games \cite{Isaacs-Differential-Games,Friedman-Differential-Games,evans-1984}, but the standard integrated-running-cost functional is replaced in HJR by performance functionals that are natural for obstacle-avoidance and target-reaching tasks.
Building on the canonical HJR works \cite{mitchell-2005,lygeros-2011,fisac-chen-2015}, recent works have elucidated deep connections between HJR and control barrier functions \cite{choi-2021,viscosity-cbfs}, control Lyapunov functions \cite{reach-and-stabilize-avoid},
state-constrained optimal control \cite{Bokanowski_2010,Altarovici_2012,Lee_2020,Lee_2023,efficient-state-constrained,Gammoudi_2023}, and safe reinforcement learning \cite{Reachability-RL,hsu-2021,so-2024, Ganai_2023,Ganai_2024,sharpless2026dual,sharpless2026bellman}.
Moreover, this framework has been applied to solve problems in transportation \cite{Jiang-2020,Jiang-2024}, aerospace \cite{Chen_2018}, and robotics \cite{Fisac_robotics,Pandya_2025,FasTrack,pmlr-v305-tonkens25a}.

In the HJR framework, one first chooses a performance functional that scores how well a given trajectory satisfies some task (e.g. distance to a goal).
One next computes the value function associated with this performance functional under the best-case control input and worst-case disturbance strategy.
The sign of this value function at a given initial state and time indicates whether one can control the system to satisfy the task despite the worst-case disturbance.
Finally, one synthesizes a closed-loop controller using this value function.

\begin{figure}[t]
    \centering
    \includegraphics[trim={17cm 0cm 14.75cm 0cm}, clip, width=1\linewidth]{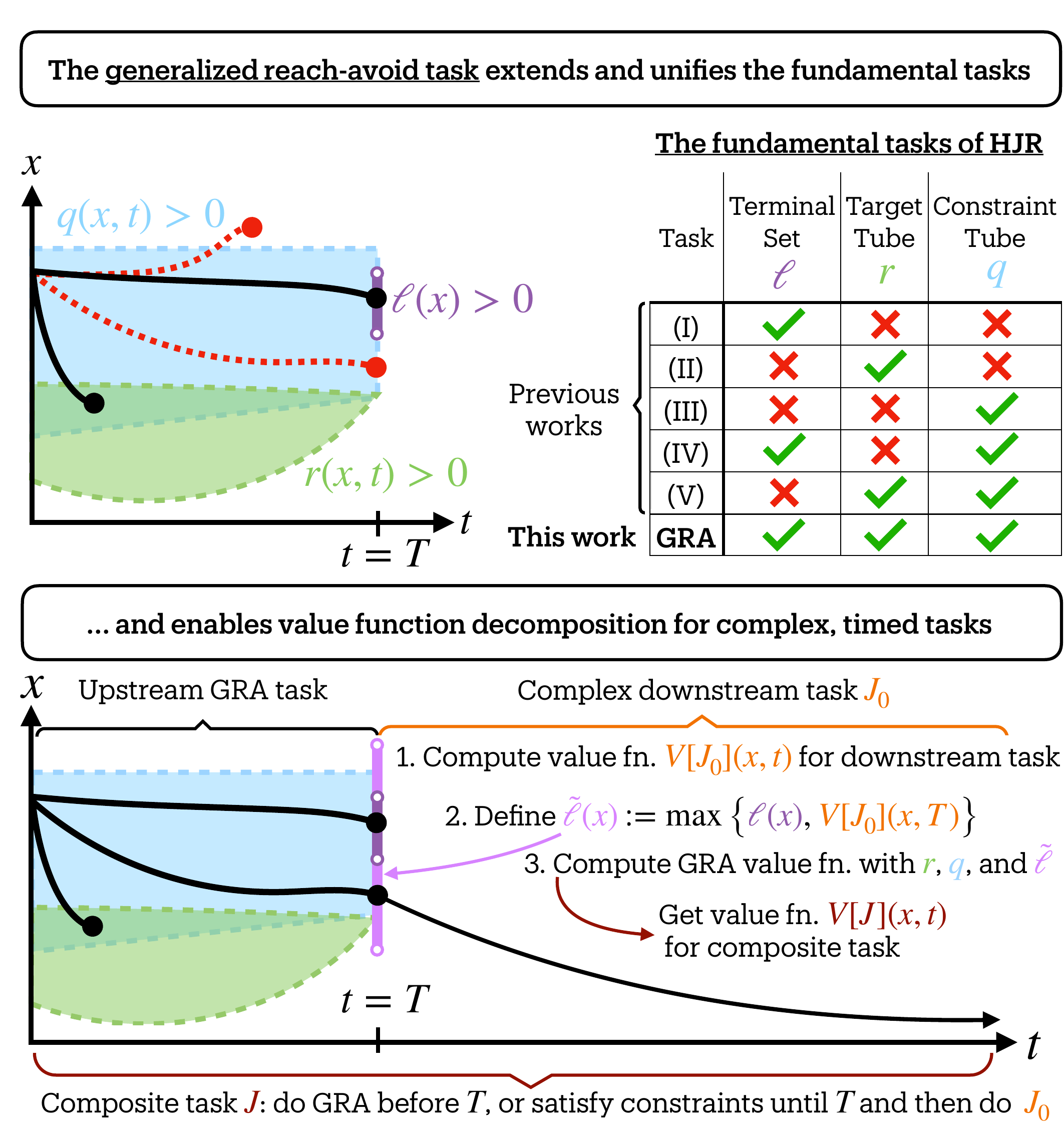}
    \caption{Graphical abstract.
    We introduce the generalized reach avoid task, a fundamental task that requires a system to either be in a terminal set (purple) \textit{at} the final time or reach a target tube (green) \textit{by} the final time, while remaining in a constraint tube (light blue) \textit{until} achieving either objective.
    Satisfactory trajectories are represented by solid, black curves and unsatisfactory ones by dotted, red curves. 
    Unlike the tubes, which are open regions in space-time, the terminal set exists only at the task's terminal time.
    By simultaneously considering these three elements, we extend the fundamental tasks that can be addressed with Hamilton-Jacobi reachability, unify the previous fundamental tasks into a single primitive, and provide new approaches to compute value functions for composite tasks.}
    \label{fig:graphical-abstract}
\end{figure}

Together with the classical theory of terminal-time tasks \cite{evans-1984}, the original HJR works \cite{mitchell-2005,lygeros-2011,fisac-chen-2015} established five canonical finite-horizon tasks that could be analyzed via the aforementioned pipeline (Fig. \ref{fig:graphical-abstract}, top):
\begin{enumerate}
    \item [(I)] \textbf{Terminal-set task:} be in a terminal set (e.g. a rendezvous area) \textit{at} a final time \cite{evans-1984},
    \item[(II)] \textbf{Reach task:} reach a target tube (e.g. a moving goal) \textit{by} a final time \cite{mitchell-2005},
    \item[(III)] \textbf{Avoid task:} remain within a constraint tube (e.g. the complement of a moving obstacle) \textit{until} a final time \cite{mitchell-2005},
    \item[(IV)] \textbf{Constrained terminal-set task:} be in a terminal set \textit{at} a final time, while remaining within a constraint tube \textit{until} that time \cite{lygeros-2011},
    \item[(V)] \textbf{Reach-avoid task:} reach a target tube \textit{by} a final time, while remaining within a constraint tube up \textit{until} reaching the target \cite{lygeros-2011,fisac-chen-2015}.
\end{enumerate}

In terms of their mathematical formulation within the HJR framework, no one of these five tasks can be considered most general.
For example, neither the terminal-set task (I) nor the avoid task (III) can be formulated as a reach-avoid task (V).

A major focus of recent works involves computing the value functions for complex tasks, such as those specified using temporal logic \cite{stl-meets-reachability,sharpless2026dual,sharpless2026bellman,Jiang-2020,Jiang-2024,Xiang-CDC-2025}.
To do so, dynamic programming approaches can be used to decompose the value function for the complex task into value functions for these canonical tasks \cite{sharpless2026dual,sharpless2026bellman,Xiang-CDC-2025,hirsch-2206}.
Computationally, to solve the value function for the complex task, one then iteratively solves the canonical tasks in a manner determined by the structure of the decomposition.
We can thus think of the canonical tasks as a set of fundamental (a.k.a. ``primitive'', ``atomic'') tasks and the complex tasks as composite tasks.
In particular, the fundamental tasks whose value function one can compute also determine the composite tasks one can analyze.

For the five canonical tasks listed above, each value function can be computed by solving an associated Hamilton-Jacobi terminal value problem (HJ-TVP), which consists of a Hamilton-Jacobi partial differential equation (HJ-PDE) together with boundary conditions at the task's terminal time.
In particular, the HJ-PDEs for the canonical tasks can all be written in the same form, i.e. with the same Hamiltonian, domain, and variational structure.
It is thus natural to ask the following questions:

\begin{itemize}
    \item Are there additional fundamental tasks whose value functions can be computed using HJR's canonical HJ-PDE?
    \item If so, do these new fundamental tasks enable one to handle new composite tasks?
    \item In terms of the underlying HJ-PDE, does this generalization ``complete'' the previous tasks in some mathematical sense?
\end{itemize}

In this work, we show the answer to all of these questions is yes.
In particular, we introduce the generalized reach-avoid (GRA) task, which requires either being in a terminal set \textit{at} a final time, or reaching a target tube \textit{by} this final time, while remaining in a constraint tube \textit{until} doing so (Fig. \ref{fig:graphical-abstract}, top).

We show that the value function for the GRA task can be solved via an HJ-TVP that strictly generalizes the HJ-TVPs of the canonical tasks.
Specifically, the underlying HJ-PDE remains unchanged, but the terminal condition is relaxed.
Through this relaxation, the GRA formulation not only extends the five canonical tasks in HJR, but in fact unifies them within a single HJ-TVP.

With regard to the second question, in addition to handling new fundamental tasks, one can use the GRA formulation to compute the value functions for certain composite tasks that depend on specific times, as in timed temporal logic, or for systems with dynamics that are piecewise-continuous in time.
These utilities are enabled by a key property of the GRA, a particular value function decomposition that specifically leverages the terminal set (Fig. \ref{fig:graphical-abstract}, bottom).

As for the third question, we show that every sufficiently regular solution of HJR's canonical HJ-PDE can be represented by the value function for a GRA task with an appropriate terminal set function.
(By contrast, the value function for e.g. a reach-avoid task (V) can represent only one solution to this PDE.)

An auxiliary utility of this work is that we establish the key theoretical results of HJR under relatively weak assumptions.
Previous works established the theory for the canonical tasks separately under different technical assumptions.
For example, to our knowledge, problem (IV) has only been formally studied within the context of static constraints \cite{lygeros-2011}, and problem (V) has only been formally studied for bounded, uniformly continuous target and obstacle functions when the target is time-varying \cite{fisac-chen-2015}.
Through unifying the tasks via the GRA, we thus establish general, uniform results for the five standard tasks as well.
In particular, we consider dynamic targets and constraints, and we only assume continuity of their representative functions.
~~\\
\\
\noindent\textbf{Primary contributions:}
\begin{itemize}
    \item We introduce the GRA task, which both extends and unifies the five canonical fundamental tasks of HJR.
    \item We establish a value-function decomposition (Theorem \ref{thm:value-decomposition}) for a certain composite task involving an upstream GRA task and a separate downstream task, which need not itself be a GRA task (Fig. \ref{fig:graphical-abstract}, bottom).
    We show that the value function of this composite task can be represented exactly as a single GRA value function by modifying the terminal set function of the upstream GRA task.
    \item We characterize the GRA value function as the unique continuous viscosity solution of an HJ-TVP (Theorem \ref{thm:dpe}).
    We then show a critical theoretical property of the GRA value function, namely that every sufficiently regular solution of the canonical HJ-PDE in HJR can be represented as the value function of a GRA task using an appropriate terminal-set function (Corollary \ref{cor:parameterization}).
    \item Through this unification of the previous tasks, we establish the key theory for fundamental tasks in HJR under a uniform set of weak technical assumptions.
    \item We show examples of how our results can be used to compute the value functions for a range of new problems.
    This includes a ``reach-avoid or always-avoid'' problem, a problem in timed temporal logic, and a problem where the system dynamics are piecewise-continuous.
\end{itemize}
~~\\
\noindent\textbf{Content of each section:}
It will first be convenient to abstract HJR's usual pipeline for constructing value functions.
This abstraction will allow us to simultaneously consider the five canonical tasks of HJR, the GRA task, and composite tasks that can be built from these fundamental tasks.
We develop this pipeline in Section \ref{sec:background}.
We then formally define the GRA task in Section \ref{sec:gra} and here connect it to the canonical fundamental tasks.
The main theoretical results are in Section \ref{sec:results}.
Finally, we explore three examples that demonstrate the utility of the GRA formulation in Section \ref{sec:examples}.
~\\
~\\
\noindent\textbf{Notation:}
We denote the extended real line by $\extended := [-\ii,\ii]$, endowed with its usual topology.
We say a map from some set into $\extended$ is real-valued if it nowhere attains $\pm \ii$.
Given $g: A \to Y$ and $B \sbs A$, we let $g|_B: B \to Y$ be the restriction of $g$ to $B$.
Given a continuous function $w: C \sbs \Rm \to \R$, we say that $w$ continuously extends to $D \sbs \partial C$ if there exists a continuous $\bar{w}: C \cup D \to \R$ such that $\bar{w}|_C = w$ (note that the extension must be real-valued).
For each $\tFin \in \R$ we let $\RleT := (-\ii,\tFin]$.
Given $z \in \Rm$ and $R > 0$, we let $B(z; R) \sbs \Rm$ be the open ball of radius $R$ with center $z$.
Given a map $h: \Rm \to \extended$, we define the strict zero super-level set of $h$ to be $\{z \in \Rm \mid h(z) > 0\}$.

\section{Setup and Background}\label{sec:background}
\subsection{Dynamics}
We consider a system with dynamics
\begin{equation}\label{eqn:dynamics}
\dot{\xSig}(\tVal) = \dynamics(\xSig(\tVal), \uSig(\tVal), \dSig(\tVal), \tVal),
\end{equation}
where $\dynamics:\Rn \tms \uVals \tms \dVals \tms \R \to \Rn$, with $\uVals \ssbs \Ru$ and $\dVals \ssbs \Rd$.
Here, $\xSig$ is the state, $\uSig$ is the control, and $\dSig$ is the disturbance.
We will assume the following throughout the sequel, which guarantee global existence and uniqueness of trajectories for this system.

\begin{assumption}\label{assumption:compactness}
    The sets $\uVals$ and $\dVals$ are compact.
\end{assumption}
\begin{assumption}\label{assumption:regularity}
~
\begin{itemize}
    \item For all $\tVal \in \R$, the map $(\xVal, \uVal, \dVal) \mapsto \dynamics(\xVal, \uVal, \dVal, \tVal)$ is continuous.
    \item For all $\xVal \in \Rn$, $\uVal \in \uVals$, and $\dVal \in \dVals$, the map $\tVal \mapsto \dynamics(\xVal, \uVal, \dVal, \tVal)$ is measurable.
    \item There exists some $K > 0$ such that $\|\dynamics(\xVal,\uVal,\dVal,\tVal)\| \le K(1 + \|\xVal\|)$ for all $\uVal \in \uVals$, $\dVal \in \dVals$, and $\tVal \in \R$.
    \item For each compact set $\compact \ssbs \Rn$ there exists an $L > 0$ such that $\|\dynamics(\xVal_1,\uVal,\dVal,\tVal) - \dynamics(\xVal_2,\uVal,\dVal,\tVal)\| \le L \|\xVal_1 - \xVal_2\|$ for all $\xVal_1,\xVal_2 \in \compact$, $\uVal \in \uVals$, $\dVal \in \dVals$, and $\tVal \in \R$.
\end{itemize}
\end{assumption}

\subsection{Signals and trajectories}
We denote by $\uSigs$ the set of all measurable control signals $\uSig: \R \to \uVals$ and by $\dSigs$ the set of all measurable disturbance signals $\dSig: \R \to \dVals$.
We now precisely define the state trajectory $\naughtTraj: \R \to \Rn$ that results from a control signal $\uSig$ and disturbance signal $\dSig$, given that the system is in state $\xNot$ at time $\tNot$.

For each $\xNot \in \Rn$, $\tNot \in \R$, $\uSig \in \uSigs$, and $\dSig \in \dSigs$,
we let $\naughtTraj$ be the Carath\'{e}odory solution of \eqref{eqn:dynamics} under the condition $\xSig(\tNot) = \xNot$.
More explicitly, $\naughtTraj$ is defined to be the unique locally absolutely continuous map from $\R$ to $\Rn$ for which $\naughtTraj(\tNot) = \xNot$ and $\naughtTrajDerivative(\tVal) = \dynamics( \naughtTraj(\tVal), \uSig(\tVal), \dSig(\tVal), \tVal)$ for a.e. $\tVal \in \R$ (existence and uniqueness of this solution follow from Assumptions \ref{assumption:compactness} and \ref{assumption:regularity}; see Theorem 1.2.1 in \cite{Friedman-Differential-Games}). 

\subsection{Performance functionals}

We will denote by $\xSigs$ the set of all plausible trajectories, i.e. continuous maps $\xSig: \R \to \Rn$.
We endow $\xSigs$ with the topology of uniform convergence on compact sets.
In this topology, a sequence $(\xSig_{i})_{i \in \N}$ in $\xSigs$ converges to $\xSig \in \xSigs$ iff $\max_{\tVal \in \compact} \| \xSig_i(\tVal) - \xSig(\tVal) \| \to 0$ for all non-empty, compact $\compact \ssbs \R$.

We can think of any map $\rob: \xSigs \tms \evalTimes \to \extended$, where $\evalTimes \sbs \R$, as a performance functional for some task, where $\rob$ assigns a score $\rob(\xSig, \tVal)$ to a trajectory $\xSig$ based on the time $\tVal$ at which the task begins.
Here, $\evalTimes$ represents the set of times at which it is reasonable for the task to begin (e.g. for the task ``remain inside until noon today," we may define $\evalTimes$ to be times no later than noon).

In HJR, the performance functional $\rob$ for each task is chosen such that $\rob(\xSig, \tVal) > 0$ iff the trajectory $\xSig$ satisfies the task starting at time $\tVal$.\footnote{More generally, the robustness metric for any specification in temporal logic is a performance functional whose sign corresponds to qualitative satisfaction of the specification \cite{donze-robustness-metric}. Examples of the robustness metric for various task specifications will be provided in Section \ref{sec:examples}.}

\subsection{Differential games and value functions}
We now consider a differential game in which the controller player would like to maximize a given performance functional and the disturbance player would like to minimize it.

We first specify the information pattern for the game.
Conceptually, we can think of any map $\dStrat:\uSigs \to \dSigs$ as representing a strategy for a disturbance player that selects the disturbance input $\dSig = \dStrat(\uSig)$ based upon the control input $\uSig$.
In HJR, we make the restriction that the disturbance's strategy cannot use future information to inform current decisions via the following definition.

\begin{definition}[Non-Anticipativity]
The map $\dStrat :\uSigs \to \dSigs$ is \textbf{non-anticipative} if for all $\tVal \in \R$ and $\uSig_1,\uSig_2 \in \uSigs$ such that $\uSig_1(\tDum) = \uSig_2(\tDum)$ for a.e. $\tDum \le \tVal$, we also have $\dStrat[\uSig_1](\tDum) = \dStrat[\uSig_2](\tDum)$ for a.e. $\tDum \le \tVal$.
\end{definition}
\noindent We denote by $\dStrats$ the set of all non-anticipative $\dStrat: \uSigs \to \dSigs$.

Adapting the notation used in \cite{sharpless2026bellman}, given a performance functional $\rob: \xSigs \tms \evalTimes \to \extended$, we define the value function $\val[\rob]: \Rn \tms \evalTimes \to \extended$ by
\begin{equation}\label{eqn:value-function-definition}
    \val[\rob](\xVal, \tVal) = \infd \supu \rob(\stratTraj, \tVal).
\end{equation}

The value function $\val[\rob](\xVal, \tVal)$ conceptually represents the performance score that will result from both players acting optimally under the non-anticipative information structure, given that the system is initialized in state $\xVal$ at time $\tVal$.
We have the following useful fact.
\begin{lemma} \label{lem:continuous-evaluations-have-continuous-values}
Let $\rob: \xSigs \tms \evalTimes \to \extended$ for some $\evalTimes \sbs \R$.
If $\rob$ is continuous, then $\val[\rob]$ is also continuous.
If we in addition have $\rob < \ii$ (resp. $\rob > -\ii$), then $\val[\rob] < \ii$ (resp. $\val[\rob] > -\ii$). 
\begin{proof}
See Section \ref{sec:appendix-basic-proofs} in the appendix.
\end{proof}

\end{lemma}

In HJR, the sign of the value $\val[\rob](\xVal, \tVal)$ is used to determine whether a task can be satisfied starting from the state $\xVal$ at time $\tVal$.
Formally, we have the following result:
\begin{lemma}\label{lem:value-and-satisfiability}
Let $\evalTimes \sbs \R$, $\rob: \xSigs \tms \evalTimes \to \extended$, $\xVal \in \Rn$, and $\tVal \in \evalTimes$.
If $\val[\rob](\xVal, \tVal) > 0$, then for each $\dStrat \in \dStrats$ there is a $\uSig \in \uSigs$ such that $\rob(\stratTraj, \tVal) > 0$.
Similarly, if $\val[\rob](\xVal, \tVal) < 0$, then there is some $\dStrat \in \dStrats$ such that $\rob(\stratTraj, \tVal) < 0$ for all $\uSig \in \uSigs$.
\end{lemma}
\begin{proof}
    Follows directly from \eqref{eqn:value-function-definition}. 
\end{proof}

\section{The generalized reach-avoid task}\label{sec:gra}

\subsection{The GRA performance functional and its underlying task}

\begin{definition}
    A \textbf{GRA instance} is a tuple $(\graData)$, where
    \begin{itemize}
        \item $\tarFn: \Rn \tms \R \to [-\ii,\ii)$ is continuous,
        \item $\conFn: \Rn \tms \R \to (-\ii,\ii]$ is continuous,
        \item $\terFn: \Rn \to \R$ is
        continuous,
        \item $\tFin \in \R$.
    \end{itemize}
\end{definition}
In this GRA instance, we refer to $\tarFn$ as the target tube function, $\conFn$ as the constraint tube function, $\terFn$ as the terminal set function, and $\tFin$ as the terminal time of the task.
A GRA instance is visualized in the top panel of Fig. $\ref{fig:graphical-abstract}$.

For each GRA instance $(\graData)$, we define the GRA performance functional $\robGRA[\graData]: \xSigs \tms \RleT \to \R$ by
\begin{align}\label{eqn:gra-definition}
    \robGRA[\graData](\xSig, & \tVal) :=  \nonumber \\
    \max\biggl\{
    \max_{\tDum \in [\tVal, \tFin]} &\min\Bigl\{ \tarEval{\xSig}{\tDum}, \min_{\tDumDum \in [\tVal, \tDum]} \sconeval{\xSig}{\tDumDum} \Bigr\}, \nonumber\\
    &\min\Bigl\{ \terEval{\xSig}{\tFin}, \min_{\tDumDum \in [\tVal, \tFin]} \sconeval{\xSig}{\tDumDum} \Bigr\}
    \biggr\}.
\end{align}

Let us now make explicit the qualitative task that conceptually underlies this performance functional.
Consider a target tube $\tarSet \sbs \Rn \tms \R$, a constraint tube $\conSet \sbs \Rn \tms \R$, a terminal set $\terSet \sbs \Rn$, and a terminal time $\tFin \in \R$.\footnote{We use the term \textit{tube} to emphasize that $\tarSet$ and $\conSet$ are open subsets of space-time $(\Rn \tms \R)$, which in practice allows one to represent moving targets or obstacles.
By contrast, the terminal \textit{set} $\terSet$ is an open subset of only space $(\Rn)$.
Conceptually, it is useful to think of $\terSet$ as being embedded within the slice of space-time at time $\tFin$, i.e. $\Rn \tms \{\tFin\}$, as visualized in Figure \ref{fig:graphical-abstract}.}
We assume $\tarSet$, $\conSet$, and $\terSet$ are all open sets.
The GRA task consists of either reaching $\tarSet$ \textit{by} time $\tFin$ while remaining in $\conSet$ \textit{until} the reach time, or being in $\terSet$ \textit{at} time $\tFin$ while remaining in $\conSet$ \textit{until} then.
More formally:
\begin{definition}
Let $\tarSet \sbs \Rn \tms \R$, $\conSet \sbs \Rn \tms \R$, and $\terSet \sbs \Rn$ all be open, and let $\tFin \in \R$.
A pair $(\xSig, \tVal) \in \xSigs \tms \RleT$ \textbf{satisfies the GRA task} w.r.t. $\tarSet$, $\conSet$, $\terSet$, and $\tFin$ if either of the two following conditions holds:
\begin{itemize}
    \item[(\textit{i})] there is some $\tDum \in [\tVal, \tFin]$ such that $(\xSig(\tDum), \tDum) \in \tarSet$ and $(\xSig(\tDumDum), \tDumDum) \in \conSet$ for all $\tDumDum \in [\tVal, \tDum]$,
    \item[(\textit{ii})] $\xSig(\tFin) \in \terSet$ and $(\xSig(\tDumDum), \tDumDum) \in \conSet$ for all $\tDumDum \in [\tVal, \tFin]$.
\end{itemize}
\end{definition}

\begin{remark}\label{rem:why-care?}
In the above definition, (\textit{i}) is known as a reach-avoid task, and
(\textit{ii}) is known as a constrained terminal-set task.
Thus the GRA combines both of these tasks into a single task.
It is thus natural to ask: \textit{why would we not simply treat the two parts of the GRA task separately?}
This question will be addressed in Section \ref{sec:necessity-of-the-gra}. \blackbox
\end{remark}

The following observation makes explicit the natural correspondence between the GRA performance functional and the GRA task: 
\begin{observation}
    Let $(\graData)$ be a GRA instance.
    Define $\tarSet$, $\conSet$, and $\terSet$ to be the strict zero super-level sets of $\tarFn$, $\conFn$, and $\terFn$, respectively.
    A pair $(\xSig, \tVal) \in \xSigs \tms \RleT$ satisfies the GRA task w.r.t. $\tarSet$, $\conSet$, $\terSet$, and $\tFin$ iff $\robGRA[\graData](\xSig, \tVal) > 0$.
\end{observation}

For convenience, we will define 
$$\valGRA[\graData] := \val\lf[ \robGRA[\graData] \rg]$$
for each GRA instance $(\graData)$.
Recall that if one can compute this value function, then Lemma \ref{lem:value-and-satisfiability} can be used to determine whether the GRA task is satisfiable from an initial state $\xVal$ at time $\tVal$ based upon the sign of $\valGRA[\graData](\xVal, \tVal)$ (with positive signifying satisfiability and negative signifying the opposite).

\subsection{Relationship with the 5 canonical problems}
In this sub-section, we will fix a GRA instance $(\graData)$ and shall omit the bracket notation that is used elsewhere.

We now define the performance functionals $\robRS, \robRT, \robAT, \robRSAT, \robRTAT: \xSigs \tms \RleT \to \R$ for tasks (I)-(V) described in the introduction:
\begin{enumerate}
    \item [(I)] \textbf{Terminal-set task:}
    $$\robRS(\xSig, \tVal) := \terEval{\xSig}{\tFin},$$
    \item[(II)] \textbf{Reach task:}
    $$\robRT(\xSig, \tVal) := \max_{\tDum \in [\tVal, \tFin]} \tarEval{\xSig}{\tDum},$$
    \item[(III)] \textbf{Avoid task:}
    $$\robAT(\xSig, \tVal) := \min_{\tDumDum \in [\tVal, \tFin]} \sconeval{\xSig}{\tDumDum},$$
    \item[(IV)] \textbf{Constrained terminal-set task:}
    $$\robRSAT(\xSig, \tVal) := \min\Bigl\{ \terEval{\xSig}{\tFin}, \min_{\tDumDum \in [\tVal, \tFin]} \sconeval{\xSig}{\tDumDum} \Bigr\},$$
    \item[(V)] \textbf{Reach-avoid task:}
    $$\robRTAT(\xSig, \tVal) := \max_{\tDum \in [\tVal, \tFin]} \min\Bigl\{ \tarEval{\xSig}{\tDum}, \min_{\tDumDum \in [\tVal, \tDum]} \sconeval{\xSig}{\tDumDum} \Bigr\}.$$
\end{enumerate}

These performance functionals, and thus also their value functions, are related to those of the GRA as follows:
\begin{observation}\label{lem:equivalence}
~
    \begin{itemize}
        \item[(\textit{i})] If $\tarFn = -\ii$ and $\conFn = \ii$, then $\robGRA = \robRS$.
        \item[(\textit{ii})] If $\terFn(\cdot) \le \tarFn(\cdot, \tFin)$ and $\conFn = \ii$, then $\robGRA = \robRT$.
        \item[(\textit{iii})] If $\terFn(\cdot) \ge \conFn(\cdot, \tFin)$ and $\tarFn = -\ii$, then $\robGRA = \robAT$.
        \item[(\textit{iv})] If $\tarFn = -\ii$, then $\robGRA = \robRSAT$.
        \item[(\textit{v})] If $\terFn(\cdot) \le \tarFn(\cdot, \tFin)$, then $\robGRA = \robRTAT$.
    \end{itemize}
\end{observation}
In other words, we can think of the other fundamental tasks of HJR and their value functions as specific cases of the GRA task and its value function.
We summarize the relationship between the tasks in Table \ref{tab:hierarchy}.
\begin{table}[ht]
\centering
\begin{tabular}{>{\bfseries}r l l l }
\toprule
 & Terminal set $(\terFn)$ & Target tube $(\tarFn)$ & Constraint tube $(\conFn)$ \\
\midrule\midrule
(I) T & \cmark & \xmark~($\tarFn = -\infty$) & \xmark~ ($\conFn = +\infty$) \\
\midrule
(II) R & \xmark~($\terFn(\cdot) \le \tarFn(\cdot, \tFin)$) & \cmark &  \xmark~ ($\conFn = +\infty$) \\
\midrule
(III) A & \xmark~($\terFn(\cdot) \ge \conFn(\cdot, \tFin)$) & \xmark~($\tarFn = -\infty$) & \cmark \\
\midrule
(IV) CT & \cmark & \xmark~($\tarFn = -\infty$) & \cmark \\
\midrule
(V) RA & \xmark~($\terFn(\cdot) \le \tarFn(\cdot, \tFin)$) & \cmark & \cmark \\
\midrule
GRA & \cmark & \cmark & \cmark \\
\bottomrule
\end{tabular}
\caption{Summary of the fundamental problems in HJR, consisting of the 5 canonical tasks and the new task (GRA) introduced in this work.
In parentheses, we specify how to recover the previous tasks from the GRA (see Observation \ref{lem:equivalence}).}
\label{tab:hierarchy}
\end{table}

\subsection{Motivation for an HJ-TVP approach to compute the GRA value function}\label{sec:necessity-of-the-gra}

In Section \ref{sec:gra-tvp}, we will provide an HJ-TVP that can be used to compute the value function for a GRA task.
Before doing so, let us address why one cannot simply derive the GRA value function from the value functions for the canonical tasks.

As mentioned in Remark \ref{rem:why-care?}, the GRA task consists of either completing a reach-avoid task or a constrained terminal-set task.
Correspondingly, we have the identity
\begin{equation*}
    \robGRA[\graData] = \max\{\robRTAT[\tarFn, \conFn; \tFin], \robRSAT[\conFn, \terFn; \tFin]\}.
\end{equation*}
At first glance, it may seem that there is then a simple approach to compute the GRA value function: first compute the value function $\valRTAT[\tarFn, \conFn; \tFin]$ for the reach-avoid task, then do so for the constrained terminal-set task $\valRSAT[\conFn, \terFn; \tFin]$, and finally take the point-wise maximum $\max\{\valRTAT[\tarFn, \conFn; \tFin], \valRSAT[\conFn, \terFn; \tFin]\}$ of these two value functions.
However, the above algebraic relationship between the performance functionals does not extend to the value functions.
In particular, we only have the inequality 
    $$\valGRA[\graData] \ge \max\{\valRTAT[\tarFn, \conFn; \tFin], \valRSAT[\conFn, \terFn; \tFin]\}$$
(which follows from the standard minimax inequality).
The following counter-example, which is visualized in Fig. \ref{fig:remark-2-fig}, demonstrates that this inequality can indeed be strict.

\begin{example}\label{rem:counter-example}
Let $\dot{\xSig} = \uSig + \dSig$ be a system with a scalar state, where the control and disturbance sets are $\uVals = [0,1]$ and $\dVals = [-1,1]$, respectively.
Let the task's terminal time be $\tFin = 4$, the target tube be $\tarSet = (-1, +1) \tms \R$, the constraint tube be $\conSet = \R \tms \R$, and the terminal set be $\terSet = (-\ii, +1)$ (Fig. \ref{fig:remark-2-fig}).
From the initial state $\xVal_0 = -2$ and initial time $\tVal_0 = 0$, the disturbance player can ensure that the system never enters $\tarSet$ via $\dSig(\cdot) = -1$ (Fig. \ref{fig:remark-2-fig}, top-left), or it can ensure the system state is not in $\terSet$ at the final time $\tFin$ via $\dSig(\cdot) = 1$ (Fig. \ref{fig:remark-2-fig}, top-right).
Thus the controller cannot ensure successful completion of either the reach-avoid task or the constrained terminal-set task.
However, for any signals $\uSig$ and $\dSig$, the resulting trajectory will nevertheless satisfy the GRA task (Fig \ref{fig:remark-2-fig}, bottom-left).
In particular, if we let $\tarFn$, $\conFn$, and $\terFn$ be the signed-distance functions to $\tarSet$, $\conSet$, and $\terSet$ respectively, we have $\valGRA[\graData](\xVal_0, \tVal_0) > 0$, despite $\valRTAT[\tarFn, \conFn; \tFin](\xVal_0, \tVal_0) < 0$ and $\valRSAT[\conFn, \terFn](\xVal_0, \tVal_0) < 0$. \blackbox
\end{example}
\noindent In Example 1 in Section \ref{sec:example-1}, we will further emphasize this idea for a more complex problem, with numerical results.

\begin{figure}[!ht]
    \centering
    \includegraphics[trim={6.5cm 0cm 10cm 0cm}, clip, width=1.0 \linewidth]{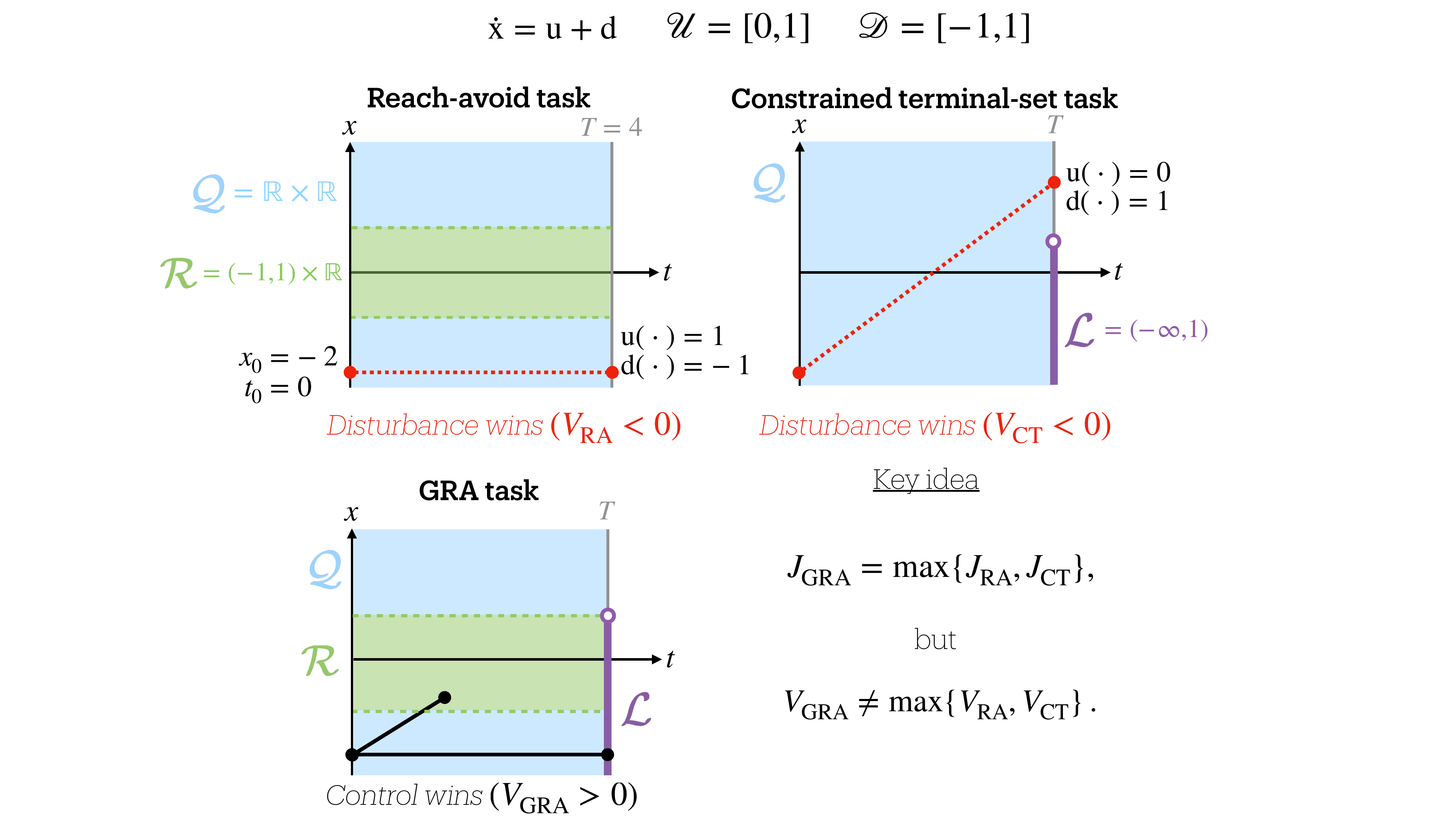}
    \caption{Visualization of Counter-Example \ref{rem:counter-example}.
    The GRA task consists of either completing an underlying reach-avoid task (top left) or an underlying constrained terminal-set task (top right).
    In the counter-example, the disturbance player can separately prevent either of these tasks from being completed.
    When the tasks are considered together, however, the controller wins (bottom left).
    As such, we need a method to compute the value function for the GRA task that is different from simply computing $\max\{\valRTAT, \valRSAT\}$.
    }
    \label{fig:remark-2-fig}
\end{figure}

To summarize, we cannot generally use the reach-avoid and constrained terminal-set value functions to compute the GRA value function.
Instead, a separate approach is needed.
Theorem \ref{thm:dpe} in Section \ref{sec:gra-tvp} will demonstrate how to compute the GRA value function by way of an HJ-TVP that generalizes the HJ-TVPs for the canonical tasks.

\section{Properties of the GRA value function}\label{sec:results}
\subsection{Basic properties}
We now explore basic properties of the GRA value function.

\begin{definition}
Let $\evalTimes \sbs \R$.
A performance functional $\rob: \xSigs \tms \evalTimes \to \extended$ is \textbf{past-independent} if for all $\tVal \in \evalTimes$ and all $\xSig_1, \xSig_2 \in \xSigs$, if $\xSig_1(\tDum) = \xSig_2(\tDum)$ for every $\tDum \ge \tVal$, then $\rob(\xSig_1, \tVal) = \rob(\xSig_2, \tVal)$.
\end{definition}

\begin{lemma}\label{lem:gra-evaluation-is-continuous}
    Let $(\graData)$ be a GRA instance.
    Then $\robGRA[\graData]$ is continuous, real-valued, and past-independent. 
    (By Lemma \ref{lem:continuous-evaluations-have-continuous-values}, $\valGRA[\graData]$ is also then continuous and real-valued.)
\end{lemma}
\begin{proof}
    See Section \ref{sec:appendix-basic-proofs} in the appendix.
\end{proof}

\begin{remark}
Consider two performance functionals $\rob_1$ and $\rob_2$ corresponding to two tasks, task 1 and task 2.
There is a standard algebra for constructing performance functionals for various composite tasks built from tasks 1 and 2 \cite{donze-robustness-metric}.
For example, the performance functional for satisfying task 1 \textit{and} task 2 is given by $\min\{\rob_1, \rob_2\}$, while satisfying task 1 \textit{or} task 2 is given by $\max\{\rob_1, \rob_2\}$.
Moreover, the performance functional for \textit{not} satisfying task 1 (resp. task 2) is given by $-\rob_1$ (resp. $-\rob_2$).
As in Counter-Example \ref{rem:counter-example}, none of these algebraic formulae apply to the value functions (e.g. $\val[-\rob_1] \ne -\val[\rob_1]$) due to the underlying differential game.
As such, other approaches must be used to compute these value functions.
The following lemma provides one such approach for the negation of a GRA task. \blackbox
\end{remark}

\begin{lemma}\label{lem:negated-gra}
    Let $(\graData)$ be a GRA instance.
    Define $\tilde{\tarFn} = -\conFn$, $\tilde{\conFn} = -\min\{\conFn, \tarFn\}$, and $\tilde{\terFn}(\cdot) = -\min\{\conFn(\cdot, \tFin), \max\{\tarFn(\cdot, \tFin), \terFn(\cdot) \}\}$.
    Then
    \begin{equation*}
        \val[-\robGRA[\graData]] = \valGRA[\tilde{\tarFn}, \tilde{\conFn}, \tilde{\terFn}; \tFin].
    \end{equation*}
\end{lemma}
\begin{proof}
    See Section \ref{sec:appendix-basic-proofs} in the appendix.
\end{proof}
This lemma signifies that to compute the value function associated with the negation of a GRA task (i.e. the controller attempts to prevent the GRA task from being satisfied), one can define new target tube, constraint tube, and terminal set functions using the old ones, and then solve for the GRA value function for this new GRA instance.
In other words, the GRA task is in some sense ``closed'' under negation.
Note from Table \ref{tab:hierarchy} that this is not the case for any of the previous fundamental tasks except for (I).

\subsection{Value function decomposition}

A primary advantage of considering the GRA task over the reach-avoid task is that this extension allows one to decompose value functions of certain composite tasks via a dynamic programming approach.
A key value function decomposition involving the GRA task is specified in the upcoming theorem, which is graphically depicted in the bottom panel of Figure \ref{fig:graphical-abstract}. 

Importantly, HJR decomposition results for the two player setting typically begin with a particular composite task and express its value function in terms of value functions for pre-defined fundamental tasks.
For example, in \cite{hirsch-2206} a ``reach-always-avoid'' task is decomposed into an upstream reach-avoid task and a downstream avoid task.

Here, we consider a broader kind of result: we define a composite task that involves an upstream GRA task and a downstream task that is arbitrary, aside from mild technical assumptions.
In particular, the downstream task need not be a GRA task.\footnote{For other decomposition results that are similarly algebraic in nature, but assume a \textit{one-player}, discrete-time setting, we refer the reader to \cite{sharpless2026bellman}.}

Specifically, we consider a composite task in which the system has two paths to success: either \textit{(i)} complete a GRA task before time $\tFin$ or \textit{(ii)} remain within the constraints until time $\tFin$ and then complete the downstream task.

The following theorem shows that the value function of this newly defined composite task can be represented exactly as a single GRA value function.
This result thus expands the class of tasks whose value functions can be computed using the GRA formulation.

In the following theorem, $\rob_0$ represents the performance functional for the downstream task and $\rob$ represents the performance functional for the composite task described above.
\begin{theorem}[Value decomposition]\label{thm:value-decomposition}
Let $\rob_0: \xSigs \tms \evalTimes \to \R$ be continuous and past-independent, with $\evalTimes \sbs \R$.
Let $(\tarFn, \conFn, \terFn; \tFin)$ be a GRA instance, with $\tFin \in \evalTimes$.
Define $\rob: \xSigs \tms \RleT \to \R$ by
\begin{align*}
    \rob(\xSig, \tVal) := 
    \max\Big\{&\robGRA[\tarFn, \conFn, \terFn; \tFin](\xSig, \tVal),\\
    &\min\big\{\rob_0(\xSig, \tFin), \min_{\tDumDum \in [\tVal, \tFin]} \conFn(\xSig(\tDumDum),\tDumDum) \big\} \Big\}.
\end{align*}
Then $\rob$ is continuous, past-independent and 
$$\val[\rob] = \valGRA[\tarFn, \conFn, \tilde{\terFn}; \tFin],$$
where $\tilde{\terFn}(\cdot) := \max\{\terFn(\cdot), \val[\rob_0](\cdot, \tFin)\}$.
\end{theorem}
\begin{proof}
    This proof proceeds by the standard dynamic programming argument, splitting the task into a portion before time $\tFin$ and a portion after.
    The full details are in Section \ref{sec:appendix-decomposition-proof} of the appendix.
\end{proof}

The significance of Theorem \ref{thm:value-decomposition} is as follows.
Computing the value function $\val[\rob]$ for the composite task can be done by: (1) computing the value function $\val[\rob_0]$ for the downstream task, (2) forming the new terminal set function $\tilde{\terFn}$, and (3) computing the GRA value function for the GRA instance $(\tarFn, \conFn, \tilde{\terFn}; \tFin)$.

\textit{Critically, recall that the reach (II), avoid (III), and reach-avoid (V) tasks do not consider a terminal set function $\terFn$ (see Table \ref{tab:hierarchy}).
However, this value function decomposition hinges upon the use of the terminal set function in step (3).
Conceptually, the terminal set function $\tilde{\terFn}$ here serves as an ``interface'' or ``port'' through which the downstream value function is incorporated into the upstream problem when doing dynamic programming.
It is precisely this interface within the GRA task that makes the decomposition in Theorem \ref{thm:value-decomposition} possible.}

As a corollary of the above theorem, we also get a dynamic programming principle for the GRA value function itself, in a surprisingly simple form:
\begin{corollary}[Dynamic Programming Principle]\label{cor:dpp}
Let $(\graData)$ be a GRA instance, and let $\tFinDum \le \tFin$.
Define $\tilde{\terFn} = \valGRA[\graData](\cdot, \tFinDum)$.
Then
$$\valGRA[\graData](\xVal, \tVal) = \valGRA[ \tarFn, \conFn, \tilde{\terFn}; \tFinDum](\xVal, \tVal)$$
for all $\xVal \in \Rn$ and $\tVal \le \tFinDum$.
\end{corollary}
\begin{proof}
    See Section \ref{sec:appendix-decomposition-proof} in the appendix.
\end{proof}

Thus the GRA value function is also in some sense ``closed'' under dynamic programming.
Note that this is not the case for tasks (II), (III), or (V), which lack a terminal set function.

\subsection{Background on viscosity solutions to HJ-TVPs}\label{sec:pde-background}
In HJR, the value functions for the fundamental tasks are each characterized as the unique viscosity solution of a corresponding HJ-TVP.
We will now make precise what we mean by a viscosity solution. 
Recall that an HJ-TVP is a HJ-PDE together with a terminal boundary condition, i.e. a problem of the form 
\begin{equation}\label{eqn:hj-tvp}\tag{HJ-TVP}
    \begin{cases}
    \hjFn\Big(\xVal, \tVal, \hjDep, \nabla_\xVal \hjDep, \ppt \hjDep\Big) = 0, & \xVal \in \Rn,~ \tVal \in (\hjInit, \hjFin), \\
    \hjDep(\xVal, \hjFin) = \hjData(\xVal), & \xVal \in \Rn.
    \end{cases}
\end{equation}
In this problem, $\xVal$ (a vector) and $\tVal$ (a scalar) are the independent variables, $\hjDep$ (a scalar) is the dependent variable, and $\Rn \tms (\hjInit, \hjFin]$ is the problem's domain.
Note that throughout the main text, the dependent variable in each HJ-PDE and HJ-TVP will always be $\hjSol$, and the independent variables will always be $\xVal$ and $\tVal$ (in the appendix $w$ is occasionally used for the dependent variable and $\xDumDum$ for the independent variable).

We define a viscosity sub-solution, super-solution, and solution to an HJ-PDE itself (i.e. with no terminal condition) as in Definitions II.1.1 and V.1.1 of \cite{Bardi-Dolcetta-Optimal-Control} (reproduced in Section \ref{sec:appendix-viscosity} of the appendix for convenience).
For an HJ-TVP, the analogous definitions are similar to those for HJ-PDEs, but also include conditions at the terminal boundary:

\begin{definition}
    Let $\hjInit, \hjFin \in \R$ with $\hjInit < \hjFin$.
    Let $\hjFn: \Rn \tms (\hjInit,\hjFin) \tms \R \tms \Rn \tms \R \to \R$ and $\hjData: \Rn \to \R$ both be continuous.
    A function $\hjSol: \Rn \tms (\hjInit, \hjFin] \to \R$ is a
    \begin{itemize}
        \item \textbf{viscosity sub-solution} of \eqref{eqn:hj-tvp} if
        \begin{enumerate}
        \item [(\textit{i})]
        $\hjSol$ is upper semi-continuous,
        \item[(\textit{ii})]
        for each $\xNot \in \Rn$, $\tNot \in (\tInit, \tFin)$, and continuously differentiable $\hjTest:\Rn \tms (\tInit, \tFin) \to \R$,
        if $\hjSol - \hjTest$ has a local maximum at $(\xNot, \tNot)$ then
        $$\textstyle \hjFn\Big(\xNot, \tNot, \hjSol(\xNot, \tNot), \nabla_\xVal \hjTest(\xNot, \tNot), \ppt \hjTest(\xNot, \tNot) \Big) \le 0,$$
        \item[(\textit{iii})] $\hjSol(\xVal, \tFin) \le \hjData(\xVal)$ for all $\xVal \in \Rn$;
        \end{enumerate}
        
        \item \textbf{viscosity super-solution} of \eqref{eqn:hj-tvp} if
        \begin{enumerate}
        \item [(\textit{i})]
        $\hjSol$ is lower semi-continuous,
        \item[(\textit{ii})]
        for each $\xNot \in \Rn$, $\tNot \in (\tInit, \tFin)$, and continuously differentiable $\hjTest:\Rn \tms (\tInit, \tFin) \to \R$,
        if $\hjSol - \hjTest$ has a local minimum at $(\xNot, \tNot)$ then
        $$\textstyle \hjFn\Big(\xNot, \tNot, \hjSol(\xNot, \tNot), \nabla_\xVal \hjTest(\xNot, \tNot), \ppt \hjTest(\xNot, \tNot) \Big) \ge 0,$$
        \item[(\textit{iii})] $\hjSol(\xVal, \tFin) \ge \hjData(\xVal)$ for all $\xVal \in \Rn$;
        \end{enumerate}

        \item \textbf{viscosity solution} of \eqref{eqn:hj-tvp} if 
        $\hjSol$ is both a viscosity sub-solution and super-solution of \eqref{eqn:hj-tvp}.
    \end{itemize}
\end{definition}

\subsection{Characterization of the GRA value function as the unique viscosity solution to an HJ-TVP}\label{sec:gra-tvp}
We now introduce the generalized reach-avoid terminal-value-problem (GRA-TVP), which takes the form
\begin{equation}\label{eqn:gra-tvp}\tag{GRA-TVP}
    \begin{cases}
     -\min\big\{ \conFn(\xVal,\tVal) - \hjDep,
     \max\big\{ \tarFn(\xVal,\tVal) - \hjDep, \\
     \qquad\quad \ppt \hjDep + \hamiltonian(\xVal, \tVal, \nabla_x \hjDep) \big\} \big\} = 0, & \xVal \in \Rn,\tVal \in (\tInit, \tFin), \\
     \hjDep(\xVal,\tFin) = \min\big\{\conFn(\xVal,\tFin), \\
     \qquad \qquad \quad \max\{ \tarFn(\xVal,\tFin), \terFn(\xVal) \}\big\}, & \xVal \in \Rn.
    \end{cases}
\end{equation}
\begin{remark}
    We note that the underlying HJ-PDE in the above HJ-TVP is the same as the HJ-PDE for the reach-avoid task (V) presented in \cite{fisac-chen-2015}.
    It is instead the boundary condition that has been generalized, now including the terminal set function $\terFn$.
    We will clarify at the end of this sub-section how the HJ-PDE for every canonical task can in fact be written in this same form. \blackbox
\end{remark}

We refer to the function $\hamiltonian$ in \eqref{eqn:gra-tvp} as the Hamiltonian.
In the upcoming theorem, we examine the uniqueness properties of \eqref{eqn:gra-tvp} for a certain class of Hamiltonians, and in the subsequent theorem we characterize the GRA value function as the unique solution of this TVP when $\hamiltonian$ is the standard Hamiltonian in HJR.

\begin{theorem}[Comparison principle]\label{thm:comparison-principle}
Let $(\graData)$ be a GRA instance and let $\tInit < \tFin$.
Assume $\hamiltonian: \Rn \tms (\tInit, \tFin] \tms \Rn \to \R$ is continuous and satisfies the following:
\begin{enumerate}
\item[(\textit{i})] there is some $K > 0$ such that $|\hamiltonian(\xVal, \tVal, \costate) - \hamiltonian(\xVal, \tVal, \coDum)| \le K (\|\xVal\| + 1)\|\costate - \coDum\|$ for all $\xVal \in \Rn$, $\tVal \in (\tInit, \tFin)$, and $\costate, \coDum \in \Rn$,
\item[(\textit{ii})] for each compact $\compact \ssbs \Rn$, there is an $M > 0$ such that $|\hamiltonian(\xVal, \tVal, \costate) - \hamiltonian(\xDum, \tVal, \costate)| \le M\|\xVal - \xDum\|(1 + \|\costate\|)$ for all $\xVal, \xDum \in \compact$, $\tVal \in (\tInit, \tFin)$, and $\costate \in \Rn$.
\end{enumerate}
If $\hjSol_-: \Rn \tms (\tInit, \tFin] \to \R$ is a viscosity sub-solution of \eqref{eqn:gra-tvp} and $\hjSol_+: \Rn \tms (\tInit, \tFin] \to \R$ is a viscosity super-solution of \eqref{eqn:gra-tvp}, then $\hjSol_- \le \hjSol_+$.
\end{theorem}
\begin{proof}
The rather technical proof of this theorem adapts the proof structure for the analogous result for an HJ-PDE of the form $-\ppt \hjSol - \hamiltonian(\xVal,\tVal, \nabla_x \hjSol) = 0$.
This proof was originally presented in \cite{Ishii-uniqueness-unbounded-1984}, and a simplified version can be found in \cite{Bardi-Dolcetta-Optimal-Control}.
In particular, via a ``cone of dependence argument,'' this proof does not require the relevant viscosity sub-solution and super-solution to be bounded and uniformly continuous.
This approach was extended in \cite{Fialho-Georgiou-TAC-Worst-Case-Analysis-1999} to prove a similar comparison principle for the ``obstacle equation'' $-\min\{\conFn(\xVal, \tVal) - \hjSol, \ppt \hjSol + \hamiltonian(\xVal,\tVal, \nabla_x \hjSol)\} = 0.$
For completeness, we provide the details of our adaptation for the HJ-PDE underlying \eqref{eqn:gra-tvp} in Section \ref{sec:appendix-comparison-principle} of the appendix, following as closely to the simplified proof in \cite{Bardi-Dolcetta-Optimal-Control} as our setting allows.
\end{proof}

\begin{theorem}[HJ-TVP for the GRA value] \label{thm:dpe}
Let $(\graData)$ be a GRA instance, let $\tInit < \tFin$, and assume that the dynamics function $\dynamics$ is continuous on $\Rn \tms \uVals \tms \dVals \tms (\tInit, \tFin]$.
Define the Hamiltonian $\hamiltonian:\Rn \tms (\tInit, \tFin] \tms \Rn \to \R$ by
\begin{equation}\label{eqn:hamiltonian}
    \hamiltonian(\xVal, \tVal, \costate) := \max_{\uVal \in \uVals} \min_{\dVal \in \dVals} \costate \cdot \dynamics(\xVal, \uVal, \dVal, \tVal).
\end{equation}
Then $\hjDep = \valGRA[\graData]|_{\Rn \tms (\tInit, \tFin]}$ is the unique viscosity solution of \eqref{eqn:gra-tvp}.
\end{theorem}
\begin{proof}
It follows directly from \eqref{eqn:gra-definition} that 
$$\valGRA[\graData](\cdot, \tFin) = \min\{\conFn(\cdot, \tFin), \max\{\tarFn(\cdot, \tFin), \terFn(\cdot)\}\},$$
so the boundary condition indeed holds.
That \eqref{eqn:gra-tvp} has at most one viscosity solution follows from Theorem \ref{thm:comparison-principle} and Assumptions \ref{assumption:compactness}-\ref{assumption:regularity}.
All that then remains is to prove that $\hjDep = \valGRA[\graData]|_{\Rn \tms (\tInit, \tFin)}$ is a viscosity solution of the HJ-PDE in \eqref{eqn:gra-tvp}. 

One approach to do so is to use Corollary \ref{cor:dpp} and follow the proof of the corresponding theorem (Theorem 1 in \cite{fisac-chen-2015}) for the reach-avoid task.
An alternative approach that provides more intuition regarding the GRA task leverages the fact that the GRA performance functional can be viewed as a sufficiently uniform limit of reach-avoid performance functionals.
One can then argue a similar limit holds for the corresponding value functions and thereafter use the theory of HJ-PDEs to show the desired result directly from Theorem 1 in \cite{fisac-chen-2015} (rather than by reconstructing its proof).
This latter approach is presented in Section \ref{sec:appendix-dpe} in the appendix.
\end{proof}

The significance of Theorem \ref{thm:dpe} is that to compute the GRA value function one can solve an HJ-TVP.
This is the generalization of the corresponding theorems in the canonical works \cite{mitchell-2005,lygeros-2011,fisac-chen-2015} to the GRA task.
In light of Observation \ref{lem:equivalence}, the value function for each fundamental task can then be characterized using this same HJ-TVP or an equivalent HJ-TVP in a simplified (and perhaps more familiar) form (Table \ref{tab:tvps}).

\begin{table}[ht]
\centering
\begin{tabular}{>{\bfseries} l l l}
\toprule
Task & ~ & Hamilton-Jacobi Terminal Value Problem\\
\midrule\midrule
GRA & PDE & $-\min\{\conFn - \hjDep, \max\{\tarFn - \hjDep, \ppt \hjDep + \hamiltonian(\xVal, \tVal, \nabla_x \hjDep)\}\} = 0$\\
~ & BC & $\hjDep(\xVal, \tFin) = \min\{\conFn(\xVal, \tFin), \max\{\tarFn(\xVal, \tFin), \terFn(\xVal)\}\}$\\
\midrule
~ & PDE & $-\ppt \hjDep - \hamiltonian(\xVal, \tVal, \nabla_x \hjDep) = 0$\\
(I) & BC & $\hjDep(\xVal, \tFin) = \terFn(\xVal)$\\
~ & & Derived from GRA HJ-TVP via $\tarFn = -\ii$, $\conFn = \ii$\\
\midrule
~ & PDE & $-\max\{\tarFn - \hjDep, \ppt \hjDep  + \hamiltonian(\xVal, \tVal, \nabla_x \hjDep)\} = 0$\\
(II) & BC & $\hjDep(\xVal, \tFin)= \tarFn(\xVal, \tFin)$\\
~ & ~ & Derived from GRA HJ-TVP via $\terFn(\cdot) \le \tarFn(\cdot, \tFin)$, $\conFn = \ii$\\
\midrule
~ & PDE & $-\min\{\conFn - \hjDep, \ppt \hjDep + \hamiltonian(\xVal, \tVal, \nabla_x \hjDep)\} = 0$\\
(III) & BC & $\hjDep(\xVal, \tFin)=\conFn(\xVal, \tFin)$\\
~ & ~ & Derived from GRA HJ-TVP via $\terFn(\cdot) \ge \conFn(\cdot, \tFin)$, $\tarFn = -\ii$\\
\midrule
~ & PDE & $-\min\{\conFn - \hjDep, \ppt \hjDep + \hamiltonian(\xVal, \tVal, \nabla_x \hjDep)\} = 0$\\
(IV) & BC & $\hjDep(\xVal, \tFin)=\min\{\conFn(\xVal, \tFin), \terFn(\xVal)\}$\\
~ & ~ & Derived from GRA HJ-TVP via $\tarFn = -\ii$\\
\midrule
~ & PDE & $-\min\{\conFn - \hjDep, \max\{\tarFn - \hjDep, \ppt \hjDep + \hamiltonian(\xVal, \tVal, \nabla_x \hjDep)\}\} = 0$\\
(V) & BC & $\hjDep(\xVal, \tFin)=\min\{\conFn(\xVal, \tFin), \tarFn(\xVal, \tFin)\}$\\
~ & ~ & Derived from GRA HJ-TVP via $\terFn(\cdot) \le \tarFn(\cdot, \tFin)$\\
\bottomrule
\end{tabular}
\caption{Summary of the HJ-TVPs for the canonical fundamental problems in HJR.
Each HJ-TVP consists of an HJ-PDE and a terminal boundary condition (BC).
For each of the tasks (I)-(V), the HJ-TVP can be derived from the GRA task's HJ-TVP. To do so, one substitutes the conditions from Observation \ref{lem:equivalence} into the GRA task's HJ-TVP and simplifies the resulting equations.}
\label{tab:tvps}
\end{table}

\subsection{Representation of solutions to HJR's canonical PDE}
In this sub-section, we fix $\tInit, \tFin \in \R$ with $\tInit < \tFin$.
We also fix continuous functions $\tarFn, \conFn: \Rn \tms \R \to \extended$ such that $\tarFn < \ii$ and $\conFn > -\ii$.
Finally, we will assume the dynamics function $\dynamics$ is continuous on $\Rn \tms \uVals \tms \dVals \tms (\tInit, \tFin]$, and we set the Hamiltonian $\hamiltonian:\Rn \tms (\tInit, \tFin] \tms \Rn \to \R$ using \eqref{eqn:hamiltonian}.

Recall that an HJ-TVP is an HJ-PDE together with a terminal boundary condition.
For a brief moment, let us ignore the boundary condition and focus solely on the HJ-PDE itself.
The canonical HJ-PDE in HJR is
\begin{align*}\label{eqn:hjr-pde}
    -\min\{& \conFn(\xVal,\tVal) - \hjDep, \max\{ \tarFn(\xVal,\tVal) - \hjDep, \\
    &\ppt \hjDep + \hamiltonian(\xVal, \tVal, \nabla_x \hjDep) \}\} = 0,
    \quad \xVal \in \Rn,~\tVal \in (\tInit, \tFin).&\tag{Canonical PDE}
\end{align*}
Indeed, the HJ-TVP for each of the five canonical tasks in HJR can be written using this same underlying PDE (see Table \ref{tab:tvps}).

Now let us consider the importance of the boundary condition.
Take the reach-avoid task, i.e. task (V), in which the terminal boundary condition is given by $$\hjDep(\xVal,\tFin) = \min\lf\{\conFn(\xVal, \tFin), \tarFn(\xVal,\tFin) \rg\}, \quad \xVal \in \Rn.$$
For this task, the boundary condition is entirely determined by the data (i.e. $\tarFn$ and $\conFn$) that appears within the HJ-PDE itself.
It is thus natural to ask: \textit{are there other solutions to \eqref{eqn:hjr-pde} (i.e. ones corresponding to alternative boundary conditions), and if so can they naturally be interpreted as value functions for some task?}

Based on the previous sub-section, the answer to the above question is yes.
Indeed, in \eqref{eqn:gra-tvp}, the boundary condition is 
$$\hjDep(\xVal,\tFin) = \min\lf\{\conFn(\xVal, \tFin), \max\{\tarFn(\xVal,\tFin), \terFn(\xVal) \} \rg\}, \quad \xVal \in \Rn.$$
This boundary condition includes the terminal set function $\terFn$, which does not appear in the HJ-PDE itself, unlike $\tarFn$ and $\conFn$.
Thus, whereas the reach-avoid boundary conditions are fixed solely by the data within the PDE, in the GRA boundary condition the term $\terFn$ serves as a degree of freedom.
This degree of freedom can be changed to obtain various solutions of \eqref{eqn:hjr-pde}, and
by Theorem \ref{thm:dpe} each of those solutions can be represented by a GRA value function.

That said, there are still limits to the expressivity of the boundary data in \eqref{eqn:gra-tvp}.
In particular, regardless of the choice of $\terFn$, we still have 
\begin{align*}
\min\{&\conFn(\cdot, \tFin), \tarFn(\cdot, \tFin)\} \le \\
&\min\{\conFn(\cdot,\tFin), \max\{\tarFn(\cdot,\tFin), \terFn(\cdot)\}\} \le \conFn(\cdot, \tFin).
\end{align*}
It happens to be the case, however, that the structure of \eqref{eqn:hjr-pde} itself requires that any of its viscosity solutions $\hjDep$ satisfies 
$$\min\{\conFn, \tarFn\} \le \hjDep \le \conFn.$$
As such, the degree of freedom that the terminal set function $\terFn$ offers within the GRA formulation is in fact enough to represent \textit{every} sufficiently regular viscosity solution of \eqref{eqn:hjr-pde}.
This idea is summarized in the following corollary of Theorem \ref{thm:dpe}.

\begin{corollary}[Representation of solutions]\label{cor:parameterization}
Let $\tInit, \tFin \in \R$ with $\tInit < \tFin$,
let $\tarFn, \conFn: \Rn \tms \R \to \extended$ be continuous with $\tarFn < \ii$ and $\conFn > -\ii$, and assume the dynamics function $\dynamics$ is continuous on $\Rn \tms \uVals \tms \dVals \tms (\tInit, \tFin]$. Define $\hamiltonian:\Rn \tms (\tInit, \tFin] \tms \Rn \to \R$ by \eqref{eqn:hamiltonian}.

For every viscosity solution $\hjDep:\Rn \tms (\tInit, \tFin) \to \R$ of \eqref{eqn:hjr-pde} that is uniformly continuous on $\compact \tms [\tVal, \tFin)$ for all $\tVal \in (\tInit, \tFin)$ and all compact $\compact \ssbs \Rn$, there is a continuous $\terFn: \Rn \to \R$ such that $\hjDep = \valGRA[\graData]|_{\Rn \tms (\tInit, \tFin)}$.
In particular, $\terFn$ can be chosen to be $\terFn(\cdot) := \hjDep(\cdot,\tFin^-)$.
\end{corollary}
\begin{proof}
    Let $\hjDep: \Rn \tms (\tInit, \tFin) \to \R$ be such a viscosity solution.
    It follows from Lemma \ref{lem:hjr-pde-bounds} (Section \ref{sec:appendix-comparison-principle} in the appendix) that $\min\{\conFn, \tarFn\} \le \hjDep \le \conFn$.
    It follows from uniform continuity on each $\compact \tms [\tVal, \tFin)$ that $\hjDep$ continuously extends to a function $\bar{\hjDep}: \Rn \tms (\tInit, \tFin] \to \R$.
    Letting $\terFn(\cdot) := \bar{\hjDep}(\cdot, \tFin)$, it follows that $\min\{\conFn(\cdot, \tFin),\tarFn(\cdot,\tFin)\} \le \terFn(\cdot) \le \conFn(\cdot, \tFin)$.
    Thus $\bar{\hjDep}(\cdot, \tFin) = \min\{\conFn(\cdot, \tFin), \max\{\tarFn(\cdot, \tFin), \terFn(\cdot)\}\}$, so that $\bar{\hjDep} = \valGRA[\tarFn, \conFn, \terFn; \tFin]$ by Theorem \ref{thm:dpe}.
\end{proof}

\begin{remark}
We require the uniform continuity on each $\compact \tms [\tVal, \tFin)$ for a rather technical reason.
A viscosity solution to an HJ-TVP by definition must be well-defined at the terminal boundary $\Rn \tms \{\tFin\}$, but a viscosity solution of an HJ-PDE need not be.
This regularity requirement is sufficient and necessary for a real-valued function on $\Rn \tms (\tInit, \tFin)$ to continuously extend to this terminal boundary.

There indeed exist viscosity solutions $\hjSol$ of \eqref{eqn:hjr-pde} for which this regularity requirement does not hold, in which case we cannot hope to represent $\hjSol$ as a GRA value function.
(See Section \ref{sec:pathology} in the appendix for an example of such a pathology.)
As such, this regularity requirement cannot be removed from the hypothesis of Corollary \ref{cor:parameterization}. \blackbox
\end{remark}

To summarize, whereas Theorem \ref{thm:dpe} states that any GRA value function must be a solution to \eqref{eqn:hjr-pde} (with the same target tube function $\tarFn$ and constraint tube function $\conFn$), Corollary \ref{cor:parameterization} states that any sufficiently regular solution to \eqref{eqn:hjr-pde} must be a GRA value function (again with the same $\tarFn$ and $\conFn$).
In this way, Corollary \ref{cor:parameterization} can be interpreted as a converse to Theorem \ref{thm:dpe}.

\subsection{Application to dynamics that are piecewise-continuous in time}
In this sub-section, we discuss how to obtain the GRA value function when the dynamics function $\dynamics$ is piecewise-continuous in time.
This setting is useful to deal with systems that switch behavior at certain prescribed times, such as turning on or off.

Suppose there exist times $\tNot < \tVal_1 < \tVal_2 < \dots < \tVal_N$ such that the dynamics $\dynamics$ is continuous on $\Rn \tms \uVals \tms \dVals \tms (\tVal_{i-1}, \tVal_{i}]$ for each $i = 1,\dots,N$.
Let $(\tarFn, \conFn, \terFn; \tVal_N)$ be a GRA instance.

To compute $\valGRA[\tarFn, \conFn, \terFn; \tVal_N]$ on $\Rn \tms [\tNot, \tVal_N]$, we cannot directly use Theorem \ref{thm:dpe} since it assumes continuity of $\dynamics$ over the entire domain. 
However, Corollary \ref{cor:dpp} and Theorem \ref{thm:dpe}  together allow us to accomplish this task for such a system via Algorithm \ref{alg:discontinuous}.

\begin{algorithm}
\caption{Iterative generalized reach-avoid}
\label{alg:discontinuous}
\begin{algorithmic}
\REQUIRE GRA instance $(\tarFn, \conFn, \terFn; \tVal_N)$, times $\tNot,\dots, \tVal_{N-1}$.
\ENSURE Compute the value function $\val := \valGRA[\tarFn, \conFn, \terFn; \tVal_N]$ on $[\tNot, \tVal_N]$.

\FOR{$i = N,\dots,1$}
    \STATE $\tFin \gets \tVal_{i}$;
    \STATE $\tInit \gets \tVal_{i-1}$;   
    \STATE Compute $\val$ on the domain $\Rn \tms [\tVal_{i-1}, \tVal_i]$ by solving \eqref{eqn:gra-tvp} with $\hamiltonian$ as in \eqref{eqn:hamiltonian}; \comment see Theorem \ref{thm:dpe}
    \STATE $\terFn \gets \val(\cdot, \tInit)$;  \comment see Corollary \ref{cor:dpp}
\ENDFOR
\RETURN $\val$
\end{algorithmic}
\end{algorithm}

\section{Examples}\label{sec:examples}

We now give three examples to showcase the GRA task.
As is increasingly common in the HJR literature (see e.g. \cite{stl-meets-reachability,sharpless2026dual,sharpless2026bellman,Xiang-CDC-2025}), in each example we use temporal logic notation to formally specify the task of interest.
Temporal logic provides a rich language for describing complex tasks \cite{STL,baier2008principles}, and it automatically associates a task specification $\phi$ with a certain performance functional $\rho_\phi$, known as the robustness metric for $\phi$ \cite{donze-robustness-metric}.
For readers unfamiliar with this notation, we will also explicitly provide this performance functional, which is all that is needed to understand the example.

Before demonstrating the importance of the GRA formulation for value function decomposition, we explore its utility as a standalone task in Section \ref{sec:example-1}.
In Section \ref{sec:example-2}, we show how Theorem \ref{thm:value-decomposition} can be used to compute the value function for a timed temporal logic specification via value function decomposition.
In Section \ref{sec:example-3}, we show how Algorithm \ref{alg:discontinuous} can be used to compute the value function for a system with dynamics that are piecewise-continuous in time.

\subsection{Example 1: reach-avoid or always-avoid}\label{sec:gr-example}\label{sec:example-1}

\newcommand{\RAAA}{\rob_\mathrm{RAT \lor AT}}
Consider a drone flying over water, with an L-shaped landing pad in the area.
There can be strong gusts of wind lasting up to 10 seconds that the drone is not capable of fighting.
These winds blow most strongly toward the north and east.
The full dynamics are
$$\dot{\xSig} = \uSig + \dSig,$$
with $\xSig(\tVal) \in \R^2$ being the horizontal position of the drone (we shall neglect the vertical position).
The control space is $\uVals = \{\uVal \in \R^2 \mid \| \uVal\|\le 1 \}$, and the disturbance space is $\dVals = \{\dVal \in \R^2 \mid -1.1 \le \dVal_1 \le 2.0, -1.1 \le \dVal_2 \le 2.0\}$.

Once a gust begins, we would like to know from which locations the drone can either remain within 9 meters of a communication antenna for 10 seconds, or first get to the landing pad.
Let $\tFin = 10$, $\tarFn$ be the signed distance function to the landing pad, and $\conFn(\cdot) := 9 - a(\cdot)$, with $a$ the distance function to the antenna (recall $9$ meters is the communication radius).
The temporal logic specification $\phi$ for this task is
\begin{equation}
    \phi := (\conFn~\mathsf{U}~\tarFn) \lor (\mathsf{G}~ \conFn),\footnote{This specification can also be written as $\phi := \tarFn~\mathsf{R}~\conFn$, where $\mathsf{R}$ is the \textit{release} operator \cite{release-operator}. Conceptually, reaching the target ``releases'' one from having to obey the constraints.
    This example thus shows that we can use the GRA formulation to compute the value function for the release operator.}
\end{equation}
which states that the system must obey the constraints $\conFn$ until reaching the target $\tarFn$ or it must globally (i.e. always) obey the constraints.

\begin{figure}
    \centering
    \includegraphics[width=1.0\linewidth]{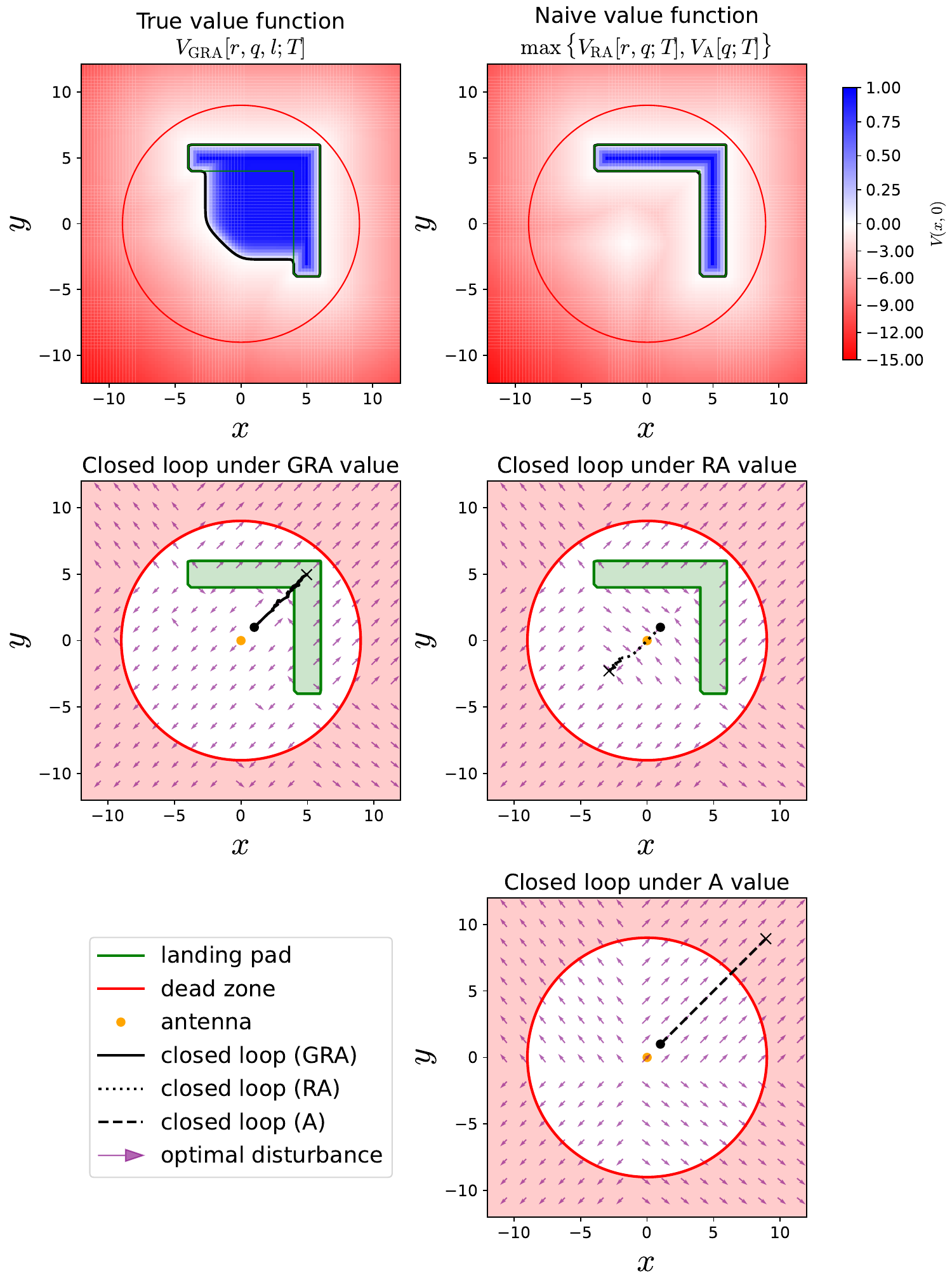}
    \caption{Example 1 (drone over water). The top plots visualize the true value function for the task in comparison with the naive value function on the right-hand side of \eqref{eqn:wrong-value}.
    The middle-left plot shows the drone completing the task.
    Here, the red circle represents the communication area in which the drone should remain, and the green, L-shaped area is the landing pad.
    The middle-right and bottom-right plots show that from this same point, the drone would fail to complete the associated reach-avoid and avoid tasks.
    In other words, the drone can complete the ``reach-avoid or always-avoid'' task but not the ``reach-avoid'' task or the ``avoid'' task. See Counter-Example \ref{rem:counter-example} for a similar example.}
    \label{fig:example-1}
\end{figure}
The performance functional corresponding to this task is the robustness metric $\rho_\phi$ for $\phi$, which is by definition
\begin{equation*}
    \rho_\phi = \max\{\rho_{\conFn\mathsf{U}\tarFn}, \rho_{\mathsf{G} \conFn} \} = \max\{ \robRTAT[\tarFn, \conFn;\tFin], \robAT[\conFn; \tFin] \}.
\end{equation*}
Conceptually, this performance functional looks for the better of two options: \textit{(i)} reach the landing pad before 10 seconds without entering the dead zone (i.e. the region where communication fails) first, or \textit{(ii)} do not enter the dead zone for the full 10 seconds.

As in Counter-Example \ref{rem:counter-example}, the value operation $\val[\cdot]$ does not generally commute with the $\max\{\cdot, \cdot\}$ operation, i.e.
\begin{equation}\label{eqn:wrong-value}
    \val[\rho_\phi] \ne \max\{ \valRTAT[\tarFn, \conFn;\tFin], \valAT[\conFn; \tFin]\}.
\end{equation}
Thus we cannot use the right-hand side to obtain the desired value function.
Instead, one can use the GRA formulation to directly obtain $\val[\rho_\phi]$.
To do so, define $\terFn := \conFn(\cdot, \tFin)$.
Then $\rho_\phi = \robGRA[\graData]$, so that
$$\val[\rho_\phi] = \valGRA[\graData].$$
In Figure \ref{fig:example-1}, we show the true value function for this task, computed via the above formula, and compare it to the naive value function $\max\{ \valRTAT[\tarFn, \conFn;\tFin], \valAT[\conFn; \tFin] \}$.

\subsection{Example 2: timed temporal logic}\label{sec:example-2}

We will now show by way of example how to use the value decomposition in Theorem \ref{thm:value-decomposition} to compute the value function for some specifications in timed temporal logic.

Suppose we have two robots on a warehouse floor that can be used to pick up packages from a shelf located at the far end of the warehouse.
Robot B is an older model that is slower and takes more time to load packages from the shelf.
We will assume Robot A has dynamics $\dot{\xSig}_{A} = \uSig_{A} + \dSig_{A}$ and Robot B has dynamics $\dot{\xSig}_{B} = \uSig_{B} + \dSig_{B}$, each having 2 states.
The whole system is thus 
$$\dot{\xSig} = \uSig + \dSig,$$
with $\xSig := (\xSig_{A}, \xSig_{B})$, $\uSig := (\uSig_{A}, \uSig_{B})$, and $\dSig := (\dSig_{A}, \dSig_{B})$.
We impose actuator constraints $\uVals = \{\uVal \in \R^4 \mid \uVal_1^2 + \uVal_2^2 \le 1, \uVal_3^2 + \uVal_4^2 \le 0.5^2 \}$ and disturbance constraints $\dVals =  \{\dVal \in \R^4 \mid \|\dVal\|_2 \le 0.2\}$, with the latter reflecting an upper bound on the total modeling error.
We let $r_A$ (resp. $r_B$) be the signed distance function from the shelf to robot A (resp. robot B), and we let $q$ be the signed distance function for the constraint set in which both robots are at least 0.5 meters from the factory walls and from each other.

We consider the following task specification $\psi$, written in timed temporal logic notation:
\begin{equation}
    \displaystyle \psi := (q\,\mathsf{U}_{[0,60]}\,r_A) \lor (q\,\mathsf{U}_{[0,40]}\, r_B).
\end{equation}
In words, this specification states that robot A must arrive at the shelf within the next 60 seconds without violating any constraints  on the way, or robot B must arrive at the shelf within the next 40 seconds without violating constraints on the way (the difference in required arrival times is needed to account for the robots' different loading times).

The standard performance functional corresponding to this task specification is the robustness metric $\rho_\psi$, defined by
\begin{align*}
    \rho_\psi (\xSig, \tVal) = &\max\{\rho_{\conFn\mathsf{U}_{[0,60]}\tarFn_A}(\xSig, \tVal), \rho_{\conFn\mathsf{U}_{[0,40]}\tarFn_B}(\xSig, \tVal)\}\\
    =&\max\Big\{\max_{\tDum \in [\tVal, \tVal + 60]} \min\big\{ r_A(\xSig(\tDum)), \min_{\tDumDum \in [\tVal, \tDum]} q(\xSig(\tDumDum)) \big\},\\
    & \qquad~~\,\max_{\tDum \in [\tVal, \tVal + 40]} \min\big\{ r_B(\xSig(\tDum)), \min_{\tDumDum \in [\tVal, \tDum]} q(\xSig(\tDumDum)) \big\} \Big\}.
\end{align*}
Because the robots' dynamics do not depend directly on time, $\val[\rho_\psi](\xVal, \tVal)$ will also be time-invariant, so it suffices to compute $\val[\rho_\psi](\xVal, 0)$.

We can do so as follows.
First, we rewrite this performance functional as
\begin{align*}
    \rho_\psi (\xSig, 0) &= \max\Big\{\robRTAT[\max\{r_A, r_B\}, q; 40](\xSig, 0), \\
    &\min\big\{\robRTAT[r_A, q; 60](\xSig, 40), \robAT[q; 40](\xSig, 0)\big\}\Big\}.
\end{align*}
Essentially, this rearrangement says that either robot must be able to reach the shelf within 40 seconds without violating constraints, or robot A must be able to reach the shelf between seconds 40 and 60 without violating constraints.
We can now simply apply Theorem \ref{thm:value-decomposition} to obtain the value function $\val[\rho_\psi](\xVal, 0)$ (in applying this theorem, we let $\tFin := 40$, $\tarFn:= \max\{r_A, r_B\}$, and $\rob_0 := \robRTAT[r_A, q; 60]$).

\begin{figure}[!ht]
    \centering
\includegraphics[width=0.8\linewidth]{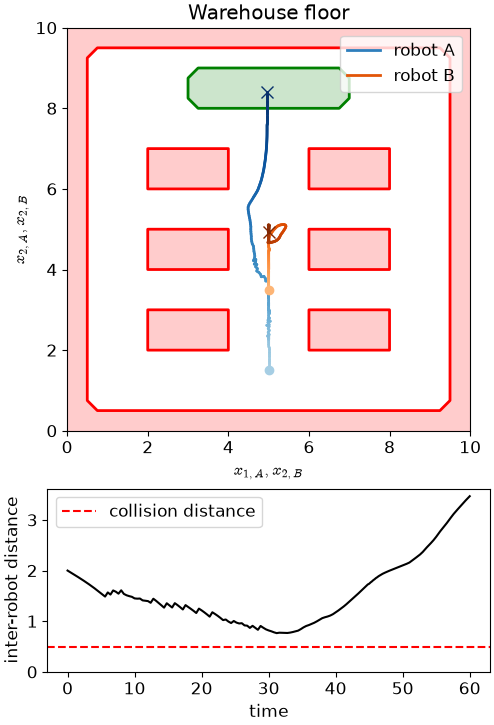}
    \caption{Example 2 (warehouse robots).
    Either robot A must get to the shelf (green) within 60 seconds or robot B must get to it in 40 seconds while avoiding hitting one another and the walls (red).
    In this initial configuration, robot B cannot reach the shelf in 40 seconds, but robot A can reach it in 60 seconds, except robot B is blocking robot A's path.
    Robot B thus moves out of the way to clear a path.
    }
    \label{fig:placeholder}
\end{figure}

Deployment of the resulting closed-loop control and disturbance is displayed in Figure \ref{fig:placeholder} for an initial configuration from which robot B cannot reach the shelf quickly enough and thus must move out of the way to make room for robot A to accomplish the task.

\subsection{Example 3: piecewise-continuous dynamics}\label{sec:example-3}
\newcommand{\drg}{\mathrm{D}}
\newcommand{\eff}{\mathrm{E}}
\newcommand{\tox}{\mathrm{Y}}
\newcommand{\effThresh}{\theta_{\textup{eff}}}
\newcommand{\toxThresh}{\theta_{\textup{tox}}}

In this example, we showcase how Algorithm \ref{alg:discontinuous} can be used to naturally handle problems where the dynamics abruptly switch.
Consider a pharmacology problem in which a patient undergoes IV infusion of two drugs over the span of $\tFin = 10$ hours.
There are 3 infusion sessions, each 2 hours long, with a 2 hour break between sessions.
The drugs interact in a way that makes the treatment more effective but also creates toxicity to the patient.
We can model the system as follows:
\begin{align*}
    &\dot{\drg}_1 = \omega(\tVal) \uSig_1 - \alpha_1(1 + \dSig_1) \drg_1\\
    &\dot{\drg}_2 = \omega(\tVal)\uSig_2 - \alpha_2(1 + \dSig_2) \drg_2\\
    &\dot{\eff} = \eta_1 \frac{\drg_1^2}{c_1^2 + \drg_1^2} +  \eta_2 \frac{\drg_2^2}{c_2^2 + \drg_2^2} +
    \eta_3
    \frac{\drg_1^2}{c_1^2 + \drg_1^2}\frac{\drg_2^2}{c_2^2 + \drg_2^2} - \beta \eff \\
    &\dot{\tox} = \chi_1 \frac{\drg_1^2}{k_1^2 + \drg_1^2} +  \chi_2 \frac{\drg_2^2}{k_2^2 + \drg_2^2} + \chi_3 \frac{\drg_1^2}{k_1^2 + \drg_1^2}\frac{\drg_2^2}{k_2^2 + \drg_2^2} - \gamma \mathrm{\tox}
\end{align*}

Here $\drg_1$ and $\drg_2$ are the concentrations of each drug in blood, $\eff$ is the therapeutic effect, and $\tox$ is the toxicity level.
The controls $\uSig_1$ and $\uSig_2$ represent supply rate of each drug and $\dSig_1$ and $\dSig_2$ represent possible fluctuations in the metabolism of each drug.
The time-varying parameter $\omega$ is $1$ during treatment and $0$ otherwise, namely during the breaks.

We would like for the therapeutic effect $\eff$ to reach a therapeutic threshold $\effThresh$ while keeping the toxicity below a threshold $\toxThresh$.
We will define the target tube function by $\tarFn := \eff - \effThresh$ and the constraint by $\conFn := \toxThresh - \tox$.
We let the control set be $\uVals := [0,1]^2$ and the disturbance set to be $\dVals := \{\dVal \in \R^2 \mid \|\dVal\|_2 \le \dVal_{\mathrm{max}}\}$.
The temporal logic formula $\mu$ for this specification is simply
$$\mu := \conFn~\mathsf{U}~\tarFn,$$
(i.e. obey constraints until the target is reached) and the corresponding performance functional $\rho_\mu$ is simply
$$\rho_\mu = \rho_{\conFn \mathsf{U} \tarFn} = \robRTAT[\tarFn, \conFn; \tFin].$$
The value function $\val[\rho_\mu]$ can then be computed via Algorithm \ref{alg:discontinuous}.
The resulting optimal dosing regimen is shown in Figure \ref{fig:example-2}.

\begin{figure}
    \centering
    \includegraphics[width=1.0\linewidth]{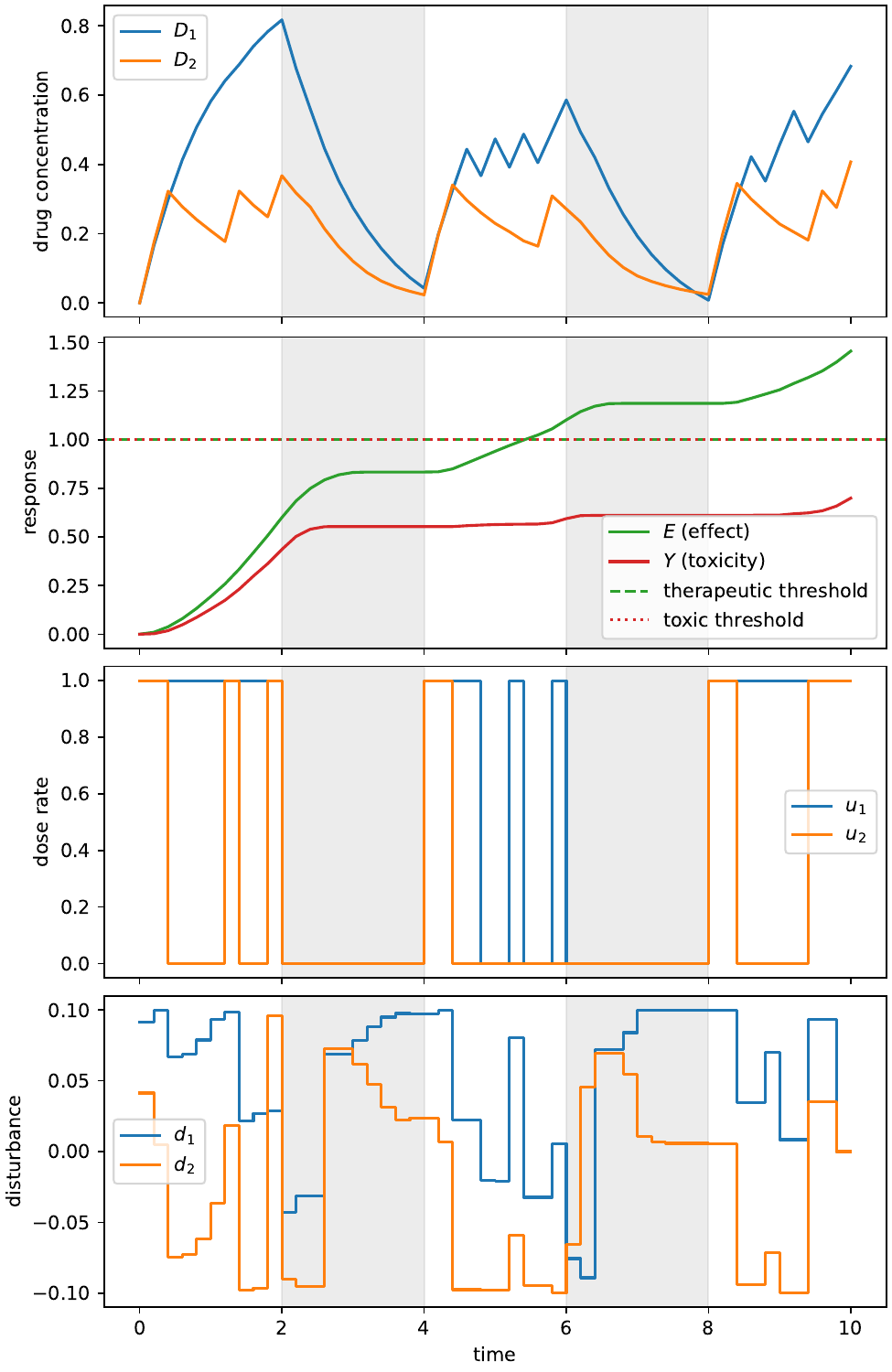}
    \caption{Example 3 (optimal dosing). The regions in white represent periods during which the patient receives the infusion while the gray regions represent break times in the treatment.
    The value functions are stitched together (backwards in time) across points of discontinuity using Algorithm \ref{alg:discontinuous}.}
    \label{fig:example-2}
\end{figure}

Note that although this problem is inherently a reach-avoid task, in order to account for the discontinuity in the dynamics, we still must make use of the terminal set function $\terFn$ in Algorithm \ref{alg:discontinuous} to stitch the value functions together across the points of discontinuity.
In other words, in the process of computing the reach-avoid value function $\valRTAT[\tarFn, \conFn; \tFin]$, we still compute the solution to several iterative GRA HJ-TVPs.

\section{Scope and Limitations}
One limitation of this work is that it only addresses the types of tasks that can be solved using HJR.
It does not alleviate the curse of dimensionality that makes practical computation of the value function challenging for high-dimensional systems, regardless of the task.
Deep-learning approaches, including deep reinforcement learning \cite{Reachability-RL,hsu-2021} and physics informed neural networks \cite{deepreach,lin-bansal-2023}, have increasingly been used to mitigate these issues when computing value functions for the fundamental tasks.
These methods can also be extended to composite tasks as in \cite{sharpless2026dual,sharpless2026bellman}.

Moreover, in this work, we have only considered finite-horizon problems.
For learning-based works, one often considers an infinite-horizon adaptation of the canonical tasks under various approaches to discounting.
Many of the results for finite-time tasks can be extended to the infinite time setting by analyzing time-invariant solutions of the time-dependent HJ-PDEs \cite{viscosity-cbfs}.
We plan to address the infinite-time case in this manner in a smaller future work.

\section{Conclusion}
In this work we introduced the GRA task, which not only extends the fundamental tasks in HJR, but also unifies the canonical, finite-horizon fundamental tasks into one common primitive.
The GRA simultaneously considers terminal sets, target tubes, and constraint tubes, whereas the previous tasks only considered one or two of these elements at a time.

Moreover, the GRA task is in some sense the natural primitive of HJR to consider from a PDE perspective.
In particular, for a given target tube and constraint tube function, every sufficiently regular solution of the canonical HJ-PDE of HJR can be represented using the GRA value function with an appropriate choice of terminal set function.

Beyond theoretical unification, however, the GRA formulation enables dynamic-programming-based value function decomposition for certain composite tasks.
To do so, it uses the terminal set function as an interface between the upstream and downstream aspects of the task.
These tasks include some tasks whose performance functionals are given by the robustness metric from timed temporal logic specifications.

Further work is now required to show how to decompose increasingly complex and important composite tasks into this primitive task, extending the line of work first studied in \cite{stl-meets-reachability}.

\appendix
\newcommand{\ep}{\varepsilon}
\subsection{Definition of a viscosity solution of an HJ-PDE}\label{sec:appendix-viscosity}
An HJ-PDE is a first-order nonlinear PDE of the form
\begin{equation}\label{eqn:hj-pde}
    \hjFn(\hjInd, \hjDep, \nabla \hjDep) = 0, \quad \hjInd \in \hjDomain. \tag{HJ-PDE}
\end{equation}
Here, $\hjInd$ represents a vector of the independent variables, $\hjDep$ represents the dependent variable, which is scalar, and $\hjDomain$ is the problem domain.

Central to HJ-PDE theory is the notion of a viscosity solution, a certain type of weak solution to an HJ-PDE.
For completeness, we here make explicit the definition we use for the viscosity solution of an HJ-PDE, namely Definitions II.1.1 and V.1.1 in \cite{Bardi-Dolcetta-Optimal-Control}:

\begin{definition}[Viscosity solution]
    Let $\hjDomain \sbs \Rm$ be open and let $\hjFn: \hjDomain \tms \R \tms \Rm \to \R$ be continuous.
    A function $\hjSol: \hjDomain \to \R$ is a
    \begin{itemize}
       \item \textbf{viscosity sub-solution} of \eqref{eqn:hj-pde} if
        \begin{enumerate}
        \item [(\textit{i})]
        $\hjSol$ is upper semi-continuous,
        \item[(\textit{ii})]
        for each $\hjInd_0 \in \hjDomain$ and continuously differentiable $\hjTest:\hjDomain \to \R$,
        if $\hjSol - \hjTest$ has a local maximum at $\hjInd_0$ then
        $$\textstyle \hjFn\Big(\hjInd_0, \hjSol(\hjInd_0), \nabla \hjTest(\hjInd_0) \Big) \le 0;$$
        \end{enumerate}
        
        \item \textbf{viscosity super-solution} of \eqref{eqn:hj-pde} if
        \begin{enumerate}
        \item [(\textit{i})]
        $\hjSol$ is lower semi-continuous,
        \item[(\textit{ii})]
        for each $\hjInd_0 \in \hjDomain$ and continuously differentiable $\hjTest:\hjDomain \to \R$,
        if $\hjSol - \hjTest$ has a local minimum at $\hjInd_0$ then
        $$\textstyle \hjFn\Big(\hjInd_0, \hjSol(\hjInd_0), \nabla \hjTest(\hjInd_0) \Big) \ge 0;$$
        \end{enumerate}

        \item \textbf{viscosity solution} of \eqref{eqn:hj-pde} if 
        $\hjSol$ is both a viscosity sub-solution and super-solution of \eqref{eqn:hj-pde}.
    \end{itemize}
\end{definition}

In HJR, we are primarily concerned with Hamilton-Jacobi terminal-value problems (HJ-TVPs), which involve a time-dependent HJ-PDE.
To emphasize this time-dependence, it is common to write the independent variable $\hjInd$ as $(\xVal, \tVal)$ and the gradient of the dependent variable $\nabla \hjDep$ as $(\nabla_\xVal \hjDep, \ppt \hjDep)$, as we do in Section \ref{sec:pde-background}.

\subsection{Proofs of the basic lemmas}\label{sec:appendix-basic-proofs}

We will make frequent use of the monotonically increasing homeomorphism $\homeo:\extended \to [-1,1]$ defined by
\begin{equation}\label{eqn:homeomorphism-def}
    \homeo(a) = \begin{cases}
    -1 & a = -\ii,\\
    2\tan^{-1}(a)/\pi & a \in \R,\\
    +1 & a = \ii. \end{cases}
\end{equation}

\begin{proof}[Proof of Lemma \ref{lem:continuous-evaluations-have-continuous-values}]

\newcommand{\Xxt}{\mathcal{F}}
\newcommand{\trajDum}{\traj{\xDum, \tDum}{\uSig, \dSig}}
\newcommand{\traji}{\traj{\xVal, \tVal}{\uSig_i, \dSig_i}}
\newcommand{\trajdumi}{\traj{\xDum_i, \tDum_i}{\uSig_i, \dSig_i}}
\newcommand{\trajj}{\traj{\xVal, \tVal}{\uSig_{i_j}, \dSig_{i_j}}}
\newcommand{\trajdumj}{\traj{\xDum_{i_j}, \tDum_{i_j}}{\uSig_{i_j}, \dSig_{i_j}}}
\newcommand{\Jtil}{\hat{\rob}}

Suppose $\rob: \xSigs \tms \evalTimes \to \extended$ is continuous.
Fix $\xVal \in \Rn$ and $\tVal \in \evalTimes$.
Let $\homeo:\extended \to [-1,1]$ be given by \eqref{eqn:homeomorphism-def}.
Let $\Jtil = \sigma \circ \rob$.

First we show that $\val[\Jtil](\xVal, \tVal)$ is continuous at $(\xVal, \tVal)$.
Fix $\ep > 0$.
For each $i \in \N$, let $R_i = 1/i$ and choose $\xDum_i \in B(\xVal; R_i)$, $\tDum_i \in \evalTimes \cap (\tVal - R_i, \tVal + R_i)$, $\uSig_i \in \uSigs$, and $\dSig_i \in \dSigs$ such that
\begin{align*}
    \sup_{\xDum \in B(\xVal; R_i)}& \sup_{\tDum \in \evalTimes \cap (\tVal - R_i, \tVal + R_i)} \sup_{\uSig \in \uSigs} \sup_{\dSig \in \dSigs}|\Jtil(\standardTraj, \tVal) - \Jtil(\trajDum, \tDum)| \\
    &\le |\Jtil(\traji, \tVal) - \Jtil(\trajdumi, \tDum_i)| + \ep.
\end{align*}

Let $\Xxt = \{\traj{\xDum, \tDum}{\uSig, \dSig} \in \xSigs \mid \|\xDum - \xVal\| \le 1, |\tDum - \tVal| \le 1, \uSig \in \uSigs, \dSig \in \dSigs \}$.
It is a standard result that $\Xxt$ is locally uniformly bounded and equi-continuous under Assumptions \ref{assumption:compactness}-\ref{assumption:regularity}, so $\Xxt$ has compact closure by the Arzel\`{a}-Ascoli theorem (Theorem 4.44 in \cite{Folland-Real-Analysis}).
Thus we can choose subsequences $\trajj$ and $\trajdumj$ of $\traji$ and $\trajdumi$, respectively, that converge in $\Xxt$.
By Assumptions \ref{assumption:compactness}-\ref{assumption:regularity}, these subsequences must in fact converge to the same trajectory, say $\xSig \in \xSigs$.
Then for each $j \in \N$,
\begin{align*}
    \sup_{\xDum \in B(\xVal; R_{i_j})}& \sup_{\tDum \in \evalTimes \cap (\tVal - R_{i_j}, \tVal + R_{i_j})}
    |\val[\Jtil](\xVal, \tVal) - \val[\Jtil](\xDum, \tDum)| \\
    &\le |\Jtil(\trajj, \tVal) - \Jtil(\trajdumj, \tDum_{i_j})| + \ep\\
    &\to |\Jtil(\xSig, \tVal) - \Jtil(\xSig, \tVal)| + \ep = \ep.
\end{align*}
Since $\ep$ was chosen arbitrarily, continuity is shown.

Now suppose that $\rob < \ii$ so that $\Jtil < 1$.
Let $\Xxt = \{\standardTraj \in \xSigs \mid \uSig \in \uSigs, \dSig \in \dSigs \}$.
Since $\Xxt$ again has compact closure by Arzel\`{a}-Ascoli and $\Jtil$ is continuous, it follows that the image $\Jtil(\bar{\Xxt},\tVal)$ is a compact subset of $[-1,1)$. Thus $\val[\Jtil](\xVal, \tVal) \le \max_{\xSig \in \bar{\Xxt}} \Jtil(\xSig,\tVal) < 1$.
An analogous argument shows that $\val[\Jtil](\xVal, \tVal) > -1$ if $\rob > -\ii$.

The lemma then follows from noticing that $\val[\rob] = \sigma^{-1} \circ \val[\Jtil]$ by monotonicity and continuity of $\sigma^{-1}$.
\end{proof}

\begin{proof}[Proof of Lemma \ref{lem:gra-evaluation-is-continuous}]
    \newcommand{\rt}{\hat{\tarFn}}
    \newcommand{\qt}{\hat{\conFn}}
    \newcommand{\lt}{\hat{\terFn}}
    
    It follows directly from \eqref{eqn:gra-definition} that $\robGRA[\graData]$ is real-valued and past-independent.
    It thus only remains to confirm that $\robGRA[\graData]$ is continuous.
    
    Let $\rt = \homeo \circ \tarFn$, $\qt = \homeo \circ \conFn$, and $\lt = \homeo \circ \terFn$.
    Since $\robGRA[\graData] = \homeo^{-1} \circ \robGRA[\rt, \qt, \lt; \tFin]$, it suffices to show that $\robGRA[\rt, \qt, \lt; \tFin]$ is continuous.

    Fix $\tNot < \tFin$.
    Given $\xSig_1, \xSig_2 \in \xSigs$, we can use the standard inequality for the difference of two maximums (or equivalently minimums) to conclude that
    \begin{align*}
        &|\robGRA[\rt,\qt,\lt;\tFin](\xSig_1, \tVal) - \robGRA[\rt,\qt,\lt;\tFin](\xSig_2, \tVal)|\\
        &\le \max_{\tDum \in [\tVal, \tFin]} \max\{|\rt(\xSig_1(\tDum), \tDum) - \rt(\xSig_2(\tDum), \tDum)|, \\
        &\qquad\qquad\qquad~|\qt(\xSig_1(\tDum), \tDum) - \qt(\xSig_2(\tDum), \tDum)|,\\
        &\qquad\qquad\qquad~|\lt(\xSig_1(\tFin)) - \lt(\xSig_2(\tFin))|\}\\
        &\le \max_{\tDum \in [\tNot, \tFin]} \max\{|\rt(\xSig_1(\tDum), \tDum) - \rt(\xSig_2(\tDum), \tDum)|, \\
        &\qquad\qquad\qquad~~|\qt(\xSig_1(\tDum), \tDum) - \qt(\xSig_2(\tDum), \tDum)|,\\
        &\qquad\qquad\qquad~~|\lt(\xSig_1(\tFin)) - \lt(\xSig_2(\tFin))|\}.
    \end{align*}
    for all $\tVal \in [\tNot, \tFin]$.
    It thus follows from continuity of $\rt$, $\qt$, and $\lt$ that $\robGRA(\xSig, \tVal)$ is continuous in $\xSig$, locally uniformly in $\tVal$.
    On the other hand, it follows directly from \eqref{eqn:gra-definition} and continuity of $\tarFn$, $\conFn$, and $\terFn$ that $\robGRA[\rt,\qt,\lt;\tFin](\xSig, \cdot)$ is continuous for each $\xSig \in \xSigs$.
    Thus we can conclude that $\robGRA[\rt,\qt,\lt;\tFin]$ is continuous.
\end{proof}

\begin{proof}[Proof of Lemma \ref{lem:negated-gra}]
Let $(\graData)$, $\tilde{\tarFn}$, $\tilde{\conFn}$, and $\tilde{\terFn}$ be as in the statement of the lemma.
It suffices to show that
\begin{equation*}
    -\robGRA[\graData] = \robGRA[\tilde{\tarFn}, \tilde{\conFn}, \tilde{\terFn}; \tFin].
\end{equation*}
First, assume that $\tarFn \le \conFn$ and $\tarFn(\cdot, \tFin) \le \terFn(\cdot) \le \conFn(\cdot, \tFin)$.
In this case, $\tilde{\tarFn} = -\conFn$, $\tilde{\conFn} = -\tarFn$, and $\tilde{\terFn} = -\terFn$, so that we need only show
$$-\robGRA[\graData] = \robGRA[-\conFn, -\tarFn, -\terFn; \tFin].$$

($\le$) Fix $\xSig \in \xSigs$.
For each $\tVal \le \tFin$, let $\tarFn_\tVal = \tarFn(\xSig(\tVal), \tVal)$ and $\conFn_\tVal = \conFn(\xSig(\tVal), \tVal)$, and also let $\terFn_\tFin = \terFn(\xSig(\tFin))$.

For convenience, let $c = \robGRA[\graData](\xSig, \tVal)$.
It follows that $\min\{\terFn_\tFin, \min_{\tDumDum \in [\tVal, \tFin]} \conFn_\tDumDum\} \le c$ and $\max_{\tDum \in [\tVal, \tFin]} \min\{\tarFn_\tDum, \min_{\tDumDum \in [\tVal, \tDum]} \conFn_\tDumDum\} \le c$.
There are then two cases: (i) $\terFn_\tFin \le c$ and $\tarFn_\tDum \le c$ for all $\tDum \in [\tVal, \tFin]$, or (ii) for some $\tDum \in [\tVal, \tFin]$, we have $\conFn_\tDum \le c$ and $\tarFn_\tDumDum \le c$ for all $\tDumDum \in [\tVal, \tDum]$.
In either case, we have $\robGRA[-\conFn, -\tarFn, -\terFn; \tFin](\xSig, \tVal) \ge -c$.

($\ge$) This direction follows from applying the result for the opposite direction with $-\conFn$ in place of $\tarFn$, $-\tarFn$ in place of $\conFn$, and $-\terFn$ in place of $\terFn$.

For the general case (in which the above inequality relations on $\tarFn$, $\conFn$, and $\terFn$ do not hold), observe that
\begin{align*}
    &-\robGRA[\graData] \\
    &=-\robGRA[\min\{\tarFn, \conFn\}, \conFn, \min\{\conFn(\cdot, \tFin), \max\{\tarFn(\cdot,\tFin), \terFn(\cdot)\} \}; \tFin] \\
    &=\robGRA[-\conFn, -\min\{\tarFn, \conFn\},\\
    &\qquad\qquad-\min\{\conFn(\cdot, \tFin), \max\{\tarFn(\cdot,\tFin), \terFn(\cdot)\} \}; \tFin] \\
    &= \robGRA[\tilde{\tarFn}, \tilde{\conFn}, \tilde{\terFn}; \tFin].
\end{align*}
\end{proof}

\subsection{Proofs of Theorem \ref{thm:value-decomposition} and Corollary \ref{cor:dpp}}\label{sec:appendix-decomposition-proof}

For convenience, we will adopt the following notational convention throughout this sub-section.
Given an $\xVal \in \Rn$, $\tVal \in \R$, $\uSig \in \uSigs$, and $\dStrat \in \dStrats$, we will define $\traj{\xVal, \tVal}{\uSig, \dStrat} = \stratTraj$.

\begin{proof}[Proof of Theorem \ref{thm:value-decomposition}]
    \newcommand{\Jnot}{\rob_0}
    \newcommand{\Jtil}{\tilde{\rob}}
    \newcommand{\firstStrat}{\alpha}
    \newcommand{\secondStrat}[1]{\beta_{#1}}
    \newcommand{\composedStrat}{\dStrat}
    \newcommand{\firstSig}{\mathrm{v}}
    \newcommand{\secondSig}{\mathrm{w}}
    \newcommand{\firstTraj}{\traj{\xVal, \tVal}{\uSig, \firstStrat}}
    \newcommand{\secondTraj}{\traj{\xDum,\tFin}{\uSig, \secondStrat{\xDum}}}
    \newcommand{\switchSet}[1]{\mathcal{S}_{#1}}
    \newcommand{\switchTime}[1]{\tDum_{#1}}
    \newcommand{\switchState}[1]{\xDum_{#1}}
    \newcommand{\qSweep}[2]{\min_{\tDumDum \in [#1]} \conFn(#2(\tDumDum),\tDumDum)}
    
    \newcommand{\trajxtvd}{\traj{\xVal, \tVal}{\firstSig, \composedStrat}}
    \newcommand{\trajxtud}{\traj{\xVal, \tVal}{\uSig, \composedStrat}}
    \newcommand{\trajxtva}{\traj{\xVal, \tVal}{\firstSig, \firstStrat}}
    \newcommand{\trajySwd}{\traj{\xDum, \tFin}{\secondSig, \composedStrat}}
    \newcommand{\trajySud}{\traj{\xDum, \tFin}{\uSig, \composedStrat}}
    \newcommand{\trajswitchvb}{\traj{\vState, \tFin}{\firstSig, \secondStrat{\vState}}}

    \newcommand{\vTime}{\switchTime{\firstSig}}
    \newcommand{\vState}{\switchState{\firstSig}}
    \newcommand{\specTer}{\tilde{\terFn}}
    \newcommand{\specInst}{\tarFn, \conFn, \specTer; \tFin}

    \newcommand{\altStrat}{\gamma}
    \newcommand{\ang}[1]{\langle #1 \rangle}
    \newcommand{\trajySwg}{\traj{\xDum, \tFin}{\secondSig, \altStrat}}
    \newcommand{\trajySug}{\traj{\xDum, \tFin}{\uSig, \altStrat}}
    \newcommand{\trajxtvwd}{\traj{\xVal, \tVal}{\ang{\firstSig, \secondSig}, \composedStrat}}
    
    ($\le$)
    This direction proceeds by building a near-optimal composite strategy ($\composedStrat$) from a primary strategy ($\firstStrat$) and a family of secondary strategies ($\secondStrat{\xDumDum}$).

    Fix $\xVal \in \Rn$, $\tVal \le \tFin$, and $\ep > 0$.
    Select an adversary strategy $\firstStrat \in \dStrats$ such that
    \begin{equation*}
        \supu \robGRA[\specInst](\firstTraj, \tVal) \le \valGRA[\specInst](\xVal, \tVal) + \ep.
    \end{equation*}
    For each $\xDum \in \Rn$, choose the secondary strategy $\secondStrat{\xDum} \in \dStrats$ such that 
    \begin{equation*}
        \supu \Jnot(\secondTraj,\tFin) \le \val[\Jnot](\xDum, \tFin) + \ep.
    \end{equation*}
    For each $\uSig \in \uSigs$, we define the switch state as
    \begin{equation*}
        \switchState{\uSig} = \firstTraj(\tFin).
    \end{equation*}

    We define the composite strategy $\composedStrat: \uSigs \to \dSigs$ by for each $\uSig \in \uSigs$ letting $\composedStrat(\uSig) \in \dSigs$ be given by
    \begin{equation*}
        \composedStrat(\uSig)(\tDum) = \begin{cases}
            \firstStrat(\uSig)(\tDum) & \tDum \le \tFin \\
            \secondStrat{\switchState{\uSig}}(\uSig)(\tDum) & \tDum > \tFin.
        \end{cases}
    \end{equation*}
    
    We show that $\composedStrat$ is non-anticipative.
    First, suppose $\uSig_1(\tDumDum) = \uSig_2(\tDumDum)$ for a.e. $\tDumDum \le \tDum$, where $\tDum \le \tFin$.
    Then $\composedStrat(\uSig_1)(\tDumDum) = \firstStrat(\uSig_1)(\tDumDum) = \firstStrat(\uSig_2)(\tDumDum) = \composedStrat(\uSig_2)(\tDumDum)$ for a.e. $\tDumDum \le \tDum$.
    Next, suppose that $\uSig_1(\tDumDum) = \uSig_2(\tDumDum)$ for a.e. $\tDumDum \le \tDum$, where $\tDum > \tFin$.
    Then $\switchState{\uSig_1} = \traj{\xVal, \tVal}{\uSig_1, \firstStrat}(\tFin) = \traj{\xVal, \tVal}{\uSig_2, \firstStrat}(\tFin) = \switchState{\uSig_2}$.
    It follows that $\composedStrat(\uSig_1)(\tDumDum) = \firstStrat(\uSig_1)(\tDumDum) = \firstStrat(\uSig_2)(\tDumDum) = \composedStrat(\uSig_2)(\tDumDum)$ for a.e. $\tDumDum \le \tFin$ and $\composedStrat(\uSig_1)(\tDumDum) = \secondStrat{\switchState{\uSig_1}}(\uSig_1)(\tDumDum) = \secondStrat{\switchState{\uSig_2}}(\uSig_2)(\tDumDum) = \composedStrat(\uSig_2)(\tDumDum)$ for a.e. $\tDumDum \in (\tFin, \tDum]$.
    Thus $\composedStrat$ is indeed non-anticipative.
    
    Choose $\firstSig \in \uSigs$ such that
    \begin{equation*}
        \sup_{\uSig \in \uSigs} \rob(\trajxtud, \tVal) \le \rob(\traj{\xVal, \tVal}{\firstSig, \composedStrat}, \tVal) + \ep.
    \end{equation*}
    Observe that for all $\tDum \le \tFin$,
    \begin{equation*}
        \trajxtvd(\tDum) = \traj{\xVal, \tVal}{\firstSig, \firstStrat}(\tDum),
    \end{equation*}
    and for all $\tDum \ge \tFin$,
    \begin{equation*}
        \trajxtvd(\tDum) = \trajswitchvb(\tDum).
    \end{equation*}

    It follows that
    \begin{align*}
        &\val[\rob](\xVal, \tVal) \\
        &\le \rob(\traj{\xVal, \tVal}{\firstSig, \composedStrat}, \tVal) + \ep \\
        &=\max\Big\{\robGRA[\tarFn, \conFn, \terFn; \tFin](\trajxtvd, \tVal),\\
        &\qquad\min\big\{\rob_0(\trajxtvd, \tFin), \min_{\tDumDum \in [\tVal, \tFin]} \conFn(\trajxtvd(\tDumDum),\tDumDum) \big\} \Big\} + \ep\\
        &=\max\Big\{\robGRA[\tarFn, \conFn, \terFn; \tFin](\trajxtva, \tVal),\\
        &\qquad\min\big\{\rob_0(\trajswitchvb, \tFin), \min_{\tDumDum \in [\tVal, \tFin]} \conFn(\trajxtva(\tDumDum),\tDumDum) \big\} \Big\} + \ep\\
        &\le\max\Big\{\robGRA[\tarFn, \conFn, \terFn; \tFin](\trajxtva, \tVal),\\
        &\qquad\min\big\{\val[\Jnot](\vState, \tFin), \min_{\tDumDum \in [\tVal, \tFin]} \conFn(\trajxtva(\tDumDum),\tDumDum) \big\} \Big\} + 2\ep\\
        &=\max\Big\{\max_{\tDum \in [\tVal, \tFin]} \min\{\tarFn(\trajxtva(\tDum), \tDum), \min_{\tDumDum \in [\tVal, \tDum]} \conFn(\trajxtva(\tDumDum),\tDumDum)\},\\
        &\qquad \min\{\terFn(\trajxtva(\tFin)), \min_{\tDumDum \in [\tVal, \tFin]} \conFn(\trajxtva(\tDumDum),\tDumDum)\},\\
        &\qquad\min\big\{\val[\Jnot](\trajxtva(\tFin), \tFin), \min_{\tDumDum \in [\tVal, \tFin]} \conFn(\trajxtva(\tDumDum),\tDumDum) \big\} \Big\} + 2\ep\\
        &=\max\Big\{\max_{\tDum \in [\tVal, \tFin]} \min\{\tarFn(\trajxtva(\tDum), \tDum), \min_{\tDumDum \in [\tVal, \tDum]} \conFn(\trajxtva(\tDumDum),\tDumDum)\},\\
        &\qquad \min\{\specTer(\trajxtva(\tFin)), \min_{\tDumDum \in [\tVal, \tFin]} \conFn(\trajxtva(\tDumDum),\tDumDum)\}\Big\} + 2\ep\\
        &= \robGRA[\specInst](\trajxtva, \tVal) + 2\ep \\
        & \le \valGRA[\specInst](\xVal, \tVal) + 3\ep.
    \end{align*}

    ($\ge$) Fix $\xVal \in \Rn$, $\tVal \le \tFin$, and $\ep > 0$.
    Select an adversary strategy $\dStrat \in \dStrats$ such that
    \begin{equation*}
        \val[\rob](\xVal, \tVal) \ge \supu \rob(\traj{\xVal, \tVal}{\uSig, \dStrat}, \tVal) - \ep.
    \end{equation*}
    Choose $\firstSig \in \uSigs$ such that
    \begin{equation*}
        \robGRA[\specInst](\trajxtvd, \tVal) \ge \supu \robGRA[\specInst](\trajxtud, \tVal) - \ep.
    \end{equation*}

    For each $\uSig \in \uSigs$, let $\ang{\firstSig, \uSig}: \R \to \uVals$ be defined by
    \begin{equation*}
        \ang{\firstSig, \uSig}(\tDum) = \begin{cases}
        \firstSig(\tDum) & \tDum \le \tFin,\\
        \uSig(\tDum) & \tDum > \tFin.
        \end{cases}
    \end{equation*}
    Define the strategy $\altStrat \in \dStrats$ by for each $\uSig \in \uSigs$ setting $\altStrat(\uSig) = \dStrat(\ang{\firstSig, \uSig})$.
    Next let $\xDum = \trajxtvd(\tFin)$ and choose $\secondSig \in \uSigs$ such that
    \begin{equation*}
        \Jnot(\trajySwg,\tFin) \ge \supu \Jnot(\trajySug,\tFin) - \ep.
    \end{equation*}
    Then
    
    \begin{align*}
        &\val[\rob](\xVal, \tVal) \\
        &\ge \rob(\trajxtvwd, \tVal) - \ep \\
        &=\max\Big\{\robGRA[\tarFn, \conFn, \terFn; \tFin](\trajxtvd, \tVal),\\
        &\qquad\min\big\{\rob_0(\trajySwg, \tFin), \min_{\tDumDum \in [\tVal, \tFin]} \conFn(\trajxtvd(\tDumDum),\tDumDum) \big\} \Big\} - \ep\\
        &\ge\max\Big\{\robGRA[\tarFn, \conFn, \terFn; \tFin](\trajxtvd, \tVal),\\
        &\qquad\min\big\{\val[\rob_0](\trajxtvd(\tFin), \tFin), \min_{\tDumDum \in [\tVal, \tFin]} \conFn(\trajxtvd(\tDumDum),\tDumDum) \big\} \Big\} - 2\ep\\
        &=\robGRA[\tarFn, \conFn, \specTer; \tFin](\trajxtvd, \tVal) - 2\ep\\
        &\ge\valGRA[\tarFn, \conFn, \specTer; \tFin](\xVal, \tVal) - 3\ep.\qedhere
    \end{align*}
\end{proof}
 
\begin{proof}[Proof of Corollary \ref{cor:dpp}]
    First, observe that
    \begin{align*}
    \robGRA[\tarFn, \conFn,& \terFn; \tFin](\xSig, \tVal) =\\ 
    \max\Big\{&\robGRA[\tarFn, \conFn, \min\{\conFn(\cdot,\tFinDum), \tarFn(\cdot,\tFinDum)\}; \tFinDum](\xSig, \tVal),\\
    &\min\big\{\robGRA[\tarFn, \conFn, \terFn; \tFin](\xSig, \tFinDum), \min_{\tDumDum \in [\tVal, \tFinDum]} \conFn(\xSig(\tDumDum),\tDumDum) \big\} \Big\}
    \end{align*}
    for each $\xSig \in \xSigs$ and $\tVal \le \tFinDum$.
    Let $\rob = \robGRA[\tarFn, \conFn, \terFn; \tFin]|_{\xSigs \tms \RleS}$ and $\rob_0 = \robGRA[\tarFn, \conFn, \terFn; \tFin]$,
    so that the above equation becomes
    \begin{align*}
    \rob(\xSig, \tVal) =\max\Big\{&\robGRA[\tarFn, \conFn, \min\{\conFn(\cdot,\tFinDum), \tarFn(\cdot,\tFinDum)\}; \tFinDum](\xSig, \tVal),\\
    &\min\big\{\rob_0(\xSig, \tFinDum), \min_{\tDumDum \in [\tVal, \tFinDum]} \conFn(\xSig(\tDumDum),\tDumDum) \big\} \Big\}.
    \end{align*}
    Applying Theorem \ref{thm:value-decomposition} (with the GRA instance $(\tarFn,\conFn,\min\{\conFn(\cdot,\tFinDum), \tarFn(\cdot,\tFinDum)\};\tFinDum)$ in place of $(\graData)$) and noticing that 
    $$\max\{\min\{\conFn(\cdot,\tFinDum), \tarFn(\cdot,\tFinDum)\}, \val[\rob_0](\cdot, \tFinDum)\} = \val[\rob_0](\cdot, \tFinDum),$$
    then gives the desired result.
    
\end{proof}

\subsection{Proof of Theorem \ref{thm:comparison-principle}}\label{sec:appendix-comparison-principle}

Before we prove Theorem \ref{thm:comparison-principle}, we will need the following lemma, which uses an adaptation of the proof technique in Lemma II.1.8 in \cite{Bardi-Dolcetta-Optimal-Control}.  

\begin{lemma}\label{lem:hjr-pde-bounds}
Let $m \in \N$ and let $\hjDomain \sbs \Rm$ be open.
Suppose $\tarFn: \hjDomain \to \extended$ is continuous with $\tarFn < \ii$, $\conFn: \hjDomain \to \extended$ is continuous with $\conFn > -\ii$, and $\hjFn: \hjDomain \tms \R \tms \Rm \to \R$ is continuous.
Consider the HJ-PDE 
\begin{align}\label{eqn:omega-hjr-pde}
    -\min\{& \conFn(\hjInd) - \hjDep, \max\{ \tarFn(\hjInd) - \hjDep, \hjFn(\hjInd, \hjDep, \nabla \hjDep)\}\} = 0,
    \quad \hjInd \in \hjDomain.
\end{align}
\begin{itemize}
    \item
    If $\hjSol_-:\Omega \to \R$ is a viscosity sub-solution of \eqref{eqn:omega-hjr-pde}, then $\hjSol_- \le \conFn$.
    \item
    If $\hjSol_+: \Omega \to \R$ is a viscosity super-solution of \eqref{eqn:omega-hjr-pde}, then $\hjSol_+ \ge \min\{\conFn, \tarFn\}$.
\end{itemize}
\end{lemma}
\begin{proof}
    Let $\bar{\hjInd} \in \hjDomain$, and suppose $\hjSol_-(\bar{\hjInd}) > \conFn(\bar{\hjInd})$.
    Choose a compact neighborhood $E \ssbs \hjDomain$ of $\bar{\hjInd}$ such that  $\conFn(\hjInd) < \hjSol_-(\bar{\hjInd})$ for all $\hjInd \in E$.
    
     For each $\ep > 0$, let $\hjTest_\ep: \hjDomain \to \R$ be given by $\hjTest_\ep(\hjInd) = \frac{1}{\ep}\|\hjInd - \bar{\hjInd}\|^2$ and let $\hjInd_\ep \in \argmax_{\hjInd \in E} (\hjSol_-(\hjInd) - \phi_\ep(\hjInd))$.
     Then ($\hjSol_- - \phi_\ep)(\hjInd_\ep) \ge \hjSol_-(\bar{\hjInd})$, so that $\|\hjInd_\ep - \bar{\hjInd}\| \le \sqrt{\ep} \lf(\max_{\hjInd \in E} \hjSol_-(\hjInd) - \hjSol_-(\bar{\hjInd}) \rg)$.
     It follows that we can choose $\ep^* > 0$ sufficiently small that $(\hjSol_- - \phi_{\ep^*})|_E$ attains its maximum at a point $\hjInd_0 := \hjInd_{\ep^*}$ in the interior of $E$.

     Letting $\hjTest_0 := \hjTest_{\ep^*}$, by definition of a viscosity sub-solution, we have
     \begin{align*}
         \min\{ &\conFn(\hjInd_0) - \hjSol_-(\hjInd_0), \max\{ \tarFn(\hjInd_0) - \hjSol_-(\hjInd_0), \\
         &\hjFn(\hjInd_0, \hjSol_-(\hjInd_0), \nabla \hjTest_0(\hjInd_0))\}\} \ge 0,
     \end{align*}
     so that $\hjSol_-(\hjInd_0) \le \conFn(\hjInd_0)$, producing the desired contradiction.

     The statement for super-solutions is shown analogously.
     
\end{proof}

\newcommand{\cone}{\mathcal{C}}
\newcommand{\base}{\mathcal{B}}
\begin{lemma}\label{lem:cone-of-dependence}(Cone of dependence)
Let $\xNot \in \Rn$, $\tNot \in \R$, $R > 0$, and $C > 0$.
Define the cone $\cone = \{(\xVal,\tVal) \in \Rn \tms \R \mid \tNot - R < \tVal < \tNot, \|\xVal - \xNot\| < C (R - \tNot + \tVal)\}$
and its base $\base = B(\xNot; CR) \tms \{\tNot\}$.
Suppose $\hamiltonian: (\cone \cup \base) \tms \Rn \to \R$ is continuous and satisfies
\begin{enumerate}
\item[(\textit{i})] $|\hamiltonian(\xVal, \tVal, \costate) - \hamiltonian(\xVal, \tVal, \coDum)| \le C \|\costate - \coDum\|$ for all $(\xVal, \tVal) \in \cone$ and $\costate, \coDum \in \Rn$,
\item[(\textit{ii})] there exists an $M > 0$ such that $|\hamiltonian(\xVal, \tVal, \costate) - \hamiltonian(\xDum, \tVal, \costate)| \le M\|\xVal - \xDum\|(1 + \|\costate\|)$ for all $(\xVal, \tVal) \in \cone$, $\xDum \in B(\xNot; C(R - (\tNot - \tVal)))$, and $\costate \in \Rn$.
\end{enumerate}
Let $\tarFn,\conFn:\cone \cup \base \to \extended$ be continuous, with $\tarFn < +\ii$ and $\conFn > -\ii$.
Consider the HJ-PDE
\begin{align}\label{eqn:cone-hjr-pde}
    -\min\{& \conFn(\xVal,\tVal) - \hjDep, \nonumber \\
    &\max\{ \tarFn(\xVal,\tVal) - \hjDep, \nonumber \\
    &\qquad~ \ppt \hjDep + \hamiltonian(\xVal, \tVal, \nabla_x \hjDep) \}\} = 0,
    \quad (\xVal, \tVal) \in \cone.
\end{align}
Let $\hjSol_- : \cone \cup \base \to \R$ be upper semi-continuous, and let $\hjSol_+: \cone \cup \base \to \R$ be lower semi-continuous.
If $\hjSol_-|_{\cone}$ is a viscosity sub-solution of \eqref{eqn:cone-hjr-pde}, $\hjSol_+|_{\cone}$ is a viscosity super-solution of \eqref{eqn:cone-hjr-pde}, and $\hjSol_-|_{\base} \le \hjSol_+|_{\base}$, then $\hjSol_- \le \hjSol_+$.
\end{lemma}
\begin{proof}
\newcommand{\scone}{\mathcal{C}_0}
\newcommand{\clScone}{\bar{\mathcal{C}}_0}
\newcommand{\vsub}{\hjSol_-}
\newcommand{\vsup}{\hjSol_+}
\newcommand{\minrq}[1]{\min\{\tarFn(#1), \conFn(#1) \}}
\newcommand{\ti}{\tilde}
\newcommand{\ang}[1]{\langle #1 \rangle}

\newcommand{\Ph}{\Phi}
\newcommand{\Phep}{\Phi_\ep}
\newcommand{\Phpar}{\Phi_{\ep,\eta}}

\newcommand{\xpar}{\xVal_{\ep,\eta}}
\newcommand{\ypar}{\xDum_{\ep,\eta}}
\newcommand{\tpar}{\tVal_{\ep,\eta}}
\newcommand{\spar}{\tDum_{\ep,\eta}}
\newcommand{\Mpar}{M_{\ep,\eta}}

\newcommand{\xep}{\xVal_{\ep}}
\newcommand{\yep}{\xDum_{\ep}}
\newcommand{\tep}{\tVal_{\ep}}
\newcommand{\sep}{\tDum_{\ep}}

\newcommand{\xnot}{\xVal_*}
\newcommand{\ynot}{\xDum_*}
\newcommand{\tnot}{\tVal_*}
\newcommand{\snot}{\tDum_*}

\newcommand{\xTil}{\ti{\xVal}}
\newcommand{\tTil}{\ti{\tVal}}

Assume by way of contradiction that
\begin{align*}
    &\vsub(\xTil, \tTil) - \vsup(\xTil, \tTil) \ge \theta\\
    &\|\xTil - \xNot\| \le C(R - \tNot + \tTil) - 3\theta
\end{align*}
for some $\theta \in (0,R)$ and $(\xTil, \tTil) \in \cone$.
Choose $R_0 = R - C^{-1}\theta$, so that
\begin{equation*}
    \|\xTil - \xNot\| \le C(R_0 - \tNot + \tTil) - 2\theta.
\end{equation*}
Let $\scone := \{(\xVal,\tVal) \in \cone \mid \tNot - R_0 < \tVal < \tNot, \|\xVal - \xNot\| < C (R_0 - \tNot + \tVal)\}$, and note that $(\xTil, \tTil) \in \scone$.

Set $m = \max_{(\xVal, \tVal),(\xDum, \tDum) \in \clScone}(\vsub(\xVal, \tVal) - \vsup(\xDum, \tDum))$.
Let $h: \R \to \R$ be continuously differentiable and satisfy $h' \le 0$, $h(a) = 0$ for $a \le -\theta$, $h(a) = -m$ for $a \ge 0$.
For convenience, let  
$$\ang{\xVal} := \sqrt{\|\xVal - \xNot\|^2 + \theta^2}$$
for each $\xVal \in \Rn$.

First, define $\Ph:\clScone^2 \to \R$ by
\begin{align*}
    &\Ph(\xVal, \tVal, \xDum, \tDum) = \vsub(\xVal, \tVal) - \vsup(\xDum, \tDum) - \frac{\theta(2\tNot - \tVal - \tDum)}{4(\tNot - \ti{\tVal})} \\
    &+ h(\ang{\xVal} - C(R_0 - \tNot + \tVal)) + h(\ang{\xDum} - C(R_0 - \tNot + \tDum)).
\end{align*}
Next, for each $\ep > 0$, define $\Phep:\clScone^2 \to \R$ by $$\Phep(\xVal, \tVal, \xDum, \tDum) = \Ph(\xVal, \tVal, \xDum, \tDum) - \frac{\|\xVal-\xDum\|^2}{2 \ep}.$$
Finally, for each $\ep > 0$ and $\eta > 0$, let $\Phpar:\clScone^2 \to \R$ be given by
$$\Phpar(\xVal, \tVal, \xDum, \tDum) = \Phep(\xVal, \tVal, \xDum, \tDum) - \frac{(\tVal-\tDum)^2}{2 \eta},$$
choose 
$$(\xpar, \tpar, \ypar, \spar) \in \argmax_{(\xVal, \tVal, \xDum, \tDum) \in \clScone^2} \Phpar(\xVal, \tVal, \xDum, \tDum),$$
and let $\Mpar = \Phpar(\xpar, \tpar, \ypar, \spar)$.

\newcommand{\slice}{\mathcal{S}}
\newcommand{\Mep}{M_\ep}

\newcommand{\Phepj}{\Ph_{\ep, \eta_j}}
\newcommand{\xepj}{\xVal_{\ep, \eta_j}}
\newcommand{\tepj}{\tVal_{\ep, \eta_j}}
\newcommand{\yepj}{\xDum_{\ep, \eta_j}}
\newcommand{\sepj}{\tDum_{\ep, \eta_j}}

(Step 1) Let $\slice = \{(\xVal, \tVal, \xDum, \tDum) \in \clScone^2 \mid \tVal = \tDum \}$.
We show that for each $\ep > 0$, we have $\Phep(\xpar, \tpar, \ypar, \spar) \to \max_{(\xVal, \tVal, \xDum, \tDum) \in \slice} \Phep(\xVal, \tVal, \xDum, \tDum) =: \Mep$ and $(\tpar - \spar)^2/\eta \to 0$ as $\eta \to 0^+$.

Fix $\ep > 0$.
First, note that for each $\eta > 0$, we have $\Mpar \ge \Phpar(\xTil, \tTil, \xTil, \tTil) \ge \theta / 2$ and $\Phep(\xpar, \tpar, \ypar, \spar) \le m$.
Thus 
\begin{align*}
    \theta / 2 \le \liminf_{\eta \to 0^+} \Mpar \le m - \limsup_{\eta \to 0^+} \frac{(\tpar - \spar)^2}{2 \eta},
\end{align*}
 so that $|\tpar - \spar| \to 0$ as $\eta \to 0^+$.

Now, observe that
\begin{align*}
    \Mep &\le \Phpar(\xpar, \tpar, \ypar, \spar)\\
    &= \Phep(\xpar, \tpar, \ypar, \spar) - \frac{|\tpar - \spar|^2}{2\eta}\\
    &\le \Phep(\xpar, \tpar, \ypar, \spar)
\end{align*} for each $\eta > 0$.
First, observe that $\Mep \le \Phep(\xpar, \tpar, \ypar, \spar)$ for each $\eta > 0$, so that $\Mep \le \liminf_{\eta \to 0^+} \Phep(\xpar, \tpar, \ypar, \spar)$.
Now choose a sequence $\eta_j > 0$ such that $\eta_j \to 0$ and 
$\lim_{j \to \ii} \Phep(\xepj, \tepj, \yepj, \sepj) = \limsup_{\eta \to 0^+} \Phep(\xpar, \tpar, \ypar, \spar)$.
By restricting to a subsequence if necessary, we can assume $(\xepj, \tepj, \yepj, \sepj)$ converges to some $(\xVal^*, \tVal^*, \xDum^*, \tVal^*) \in \clScone^2$.
Then $\limsup_{\eta \to 0^+} \Phep(\xpar, \tpar, \ypar, \spar) \le \Phep(\xVal^*, \tVal^*, \xDum^*, \tVal^*) \le \Mep$, showing the claim.

\newcommand{\etaepj}{\eta_{\ep,j}}
\newcommand{\xepetaj}{\xVal_{\ep, \etaepj}}
\newcommand{\tepetaj}{\tVal_{\ep, \etaepj}}
\newcommand{\yepetaj}{\xDum_{\ep, \etaepj}}
\newcommand{\sepetaj}{\tDum_{\ep, \etaepj}}
(Step 2) For each $\ep > 0$, choose a sequence $\etaepj > 0$ such that $\etaepj \to 0^+$ and $(\xepetaj, \tepetaj, \yepetaj, \sepetaj)$ converges to some $(\xep, \tep, \yep, \sep) \in \clScone^2$.
By Step 1, we have $\tep = \sep$ and $\Mep = \lim_{j \to \ii} \Phep(\xepetaj, \tepetaj, \yepetaj, \sepetaj) \le \Phep(\xep, \tep, \yep, \tep) \le \Mep$, so that $\Phep(\xep, \tep, \yep, \tep) = \Mep$ for each $\ep > 0$.

\newcommand{\proj}{\mathcal{P}}
\newcommand{\Mnot}{M_*}
(Step 3)
Let $\proj = \{(\xVal, \tVal, \xDum, \tDum) \in \clScone^2 \mid \tVal = \tDum, \xVal = \xDum \}$.
It follows as in Step 1 that $\Ph(\xep, \tep, \yep, \tep) \to \max_{(\xVal, \tVal, \xDum, \tDum) \in \proj} \Ph(\xVal, \tVal, \xDum, \tDum) =: \Mnot$ and $\| \xep - \yep \|^2 / \ep \to 0$ as $\ep \to 0^+$.

\newcommand{\epI}{\ep_i}
\newcommand{\xI}{\xVal_{\ep_i}}
\newcommand{\tI}{\tVal_{\ep_i}}
\newcommand{\yI}{\xDum_{\ep_i}}
\newcommand{\sI}{\tDum_{\ep_i}}
(Step 4) Choose a sequence $\epI > 0$ such that $\epI \to 0$ and $(\xI, \tI, \yI, \tI)$ converges to some $(\xnot, \tnot, \ynot, \tnot) \in \clScone^2$, respectively.
It follows from Step 3 that $\xnot = \ynot$ and $\Ph(\xnot, \tnot, \xnot, \tnot) = \Mnot$.

(Step 5) We show that $(\xnot, \tnot) \in \scone$.
First, observe that $\Mnot \ge \vsub(\xTil, \tTil) - \vsup(\xTil, \tTil) - \frac{\theta}{2} - 2 h(\ang{\xTil} - C(R_0 - \tNot + \tTil)) \ge \theta/2$.

Now, suppose that $\tnot = \tNot$.
Then $\Mnot \le \vsub(\xnot, \tnot) - \vsup(\xnot, \tnot) = \vsub(\xnot, \tNot) - \vsup(\xnot, \tNot) \le 0$, producing a contradiction.

Next, instead suppose that $\tnot < \tNot$ and $(\xnot, \tnot) \in \partial \scone$.
Then $\|\xNot - \xnot\| = C(R_0 - \tNot + \tnot)$.
Thus $\Mnot < \vsub(\xnot, \tnot) - \vsup(\xnot, \tnot) + 2h(0) \le m - 2m < 0$, again producing a contradiction and showing the claim.

\newcommand{\xk}{\xVal_{\ep_k,\eta_k}}
\newcommand{\tk}{\tVal_{\ep_k,\eta_k}}
\newcommand{\yk}{\xDum_{\ep_k,\eta_k}}
\newcommand{\sk}{\tDum_{\ep_k,\eta_k}}
\newcommand{\Mk}{M_{\ep_k,\eta_k}}
(Step 6) We show that for all sufficiently small $\ep > 0$ and $\eta > 0$, we have $\vsup(\ypar, \spar) < \conFn(\ypar, \spar)$ and $\minrq{\xpar, \tpar} < \vsub(\xpar, \tpar)$. 
First observe that for each $\ep > 0$ and $\eta > 0$, we have $\vsub(\xpar, \tpar) - \vsup(\ypar, \spar) \ge \Mnot \ge \theta / 2$.

Suppose we can choose sequences $\ep_k > 0$ and $\eta_k > 0$ such that $\ep_k \to 0$, $\eta_k \to 0$, and $\vsup(\yk, \sk) \ge \conFn(\yk, \sk)$.
Following the analogous argument in Step 1, we have $\|\xk - \yk\| \to 0$ and $|\tk - \sk| \to 0$, since
\begin{align*}
    &\theta / 2 \le \liminf_{k \to \ii} \Mk\\
    &\le m - \limsup_{k \to \ii} \lf(\frac{(\tk - \sk)^2}{2\eta_k} + \frac{(\xk - \yk)^2}{2\ep_k}\rg).
\end{align*}
By passing to a subsequence if necessary, we can assume $(\xk, \tk, \yk, \sk)$ converges to some $(\xVal^*, \tVal^*, \xVal^*, \tVal^*) \in \clScone^2$.
It is the case that $\tVal^* < \tNot$, for otherwise we would have
\begin{align*}
  \theta / 2 &\le \limsup_{k \to \ii} \lf(\vsub(\xk, \tk) - \vsup(\yk, \sk)\rg) \\
  &\le \vsub(\xVal^*, \tNot) - \vsup(\xVal^*, \tNot) \\
  &\le 0.
\end{align*}

By assumption,
\begin{align*}
    \conFn(\yk, \sk) &\le \vsup(\yk, \sk) \\
    &\le \vsub(\xk, \tk) - \frac{\theta}{2}
\end{align*}
for each $k \in \N$.
By taking the superior limit as $k \to \ii$, we get $\conFn(\xVal^*, \tVal^*) \le \vsub(\xVal^*, \tVal^*) - \frac{\theta}{2}$.
However, by Lemma \ref{lem:hjr-pde-bounds}, $\conFn(\xVal^*, \tVal^*) \ge \vsub(\xVal^*, \tVal^*)$, providing the desired contradiction.

Thus we indeed have $\vsup(\ypar, \spar) < \conFn(\ypar, \spar)$.
A symmetric argument shows the claim for $\vsub$.

\newcommand{\xij}{\bar{\xVal}_{i,j}}
\newcommand{\tij}{\bar{\tVal}_{i,j}}
\newcommand{\yij}{\bar{\xDum}_{i,j}}
\newcommand{\sij}{\bar{\tDum}_{i,j}}
\newcommand{\etaij}{\eta_{\ep_i,j}}
\newcommand{\xib}{\bar{\xVal}_i}
\newcommand{\tib}{\bar{\tVal}_i}
\newcommand{\yib}{\bar{\xDum}_i}
\newcommand{\sib}{\bar{\tDum}_i}
\newcommand{\aij}{\costate_{i,j}}
\newcommand{\aib}{\costate_i}
\newcommand{\anot}{\costate, a_*}
\newcommand{\bij}{\coDum_{i,j}}
\newcommand{\bib}{\coDum_i}
\newcommand{\bnot}{\coDum_*}
\newcommand{\Xij}{X_{i,j}}
\newcommand{\Yij}{Y_{i,j}}
\newcommand{\Xib}{X_i}
\newcommand{\Yib}{Y_i}
\newcommand{\Xnot}{X_*}
\newcommand{\Ynot}{Y_*}
(Step 7)
For convenience, let $\xij = \xVal_{\ep_i, \eta_{\ep_i, j}}$, $\tij = \tVal_{\ep_i, \eta_{\ep_i, j}}$, $\yij = \xDum_{\ep_i, \eta_{\ep_i, j}}$, and $\sij = \tDum_{\ep_i, \eta_{\ep_i, j}}$ for each $i,j \in \N$.
Similarly, let $\xib = \xVal_{\ep_i}$, $\tib = \tVal_{\ep_i}$, $\yib = \xDum_{\ep_i}$, and $\sib = \tDum_{\ep_i}$ for each $i \in \N$.
By Steps 1-6, we can assume (by passing to subsequences) that $(\xij, \tij, \yij, \sij) \in \scone^2$, $(\xij, \sij, \yij, \tij) \in \scone^2$, and 
\begin{align}
    &\minrq{\xij,\tij} < \vsub(\xij, \tij),\label{proof:cone-lemma-rqlv}\\
    &\conFn(\yij, \sij) > \vsup(\yij, \sij)\label{proof:cone-lemma-qgv}.
\end{align}

For each $i,j \in \N$, let 
\begin{align*}
    \Xij &= \ang{\xij} - C(R_0 - \tNot + \tij),\\
    \Yij &= \ang{\yij} - C(R_0 - \tNot + \sij),\\
    \aij &= \frac{\xij - \yij}{\ep_i} - h'(\Xij) \frac{\xij - \xNot}{\ang{\xij}},\\
    \bij &= \frac{\xij - \yij}{\ep_i} + h'(\Yij) \frac{\yij - \xNot}{\ang{\yij}},
\end{align*}
and define the test functions $\phi_{i,j}, \psi_{i,j} : \cone \to \R$ by
\begin{align*}
    \phi_{i,j}(\xVal, \tVal) =& \frac{\|\xVal - \yij\|^2}{2\ep_i} + \frac{(\tVal - \sij)^2}{2\etaij} - \theta_0 \tVal\\
    &- h(\ang{\xVal} - C(R_0 - \tNot + \tVal)),
\end{align*}
and
\begin{align*}
    \psi_{i,j}(\xDum, \tDum) =& -\frac{\|\xij - \xDum\|^2}{2 \ep_i} - \frac{(\tij - \tDum)^2}{2 \etaij} + \theta_0\tDum\\
    &+ h(\ang{\xDum} - C(R_0 - \tNot + \tDum)),
\end{align*}
where $\theta_0 = \theta / 4(\tNot - \tilde{\tVal}) > 0$.
We claim that
\begin{align}
\ppt \phi_{i,j}(\xij, \tij) + \hamiltonian(\xij, \tij, \nabla_\xVal \phi_{i,j}(\xij, \tij)\}\} &\ge 0, \label{proof:wwts-phi}\\
\pps \psi_{i,j}(\yij, \sij) + \hamiltonian(\yij, \sij, \nabla_\xDum \psi_{i,j}(\yij, \sij) \}\} &\le 0. \label{proof:wwts-psi}
\end{align}
 
Fix $i,j \in \N$.
First, observe that $\vsub - \phi_{i,j}$ has a maximum at $(\xij, \tij)$.
By the definition of a viscosity sub-solution,
\begin{align*}
    -&\min\{\conFn(\xij, \tij) - \vsub(\xij, \tij),\\
    &\max\{\tarFn(\xij, \tij) - \vsub(\xij, \tij),\\
    &\ppt \phi_{i,j}(\xij, \tij) + \hamiltonian(\xij, \tij, \nabla_\xVal \phi_{i,j}(\xij, \tij)\}\} \le 0.
\end{align*}
The above inequality implies that $\conFn(\xij, \tij) - \vsub(\xij, \tij) \ge 0$ and
\begin{align*}
    \max&\{\tarFn(\xij, \tij) - \vsub(\xij, \tij),\\
    &\ppt \phi_{i,j}(\xij, \tij) + \hamiltonian(\xij, \tij, \nabla_\xVal \phi_{i,j}(\xij, \tij)\} \ge 0.
\end{align*}
But then $\tarFn(\xij,\tij) < \vsub(\xij, \tij)$ by \eqref{proof:cone-lemma-rqlv}, so that \eqref{proof:wwts-phi} follows.

Next, observe that $\vsup - \psi$ has a minimum at $(\yij, \sij)$.
By the definition of a viscosity super-solution,
\begin{align*}
    -&\min\{\conFn(\yij, \sij) - \vsup(\yij, \sij),\\
    &\max\{\tarFn(\yij, \sij) - \vsup(\yij, \sij),\\
    &\quad \pps \psi_{i,j}(\yij, \sij) + \hamiltonian(\yij, \sij, \nabla_\xDum \psi_{i,j}(\yij, \sij)\}\} \ge 0.
\end{align*}
The above inequality implies that $\conFn(\yij, \sij) - \vsup(\yij, \sij) \le 0$ or 
\begin{align*}
    \max&\{\tarFn(\yij, \sij) - \vsup(\yij, \sij),\\
    &\pps \psi_{i,j}(\yij, \sij) + \hamiltonian(\yij, \sij, \nabla_\xDum \psi_{i,j}(\yij, \sij)\} \le 0.
\end{align*} 
But since the former condition cannot hold by \eqref{proof:cone-lemma-qgv}, it must be the case that the latter does.
This immediately implies \eqref{proof:wwts-psi}.

(Step 8)
We now derive a contradiction.
Fix $i,j \in \N$.
Note that \eqref{proof:wwts-phi}-\eqref{proof:wwts-psi} are equivalent to
\begin{align*}
    &\frac{\tij - \sij}{\etaij} + Ch'(\Xij) + \hamiltonian\lf(\xij, \tij, \aij \rg) \ge \theta_0,\\
    &\frac{\tij - \sij}{\etaij} - Ch'(\Yij) + \hamiltonian\lf(\yij, \sij, \bij \rg) \le - \theta_0.
\end{align*}
Thus,
\begin{align*}
    2\theta_0 \le& Ch'(\Xij) + Ch'(\Yij) \\
    &+ \hamiltonian(\xij, \tij, \aij) - \hamiltonian(\yij, \sij, \bij)\\
    \le& Ch'(\Xij) + Ch'(\Yij) \\
    &+ \hamiltonian(\xij, \tij, \aij) - \hamiltonian(\xij, \tij, \bij)\\
    &+ \hamiltonian(\xij, \tij, \bij) - \hamiltonian(\yij, \tij, \bij)\\
    &+ \hamiltonian(\yij, \tij, \bij) - \hamiltonian(\yij, \sij, \bij).
\end{align*}
Now
\begin{align*}
    &|\hamiltonian(\xij, \tij, \aij) - \hamiltonian(\xij, \tij, \bij)| \\
    &\le C\|\aij - \bij\|\\
    &\le C|h'(\Xij)| \lf\| \frac{\xij - \xNot}{\ang{\xij}}\rg\| + C|h'(\Yij)| \lf\| \frac{\yij - \xNot}{\ang{\yij}}\rg\|\\
    &\le C|h'(\Xij)| + C|h'(\Yij)|
\end{align*}
and
\begin{align*}
    &|\hamiltonian(\xij, \tij, \bij) - \hamiltonian(\yij, \tij, \bij)| \\
    &\le M\|\xij - \yij\|\lf(1 + \|\bij\| \rg)\\
    &\le M\|\xij - \yij\|\lf(1 + \frac{\|\xij - \yij\|}{\ep_i} + |h'(\Yij)|\rg).
\end{align*}
Thus
\begin{align*}
    2\theta_0 \le& M\|\xij - \yij\|\lf(1 + \frac{\|\xij - \yij\|}{\ep_i} + |h'(\Yij)|\rg)\\
    & +\hamiltonian(\yij, \tij, \bij) - \hamiltonian(\yij, \sij, \bij).
\end{align*}

Taking the limit as $j \to \ii$ gives
\begin{align*}
    2\theta_0 \le M \|\xib - \yib\|\lf(1 + \frac{\|\xib - \yib\|}{\ep_i} + |h'(\Yib)| \rg), 
\end{align*}
where $Y_i := \ang{\yib} - C(R_0 - \tNot + \sib)$.
Letting $i \to \ii$ then gives $2\theta_0 \le 0$, producing the desired contradiction.
\end{proof}

\begin{proof}[Proof of Theorem \ref{thm:comparison-principle}]
Assume the theorem hypotheses hold.
Fix $\ep \in (0, \tFin - \tInit)$, and let $R = 2(\tFin - \tInit - \ep) / N$, where $N \in \N$ is sufficiently large that $R < \min\{1/K, \ep\}$.
Let $\tDumDum_j = \tFin - \frac{1}{2}Rj$ for each $j = 0,\dots,N+1$.

Fix $j \in \{0, \dots, N-1\}$.
Suppose $\hjSol_- \le \hjSol_+$ on $\Rn \tms [\tDumDum_{j}, \tFin]$.
We show that $\hjSol_- \le \hjSol_+$ on $\Rn \tms [\tDumDum_{j+1}, \tFin]$.

Fix $\xVal_0 \in \Rn$, and let $\tVal_0 = \tDumDum_j$.
Set $C = K(\|\xVal_0\| + 1) / (1 - KR)$, and
define the cone $\cone := \{(\xVal, \tVal) \in \Rn \tms \R \mid \tVal_0 - R < \tVal < \tVal_0, \|\xVal - \xVal_0\| < C(R - (\tVal_0 - t))\}$.
For all $(\xVal,\tVal) \in \cone$ and $\costate, \coDum \in \Rn$, we have
\begin{align*}
    |\hamiltonian(\xVal, \tVal, \costate) - \hamiltonian(\xVal, \tVal, \coDum)| &\le K (\|\xVal\| + 1) \|\costate - \coDum \| \\
    &\le K (\|\xVal_0\| + CR + 1) \|\costate - \coDum \|\\
    &\le C \|\costate - \coDum \|.
\end{align*}
By letting $\compact = \overline{B(\xNot;CR)}$, we can by assumption choose an $M > 0$ such that $|\hamiltonian(\xVal, \tVal, \costate) - \hamiltonian(\xDum, \tVal, \costate)| \le M \|\xVal - \xDum\| (1 + \|\costate\|)$ for all $\xVal, \xDum \in \compact$ and $\tVal \in (\tInit, \tFin)$.

Thus, by Lemma \ref{lem:cone-of-dependence},
$\hjSol_- \le \hjSol_+$ on $\cone$.
Since $\xVal_0 \in \Rn$ was chosen arbitrarily, we can in fact conclude that $\hjSol_- \le \hjSol_+$ on $\Rn \tms (\tNot - R,\tNot] = \Rn \tms (\tDumDum_{j+2},\tDumDum_j]$.
But because we assumed that $\hjSol_- \le \hjSol_+$ on $\Rn \tms [\tDumDum_j,\tFin]$, we in fact have $\hjSol_- \le \hjSol_+$ on $\Rn \tms (\tDumDum_{j+2}, \tFin] \sps \Rn \tms [\tDumDum_{j+1}, \tFin]$.

By hypothesis, $\hjSol_- \le \hjSol_+$ on $\Rn \tms \{\tFin\} = \Rn \tms [\tDumDum_{0}, \tFin]$.
It then follows from induction on $j$ that $\hjSol_- \le \hjSol_+$ on $\Rn \tms [\tDumDum_{N}, \tFin] = \Rn \tms [\tInit + \ep, \tFin]$.
Since $\ep$ was chosen arbitrarily small, the result follows.
\end{proof}

\subsection{Proof of Theorem \ref{thm:dpe}}\label{sec:appendix-dpe}
\newcommand{\Lip}{\lip(\R, \Rn; L)}

\begin{definition}
    Let $\rob: \xSigs \tms \evalTimes \to \extended$, with $\evalTimes \sbs \R$.
    The time $\tFin \in \R$ is a \textbf{horizon} of $\rob$ if $\rob(\xSig_1, \cdot) = \rob(\xSig_2, \cdot)$ for all $\xSig_1, \xSig_2 \in \xSigs$ such that $\xSig_1(\tDum) = \xSig_2(\tDum)$ for every $\tDum \le \tFin$.
\end{definition}

In the remainder of this section, for each $L > 0$, we denote the set of all $L$-Lipschitz functions from $\R$ to $\Rn$ by $\Lip$, endowed with the subspace topology inherited from $\xSigs$.

\begin{lemma}\label{lem:limits}
    Let $\evalTimes \sbs \R$.
    For each $i \in \N$, let $\rob_i:\xSigs \tms \evalTimes \to \R$ be continuous and past-independent, and let $\rob:\xSigs \tms \evalTimes \to \R$ be continuous and past-independent as well.
    Suppose $\rob$ and all of the $\rob_i$ share some common horizon $\tFin \in \R$.
    If $\rob_i \to \rob$ locally uniformly on $\Lip \tms \evalTimes$ for each $L > 0$, then $\val[\rob_i] \to \val[\rob]$ locally uniformly.
\end{lemma}
\begin{proof}
    \newcommand{\Xs}{\mathcal{F}}
    Let $\compact_1 \ssbs \Rn$ and $ \compact_2 \ssbs \evalTimes$ both be compact.
    We can assume without loss of generality that $T > \tVal$ for each $\tVal \in \compact_2$.
    Choose $\tInit \in \R$ such that $\tInit \le \tVal$ for each $\tVal \in \compact_2$.
    By Assumptions \ref{assumption:compactness}-\ref{assumption:regularity}, there is a sufficiently large $L > 0$ such that
    $\standardTraj$ is $L$-Lipschitz on $[\tInit, \tFin]$ for each $\xVal \in \compact_1$, $\tVal \in \compact_2$, $\uSig \in \uSigs$, and $\dSig \in \dSigs$.
    Let $\Xs = \{\xSig \in \Lip \mid \xSig(\tVal) = \xVal \text{ for some } \xVal \in \compact_1 \text{ and } \tVal \in \compact_2 \}$.
    Then $\Xs$ is compact by the Arzel\`{a}-Ascoli Theorem.

    Let $\Xs_0 = \{\standardTraj \mid \xVal \in \compact_1, \tVal \in \compact_2, \uSig \in \uSigs, \dSig \in \dSigs\}$, and
    let $\pi: \Xs_0 \to \Xs$ be defined by
    \begin{equation*}
        \pi(\xSig)(\tVal) = \begin{cases}
            \xSig(\tInit) & \tVal < \tInit \\
            \xSig(\tVal) & \tInit \le \tVal \le \tFin \\
            \xSig(\tFin) & \tFin < \tVal.
        \end{cases}
    \end{equation*}
    It follows from past-independence and the uniform horizon bound for the $\rob_i$ and for $\rob$ that $\rob_i(\xSig, \tVal) = \rob_i(\pi(\xSig), \tVal)$ and $\rob(\xSig, \tVal) = \rob(\pi(\xSig), \tVal)$ for each $\xSig \in \Xs_0$ and each $\tVal \in \compact_2$.
    Thus

    \begin{align*}
    \sup_{\xVal \in \compact_1} &\sup_{\tVal \in \compact_2} |\val[\rob_i](\xVal, \tVal) - \val[\rob](\xVal, \tVal)|\\
    &\le \sup_{\xVal \in \compact_1} \sup_{\tVal \in \compact_2} \supu \supd| \rob_i(\standardTraj, \tVal) - \rob(\standardTraj, \tVal)|\\
    &= \sup_{\xVal \in \compact_1} \sup_{\tVal \in \compact_2} \supu \supd| \rob_i(\pi(\standardTraj), \tVal) - \rob(\pi(\standardTraj), \tVal)|\\
    &\le \max_{(\xSig, \tVal) \in \Xs \tms \compact_2} | \rob_i(\xSig, \tVal) - \rob(\xSig, \tVal)|.
    \end{align*}
    But $\max_{(\xSig, \tVal) \in \Xs \tms \compact_2} | \rob_i(\xSig, \tVal) - \rob(\xSig, \tVal)| \to 0$ as $i \to \ii$ by the assumption that $\rob_i \to \rob$ locally uniformly on $\Lip \tms \evalTimes$.
\end{proof}

\begin{lemma}\label{lem:homeomorphism}
    Let $\hamiltonian:\Rn \tms \R \tms \Rn \to \R$ be given by \eqref{eqn:hamiltonian} and $\homeo:\extended \to [-1,1]$ be given by \eqref{eqn:homeomorphism-def}.
    Let $\tInit, \tFin \in \R$ with $\tInit < \tFin$, and let $\tarFn, \conFn: \Rn \tms (\tInit,\tFin) \to \R$ be continuous with $\tarFn < +\ii$ and $\conFn > -\ii$.
    Then $\hjSol: \Rn \tms (\tInit, \tFin) \to \R$ is a viscosity sub-solution (resp. super-solution) of \eqref{eqn:hjr-pde} iff $w = \homeo \circ \hjSol$ is a viscosity sub-solution (resp. super-solution) of
    \begin{align}\label{eqn:hjr-pde-transformed}
    -\min\{& \homeo \circ \conFn(\xVal,\tVal) - w, \max\{ \homeo \circ \tarFn(\xVal,\tVal) - w, \nonumber \\
    &\ppt w + \hamiltonian(\xVal, \tVal, \nabla_x w) \}\} = 0,
    \quad \xVal \in \Rn,~\tVal \in (\tInit, \tFin).
    \end{align}
\end{lemma}
\begin{proof}
    \newcommand{\domain}{\Rn \tms (\tInit, \tFin)}
    ($\implies$) Suppose $\hjSol$ is a viscosity sub-solution of \eqref{eqn:hjr-pde}.
    Let $(\xNot, \tNot) \in \domain$ and $\hjTest:\domain \to \R$ be continuously differentiable.
    Suppose $w - \hjTest$ has a local maximum at $(\xNot, \tNot)$.
    Let $c = \hjTest(\xNot, \tNot) - w(\xNot, \tNot)$, let $\bar{\hjTest}:\Rn \tms (\tInit, \tFin) \to (-1,1)$ be a continuously differentiable function that is equal to $\hjTest - c$ in some neighborhood of $(\xNot, \tNot)$, and set $\psi = \homeo^{-1} \circ \bar{\hjTest}$. 
    
    Then $\hjSol(\xNot,\tNot) = \homeo^{-1}(w(\xNot, \tNot)) = \homeo^{-1}(\hjTest(\xNot, \tNot) - c) = \psi(\xNot, \tNot)$.
    Moreover, because $\homeo^{-1}$ is continuous and increasing, it then follows from the fact that $w \le \hjTest - c$ in some neighborhood of $(\xNot, \tNot)$ that $\hjSol \le \psi$ in some neighborhood of $(\xNot, \tNot)$.
    Thus $\hjSol - \psi$ has a local maximum at $(\xNot, \tNot)$.
    
    Therefore
    \begin{align*}
    0 \ge& -\min\{ \conFn(\xNot,\tNot) - \hjSol(\xNot,\tNot), \max\{ \tarFn(\xNot,\tNot) - \hjSol(\xNot,\tNot), \\
    &\ppt \psi(\xNot,\tNot) + \hamiltonian(\xNot, \tNot, \nabla_x \psi(\xNot,\tNot)) \}\} \\
    =& -\min\{ \conFn(\xNot,\tNot) - \hjSol(\xNot,\tNot), \max\{ \tarFn(\xNot,\tNot) - \hjSol(\xNot,\tNot), \\
    &\lf(\ppt \hjTest(\xNot,\tNot) + \hamiltonian(\xNot, \tNot, \nabla_x \hjTest(\xNot,\tNot))\rg) / \homeo'(\psi(\xNot,\tNot))  \}\},
    \end{align*}
    which is equivalent to
    \begin{align*}
    0 \ge& -\min\{ \homeo \circ \conFn(\xNot,\tNot) - w(\xNot,\tNot),\\
    &\qquad\max\{ \homeo \circ \tarFn(\xNot,\tNot) - w(\xNot,\tNot), \\
    &\qquad\ppt \hjTest(\xNot,\tNot) + \hamiltonian(\xNot, \tNot, \nabla_x \hjTest(\xNot,\tNot)) \}\}.
    \end{align*}
    An analogous argument proceeds for the super-solution case.
    ($\impliedby$) 
    Suppose $w$ is a viscosity sub-solution of \eqref{eqn:hjr-pde-transformed}.
    Let $(\xNot, \tNot) \in \domain$ and $\psi:\domain \to \R$ be continuously differentiable.
    Suppose $\hjSol - \psi$ has a local maximum at $(\xNot, \tNot)$.
    Let $c = \psi(\xNot, \tNot) - \hjSol(\xNot, \tNot)$, and let $\hjTest = \homeo \circ (\psi - c)$.

    Then $w(\xNot, \tNot) = \homeo(\hjSol(\xNot, \tNot)) = \homeo(\psi(\xNot, \tNot) - c) = \hjTest(\xNot, \tNot)$.
    Moreover, because $\homeo$ is continuous and increasing, it then follows from the fact that $\hjSol \le \psi - c$ in some neighborhood of $(\xNot, \tNot)$ that $w \le \hjTest$ in this same neighborhood.
    Thus $w - \hjTest$ has a local maximum at $(\xNot, \tNot)$.
    
    Therefore
    \begin{align*}
    0 \ge& -\min\{ \homeo \circ \conFn(\xNot,\tNot) - w(\xNot,\tNot), \\
    &\max\{ \homeo \circ \tarFn(\xNot,\tNot) - w(\xNot,\tNot), \\
    &\ppt \hjTest(\xNot,\tNot) + \hamiltonian(\xNot, \tNot, \nabla_x \hjTest(\xNot,\tNot)) \}\} \\
    =& -\min\{ \homeo \circ \conFn(\xNot,\tNot) - w(\xNot,\tNot), \\
    &\max\{ \homeo \circ \tarFn(\xNot,\tNot) - w(\xNot,\tNot), \\
    &(\ppt \psi(\xNot,\tNot) + \hamiltonian(\xNot, \tNot, \nabla_x \psi(\xNot,\tNot))) \homeo'(\hjSol(\xNot, \tNot)) \}\},
    \end{align*}
    which is equivalent to
    \begin{align*}
    0 \ge& -\min\{ \conFn(\xNot,\tNot) - \hjSol(\xNot,\tNot), \max\{ \tarFn(\xNot,\tNot) - \hjSol(\xNot,\tNot), \\
    &\qquad\ppt \psi(\xNot,\tNot) + \hamiltonian(\xNot, \tNot, \nabla_x \psi(\xNot,\tNot)) \}\}.
    \end{align*}
    
\end{proof}

The following lemma is essentially a restatement of Theorem 1 in \cite{fisac-chen-2015}, except that the theorem in this prior work assumes the system dynamics, target function, and constraint functions all are globally Lipschitz. 
We will use the two previous lemmas and Theorem \ref{thm:comparison-principle} to show that this result holds under the less conservative technical assumptions imposed in this work.

\begin{lemma}\label{lem:ra-hj-tvp}
    Let the dynamics function $\dynamics$ be continuous on $\Rn \tms \uVals \tms \dVals \tms (\tInit, \tFin]$, and
    let $\hamiltonian:\Rn \tms \R \tms \Rn \to \R$ be given by \eqref{eqn:hamiltonian}.
    Let $\tInit, \tFin \in \R$ with $\tInit < \tFin$, and let $\tarFn, \conFn: \Rn \tms (\tInit,\tFin] \to \R$ be continuous with $\tarFn < +\ii$ and $\conFn > -\ii$.
    Then $\hjSol = \valRTAT[\tarFn,\conFn;\tFin]|_{\Rn \tms (\tInit, \tFin]}$ is the unique viscosity solution of
    \begin{equation}\label{proof:ra-hj-pde}
    \begin{cases}
     -\min\Big\{ \conFn(\xVal,\tVal) - \hjDep, \\
     \qquad \max\big\{ \tarFn(\xVal,\tVal) - \hjDep, \\
     \qquad \ppt \hjDep + \hamiltonian(\xVal, \tVal, \nabla_x \hjDep) \big\} \Big\} = 0, & \xVal \in \Rn,\tVal \in (\tInit, \tFin), \\
     \hjDep(\xVal,\tFin) = \min\{\conFn(\xVal,\tFin), \tarFn(\xVal,\tFin)\}, & \xVal \in \Rn.
    \end{cases}
\end{equation}
\end{lemma}
\begin{proof} 
    By Lemma \ref{lem:homeomorphism} and the fact that $\valRTAT[\homeo \circ \tarFn, \homeo \circ \conFn; \tFin] = \homeo \circ \valRTAT[\tarFn, \conFn; \tFin]$, it suffices to assume $-1 \le \tarFn < 1$ and $-1 < \conFn \le 1$.
    (Otherwise, we can replace $\tarFn$ with $\homeo \circ \tarFn$ and $\conFn$ with $\homeo \circ \conFn$.)
    
    Let $E \sbs \Rn \tms (\tInit,\tFin)$ be bounded and open.
    We will first show that $\valRTAT[\tarFn,\conFn; \tFin]|_E$ is a viscosity solution of the HJ-PDE
    \begin{align}\label{proof:ra-hj-pde-on-E}
     -\min\big\{ &\conFn(\xVal,\tVal) - \hjDep, \max\big\{ \tarFn(\xVal,\tVal) - \hjDep, \nonumber\\
     &\ppt \hjDep + \hamiltonian(\xVal, \tVal, \nabla_x \hjDep) \big\} \big\} = 0,\quad (\xVal,\tVal) \in E.
    \end{align}
    By Assumptions \ref{assumption:compactness}-\ref{assumption:regularity}, there is some $R > 0$ such that $\|\standardTraj(\tDum)\| < R$ for all $\tDum \in (\tInit, \tFin]$, $\uSig \in \uSigs$, $\dSig \in \dSigs$, and $(\xVal, \tVal) \in E$.
    As such, we can assume that the dynamics function $\dynamics$ is bounded and Lipschitz, since modifying the dynamics outside of $B(0;R) \tms \uVals \tms \dVals \tms (\tInit,\tFin]$ does not change $\valRTAT[\tarFn,\conFn; \tFin]|_E$.

    Now, by density of Lipschitz functions within the continuous real-valued functions on compact sets, we can choose sequences of Lipschitz functions $\tarFn_i, \conFn_i: \Rn \tms (\tInit, \tFin] \to \R$ such that $\tarFn_i \to \tarFn$ and $\conFn_i \to \conFn$ locally uniformly.
    Since $\robRTAT[\tarFn_i,\conFn_i; \tFin] \to \robRTAT[\tarFn,\conFn; \tFin]$ locally uniformly, then by Lemma \ref{lem:limits}, we have that $\valRTAT[\tarFn_i,\conFn_i; \tFin]|_E \to \valRTAT[\tarFn,\conFn; \tFin]|_E$ uniformly.
    Since for each $i \in \N$, Theorem 1 in \cite{fisac-chen-2015} guarantees that $\hjDep = \valRTAT[\tarFn_i,\conFn_i; \tFin]|_E$ is a viscosity solution of 
    \begin{align*}
     -\min\big\{ \conFn_i(\xVal,\tVal) - \hjDep, &\max\big\{ \tarFn_i(\xVal,\tVal) - \hjDep, \\
     &\ppt \hjDep + \hamiltonian(\xVal, \tVal, \nabla_x \hjDep) \big\} \big\} = 0,\quad (\xVal,\tVal) \in E,
    \end{align*}
    it follows from Proposition II.2.2 in \cite{Bardi-Dolcetta-Optimal-Control} that $\valRTAT[\tarFn,\conFn; \tFin]|_E$ is a viscosity solution of \eqref{proof:ra-hj-pde-on-E}.
    
    Since $E$ was an arbitrary open set in $\Rn \tms (\tInit, \tFin]$ and $\valRTAT[\tarFn,\conFn; \tFin](\cdot, \tFin) = \min\{ \conFn(\cdot,\tFin), \tarFn(\cdot,\tFin)\}$, then $\valRTAT[\tarFn,\conFn; \tFin]|_{\Rn \tms (\tInit, \tFin]}$ indeed is a viscosity solution of \eqref{proof:ra-hj-pde}.
    That it is the unique viscosity solution follows from Theorem \ref{thm:comparison-principle}.
\end{proof}

\begin{proof}[Proof of Theorem \ref{thm:dpe}]

By Lemma \ref{lem:homeomorphism} and the fact that $\homeo \circ \valGRA[\graData] = \valGRA[\homeo \circ \tarFn, \homeo \circ \conFn, \homeo \circ \terFn; \tFin]$, it suffices to assume $-1 \le \tarFn < 1$, $-1 < \conFn \le 1$, and $-1 < \terFn < 1$.

For each $i \in \N$, define $\terFn_i:\Rn \tms \R \to \R$ by
\begin{equation}\label{proof:rhat}
    \terFn_i(\xVal, \tVal) = \begin{cases}
    -1 & \tVal \le \tFin - \frac{1}{i},\\
    \frac{(\tVal - (\tFin - \frac{1}{i}))\terFn(\xVal) - (\tFin - \tVal)}{1/i} & \tFin - \frac{1}{i} < \tVal \le \tFin,\\
    \terFn(\xVal) & \tVal > \tFin.
    \end{cases}
\end{equation}
Let $L > 0$, $\compact \ssbs \Lip$ be compact, $\tFinDum < \tFin$, and $i \in \N$ be sufficiently large that $\tFin - \frac{1}{i} > \tFinDum$.
Then we can choose some $R > 0$ such that $\|\xSig(\tVal)\| \le R$ for each $\xSig \in \compact$ and $\tVal \in [\tInit, \tFin]$.

Now fix $\tVal \in [\tFinDum, \tFin]$ and $\xSig \in \compact$.
Then 
\begin{align*}
\robRTAT[\max\{\tarFn, \terFn_i\}, \conFn; \tFin](\xSig, \tVal) =& \max\{\robRTAT[\tarFn, \conFn; \tFin](\xSig, \tVal),\\ 
&\qquad \robRTAT[\terFn_i, \conFn; \tFin](\xSig, \tVal) \},\\
\robGRA[\tarFn, \conFn, \terFn; \tFin](\xSig, \tVal) =& \max\{\robRTAT[\tarFn, \conFn; \tFin](\xSig, \tVal), \\
&\qquad \robRSAT[\conFn, \terFn; \tFin](\xSig, \tVal)\}.
\end{align*}
Thus
\begin{align*}
    &|\robRTAT[\max\{\tarFn, \terFn_i\}, \conFn; \tFin](\xSig, \tVal) - \robGRA[\tarFn, \conFn, \terFn; \tFin](\xSig, \tVal)| \\
    &\le |\robRTAT[\terFn_i, \conFn; \tFin](\xSig, \tVal) - \robRSAT[\conFn, \terFn; \tFin](\xSig, \tVal)|.
\end{align*}
Because $\robRTAT[\terFn_i, \conFn; \tFin](\xSig, \tVal) \ge \robRSAT[\conFn, \terFn; \tFin](\xSig, \tVal)$ we in fact have
\begin{align*}
    &|\robRTAT[\max\{\tarFn, \terFn_i\}, \conFn; \tFin](\xSig, \tVal) - \robGRA[\tarFn, \conFn, \terFn; \tFin](\xSig, \tVal)| \\
    &\le \robRTAT[\terFn_i, \conFn; \tFin](\xSig, \tVal) - \robRSAT[\conFn, \terFn; \tFin](\xSig, \tVal) \\
    &\le \max_{\tDum \in [\tFin - \frac{1}{i}, \tFin]}\min\{\terFn_i(\xSig(\tDum), \tDum), \min_{\tDumDum \in [\tInit, \tDum]} \sconeval{\xSig}{\tDumDum}\} - \\
    &\qquad \min\{\terEval{\xSig}{\tFin}, \min_{\tDumDum \in [\tInit, \tFin]} \sconeval{\xSig}{\tDumDum} \},
\end{align*}
where the latter inequality follows from the fact that $\terFn_i = -1$ on $\Rn \tms (\tInit, \tFin - \frac{1}{i})$.
Note that
\begin{align*}
    &\max_{\tDum \in [\tFin - \frac{1}{i}, \tFin]}\min\{\terFn_i(\xSig(\tDum), \tDum), \min_{\tDumDum \in [\tInit, \tDum]} \sconeval{\xSig}{\tDumDum}\} \\
    &\le \min\{\max_{\tDum \in [\tFin - \frac{1}{i}, \tFin]} \terFn_i(\xSig(\tDum), \tDum), \max_{\tDum \in [\tFin - \frac{1}{i}, \tFin]} \min_{\tDumDum \in [\tInit, \tDum]} \sconeval{\xSig}{\tDumDum}\}\\
    &\le \min\{\max_{\tDum \in [\tFin - \frac{1}{i}, \tFin]} \terFn(\xSig(\tDum)), \min_{\tDumDum \in [\tInit, \tFin - \frac{1}{i}]} \sconeval{\xSig}{\tDumDum}\}.
\end{align*}
It follows that
\begin{align*}
    &|\robRTAT[\max\{\tarFn, \terFn_i\}, \conFn; \tFin](\xSig, \tVal) - \robGRA[\tarFn, \conFn, \terFn; \tFin](\xSig, \tVal)|\\
    &\le \max\{|\max_{\tDum \in [\tFin - \frac{1}{i}, \tFin]} \terFn(\xSig(\tDum)) - \terFn(\xSig(\tFin))|, \\
    &\qquad\qquad\min_{\tDumDum \in [\tInit, \tFin - \frac{1}{i}]} \sconeval{\xSig}{\tDumDum}- \min_{\tDumDum \in [\tInit, \tFin]} \sconeval{\xSig}{\tDumDum}\}\}\\
    &\le \max\{|\max_{\tDum \in [\tFin - \frac{1}{i}, \tFin]} \terFn(\xSig(\tDum)) - \terFn(\xSig(\tFin))|, \\
    &\qquad\qquad \sconeval{\xSig}{\tFin - \frac{1}{i}}- \min_{\tDumDum \in [\tFin - \frac{1}{i}, \tFin]} \sconeval{\xSig}{\tDumDum}\}\}\\
    &\le \sup_{\xDum \in B(0,R)} \sup_{\xDumDum \in B(\xDum, L/i)} \max_{\tDum, \tDumDum \in [\tFin - \frac{1}{i}, \tFin]} \\
    &\qquad\max\{|\terFn(\xDum) - \terFn(\xDumDum)|, |\conFn(\xDum, \tDum) - \conFn(\xDumDum, \tDumDum)|\}.
\end{align*}

As the right-hand-side of the above inequality is independent of $\xSig$ and $\tInit$ and tends to $0$ as $i \to \ii$, it follows that $\robRTAT[\max\{\tarFn, \terFn_i\}, \conFn; \tFin] \to \robGRA[\tarFn, \conFn, \terFn; \tFin]$ locally uniformly on $\Lip \tms (\tInit, \tFin)$.
By Lemma \ref{lem:limits}, $\valRTAT[\max\{\tarFn, \terFn_i\}, \conFn; \tFin] \to \valGRA[\tarFn, \conFn, \terFn; \tFin]$ locally uniformly.

For each $i \in \N$, let $\hjFn_i: \Rn \tms (\tInit, \tFin) \tms \R \tms \Rn \tms \R \to \R$ be given by
\begin{align*}
  \hjFn_i&(\xVal, \tVal, \hjDep, \costate, \tcostate) = \min\{\conFn(\xVal, \tVal) - \hjDep,\\
  &\max\{\max\{\tarFn(\xVal, \tVal), \terFn_i(\xVal, \tVal)\} - \hjDep, \tcostate + \hamiltonian(\xVal, \tVal, \costate)\}\},  
\end{align*}
and let $\hjFn: \Rn \tms (\tInit, \tFin) \tms \R \tms \Rn \tms \R \to \R$ be given by
\begin{align*}
  \hjFn(\xVal, \tVal, \hjDep, \costate, \tcostate) = \min\{&\conFn(\xVal, \tVal) - \hjDep,\\
  &\max\{\tarFn(\xVal, \tVal) - \hjDep, \tcostate + \hamiltonian(\xVal, \tVal, \costate)\}\}.  
\end{align*}

Because $\max\{\tarFn, \terFn_i\} \to \tarFn$ locally uniformly on $\Rn \tms (\tInit, \tFin)$, it follows that $\hjFn_i \to \hjFn$ locally uniformly on this set as well.
Since $\valRTAT[\max\{\tarFn, \terFn_i\}, \conFn; \tFin]$ is a viscosity solution of the HJ-PDE $-\hjFn_i(\xVal, \tVal, \hjDep, \nabla_\xVal \hjSol, \ppt \hjSol) = 0$ by
Lemma \ref{lem:ra-hj-tvp}, it then follows from Proposition II.2.2 in \cite{Bardi-Dolcetta-Optimal-Control} that $\valGRA[\graData]$ is a viscosity solution of $-\hjFn(\xVal, \tVal, \hjDep, \nabla_\xVal \hjSol, \ppt \hjSol) = 0$.
But this latter HJ-PDE is precisely \eqref{eqn:hjr-pde}.

Because $\valGRA[\graData](\cdot, \tFin) = \min\{\conFn(\cdot, \tFin), \max\{\tarFn(\cdot, \tFin), \terFn(\cdot)\}\}$, it then follows that $\valGRA[\graData]$ is a viscosity solution of \eqref{eqn:gra-tvp}.
But by Theorem \ref{thm:comparison-principle}, it is then in fact the unique viscosity solution of \ref{eqn:gra-tvp}.    
\end{proof}

\subsection{A viscosity solution of (\ref{eqn:hjr-pde}) that does not continuously extend to the terminal boundary}\label{sec:pathology}

Consider the HJ-PDE
\begin{equation}\label{eqn:pathological-pde}
    -\partial_\tVal \hjDep - |\nabla_{\xVal} \hjDep| = 0, \quad \xVal \in \R,~\tVal \in (0,1).
\end{equation}
Note that this PDE is equivalent to \eqref{eqn:hjr-pde} when $\tarFn = -\ii$, $\conFn = \ii$, the system dynamics are $\dot{\xSig} = \uSig + \dSig$ (with $\xSig$ scalar), the control bounds are $\uVals := \{ \uVal \in \R \mid |\uVal| \le 1 \}$, and the disturbance bounds are $\dVals = \{0\}$.

Define $\hjSol:\R \tms (0,1) \to \R$ by
\begin{equation*}
    \hjSol(\xVal, \tVal) := -\frac{1}{|\xVal| + 1 - \tVal}.
\end{equation*}
We show that $\hjSol$ is a viscosity solution of $\eqref{eqn:pathological-pde}$.
Let $\xNot \in \R$ and $\tNot \in (0, 1)$.
If $\xNot \ne 0$, then $\hjSol$ is differentiable at $(\xNot,\tNot)$ and in particular:
\begin{align*}
    \ppt \hjSol(\xNot, \tNot) &= -\frac{1}{(|\xNot| + 1 - \tNot)^2} \\ 
    |\nabla_\xVal \hjSol(\xNot, \tNot)| &= \frac{1}{(|\xNot| + 1 - \tNot)^2},
\end{align*}
so $\eqref{eqn:pathological-pde}$ holds at $(\xNot, \tNot)$.

Now we consider the case where $\xNot = 0$ and $\tNot \in (0, 1)$.
Observe that
\begin{align*}
    \nabla_\xVal \hjSol(0^+, \tNot) &= \frac{1}{(1 - \tNot)^2},~\nabla_\xVal \hjSol(0^-, \tNot) &= -\frac{1}{(1 - \tNot)^2}.
\end{align*}
Consider a continuously differentiable $\hjTest: \R \tms (0, 1) \to \R$.
By the above one-sided derivatives, it is impossible for $\hjSol - \hjTest$ to have a local maximum at $(\xNot, \tNot)$, so we can conclude $\hjSol$ is a viscosity sub-solution of \eqref{eqn:pathological-pde}.
Assume instead that $\hjTest$ has a local minimum at $(\xNot, \tNot)$.
Then
\begin{equation*}
    \ppt \hjTest(\xNot, \tNot) = -\frac{1}{(1 - \tNot)^2} ,~ |\nabla_\xVal \hjTest(0, \tNot)| \le \frac{1}{(1 - \tNot)^2}
\end{equation*}
so that
\begin{equation*}
    -\partial_\tVal \hjTest(0, \tNot) - |\nabla_{\xVal} \hjTest(0, \tNot)| \ge 0.
\end{equation*}
Thus we can also conclude $\hjSol$ is a viscosity super-solution of \eqref{eqn:pathological-pde}.

So $\hjSol$ is indeed a viscosity solution of \eqref{eqn:pathological-pde}, but it does not continuously extend to the terminal boundary, $\R \tms \{1\}$, as $\lim_{\tVal \to 1^-} \hjSol(0, \tVal) = -\ii$.

\bibliographystyle{ieeetr}
\bibliography{references}

@ARTICLE{mitchell-2005,
  author={Mitchell, I.M. and Bayen, A.M. and Tomlin, C.J.},
  journal={IEEE Transactions on Automatic Control}, 
  title={A time-dependent {H}amilton-{J}acobi formulation of reachable sets for continuous dynamic games}, 
  year={2005},
  volume={50},
  number={7},
  pages={947-957},
  doi={10.1109/TAC.2005.851439}}

@ARTICLE{lygeros-2011,
  author={Margellos, Kostas and Lygeros, John},
  journal={IEEE Transactions on Automatic Control}, 
  title={{H}amilton–{J}acobi Formulation for Reach–Avoid Differential Games}, 
  year={2011},
  volume={56},
  number={8},
  pages={1849-1861},
  doi={10.1109/TAC.2011.2105730}}

@inproceedings{fisac-chen-2015,
author = {Fisac, Jaime F. and Chen, Mo and Tomlin, Claire J. and Sastry, S. Shankar},
title = {Reach-avoid problems with time-varying dynamics, targets and constraints},
booktitle={Proceedings of the 18th International Conference on Hybrid Systems: Computation and Control (HSCC)},
year = {2015},
isbn = {9781450334334},
doi = {10.1145/2728606.2728612},
pages = {11-20},
location = {Seattle, Washington},
}

@article{evans-1984,
 author = {L. C. Evans and P. E. Souganidis},
 journal = {Indiana University Mathematics Journal},
 volume = {33},
 number = {5},
 pages = {773-797},
 title = {Differential Games and Representation Formulas for Solutions of {H}amilton-{J}acobi-{I}saacs Equations},
 year = {1984},
 doi = {10.1512/iumj.1984.33.33040}
}

@InProceedings{stl-meets-reachability,
author="Chen, Mo
and Tam, Qizhan
and Livingston, Scott C.
and Pavone, Marco",
title="Signal Temporal Logic Meets Reachability: Connections and Applications",
booktitle="Algorithmic Foundations of Robotics XIII",
year="2020",
pages="581-601",
isbn="978-3-030-44051-0"
}

@INPROCEEDINGS{choi-2021,
  author={Choi, Jason J. and Lee, Donggun and Sreenath, Koushil and Tomlin, Claire J. and Herbert, Sylvia L.},
  booktitle={2021 60th IEEE Conference on Decision and Control (CDC)}, 
  title={Robust Control Barrier–Value Functions for Safety-Critical Control}, 
  year={2021},
  volume={},
  number={},
  pages={6814-6821},
  doi={10.1109/CDC45484.2021.9683085}}

@inproceedings{so-2024,
 author = {So, Oswin and Ge, Cheng and Fan, Chuchu},
 booktitle = {Advances in Neural Information Processing Systems (NeurIPS) 37},
 doi = {10.52202/079017-0974},
 pages = {30951-30984},
 title = {Solving Minimum-Cost Reach Avoid using Reinforcement Learning},
 url = {https://proceedings.neurips.cc/paper_files/paper/2024/file/3750e99b522bd36a099d2e8b9f0550c7-Paper-Conference.pdf},
 year = {2024}
}

@inproceedings{hsu-2021, series={RSS2021}, title={Safety and Liveness Guarantees through Reach-Avoid Reinforcement Learning}, url={http://dx.doi.org/10.15607/RSS.2021.XVII.077}, DOI={10.15607/rss.2021.xvii.077}, booktitle={Robotics: Science and Systems XVII}, author={Hsu, Kai-Chieh and Rubies-Royo, Vicenç and Tomlin, Claire and Fisac, Jaime}, year={2021}, month=Jul, collection={RSS2021} }

@book{Bardi-Dolcetta-Optimal-Control,
author = {Bardi, M. and Capuzzo-Dolcetta, I.},
publisher = {Birkhauser},
title = {Optimal Control and Viscosity Solutions of {H}amilton-{J}acobi-{B}ellman Equations },
year = {1997},
}

@article{sharpless2026bellman,
  title={Bellman value decomposition for task logic in safe optimal control},
  author={Sharpless, William and So, Oswin and Hirsch, Dylan and Herbert, Sylvia and Fan, Chuchu},
  journal={arXiv preprint arXiv:2602.19532},
  year={2026}
}

@inproceedings{sharpless2026dual,
  title     = {Dual-Objective Reinforcement Learning with Novel {H}amilton-{J}acobi-{B}ellman Formulations},
  author    = {Sharpless, William and Hirsch, Dylan and Tonkens, Sander and Shinde, Nikhil and Herbert, Sylvia},
  booktitle = {The Fourteenth International Conference on Learning Representations (ICLR)},
  year      = {2026},
}

@inbook{donze-robustness-metric, title={Robust Satisfaction of Temporal Logic over Real-Valued Signals}, ISBN={9783642152979}, ISSN={1611-3349}, url={http://dx.doi.org/10.1007/978-3-642-15297-9_9}, DOI={10.1007/978-3-642-15297-9_9}, booktitle={Formal Modeling and Analysis of Timed Systems}, publisher={Springer Berlin Heidelberg}, author={Donzé, Alexandre and Maler, Oded}, year={2010}, pages={92-106} }

@inbook{release-operator, title={Temporal-logic Queries}, ISBN={9783540450474}, ISSN={1611-3349}, url={http://dx.doi.org/10.1007/10722167_34}, DOI={10.1007/10722167_34}, booktitle={Computer Aided Verification}, publisher={Springer Berlin Heidelberg}, author={Chan, William}, year={2000}, pages={450-463} }

@inproceedings{HJR-Survey, 
 title={{H}amilton-{J}acobi reachability: A brief overview and recent advances}, 
 DOI={10.1109/cdc.2017.8263977}, 
 booktitle={2017 IEEE Conference on Decision and Control (CDC)},
 author={Bansal, Somil and Chen, Mo and Herbert, Sylvia and Tomlin, Claire J.}, 
 year={2017},
 pages={2242-2253} }

@INPROCEEDINGS{Reachability-RL,
  author={Fisac, Jaime F. and Lugovoy, Neil F. and Rubies-Royo, Vicenç and Ghosh, Shromona and Tomlin, Claire J.},
  booktitle={2019 International Conference on Robotics and Automation (ICRA)}, 
  title={Bridging {H}amilton-{J}acobi Safety Analysis and Reinforcement Learning}, 
  year={2019},
  volume={},
  number={},
  pages={8550-8556},
  doi={10.1109/ICRA.2019.8794107}}

@INPROCEEDINGS{Xiang-CDC-2025,
  author={Chen, Yu and Li, Shaoyuan and Yin, Xiang},
  booktitle={2025 IEEE Conference on Decision and Control (CDC)}, 
  title={Control Synthesis for Multiple Reach-Avoid Tasks via {H}amilton-{J}acobi Reachability Analysis}, 
  year={2025},
  volume={},
  number={},
  pages={5980-5985},
  doi={10.1109/CDC57313.2025.11312352}}

@INPROCEEDINGS{Jiang-2020,
  author={Jiang, Frank J. and Gao, Yulong and Xie, Lihua and Johansson, Karl H.},
  booktitle={2020 59th IEEE Conference on Decision and Control (CDC)}, 
  title={Ensuring safety for vehicle parking tasks using {H}amilton-{J}acobi reachability analysis}, 
  year={2020},
  volume={},
  number={},
  pages={1416-1421},
  doi={10.1109/CDC42340.2020.9304186}}

@inproceedings{Jiang-2024, title={Guaranteed Completion of Complex Tasks via Temporal Logic Trees and {H}amilton-{J}acobi Reachability}, url={http://dx.doi.org/10.1109/CDC56724.2024.10886233}, DOI={10.1109/cdc56724.2024.10886233}, booktitle={2024 IEEE 63rd Proc. Conference on Decision and Control (CDC)}, author={Jiang, Frank J. and Arfvidsson, Kaj Munhoz and He, Chong and Chen, Mo and Johansson, Karl H.}, year={2024}, month=dec, pages={5203-5210} }

@INPROCEEDINGS{viscosity-cbfs,
  author={Hirsch, Dylan and Fisac, Jaime Fernández and Herbert, Sylvia},
  booktitle={2026 American Control Conference (ACC)}, 
  title={Viscosity {CBF}s: Bridging the Control Barrier Function and {H}amilton-{J}acobi Reachability Frameworks in Safe Control Theory}, 
  year={2026},
  volume={},
  number={},
  pages={593-600},
  doi={}}

@inproceedings{Ganai_2023,
 title={Iterative Reachability Estimation for Safe Reinforcement Learning}, url={http://dx.doi.org/10.52202/075280-3058}, DOI={10.52202/075280-3058}, booktitle={Advances in Neural Information Processing Systems (NeurIPS) 36}, author={Ganai, Milan and Gong, Zheng and Yu, Chenning and Herbert, Sylvia and Gao, Sicun}, year={2023}, pages={69764-69797}, collection={NeurIPS 2023} }

@article{Ganai_2024, title={{H}amilton-{J}acobi Reachability in Reinforcement Learning: A Survey}, volume={3}, ISSN={2694-085X}, url={http://dx.doi.org/10.1109/OJCSYS.2024.3449138}, DOI={10.1109/ojcsys.2024.3449138}, journal={IEEE Open Journal of Control Systems}, author={Ganai, Milan and Gao, Sicun and Herbert, Sylvia L.}, year={2024}, pages={310-324} }

@article{Lee_2023, title={Convexifying State-Constrained Optimal Control Problem}, volume={68}, ISSN={2334-3303}, url={http://dx.doi.org/10.1109/TAC.2022.3221704}, DOI={10.1109/tac.2022.3221704}, number={9}, journal={IEEE Transactions on Automatic Control}, author={Lee, Donggun and Deka, Shankar A. and Tomlin, Claire J.}, year={2023}, month=Sep, pages={5608-5615} }

@article{Bokanowski_2010, title={Reachability and Minimal Times for State Constrained Nonlinear Problems without Any Controllability Assumption}, volume={48}, ISSN={1095-7138}, url={http://dx.doi.org/10.1137/090762075}, DOI={10.1137/090762075}, number={7}, journal={SIAM Journal on Control and Optimization}, author={Bokanowski, Olivier and Forcadel, Nicolas and Zidani, Hasnaa}, year={2010}, month=Jan, pages={4292-4316} }

@book{Friedman-Differential-Games,
author = {Friedman, Avner},
edition = {{D}over},
publisher = {Dover Publications},
title = {Differential Games},
year = {2013},
}

@book{Isaacs-Differential-Games,
author = {Isaacs, Rufus},
publisher = {Dover Publications},
title = {Differential Games: a Mathematical Theory with Applications to Warfare and Pursuit Control and Optimization},
year = {2012},
}

@article{hirsch-2206,
  doi = {10.48550/ARXIV.2607.14023},
  url = {https://arxiv.org/abs/2607.14023},
  author = {Hirsch, Dylan and Sharpless, William and Herbert, Sylvia},
  title = {Exact Decomposition of Adversarial Dual-Objective Value Functions, with Applications to Optimal Drug Dosing},
  journal = {arXiv preprint arXiv:2607.14023},
  year = {2026},
  copyright = {Creative Commons Attribution 4.0 International}
}

@inproceedings{Lee_2020, title={{H}amilton-{J}acobi Formulation for State-Constrained Optimal Control and Zero-Sum Game Problems}, url={http://dx.doi.org/10.1109/CDC42340.2020.9304334}, DOI={10.1109/cdc42340.2020.9304334}, booktitle={2020 59th IEEE Conference on Decision and Control (CDC)},
 author={Lee, Donggun and Keimer, Alexander and Bayen, Alexandre M. and Tomlin, Claire J.}, year={2020}, month=Dec, pages={1078-1085} }

@article{Altarovici_2012,
 title={A general {H}amilton-{J}acobi framework for non-linear state-constrained control problems}, volume={19}, ISSN={1262-3377}, url={http://dx.doi.org/10.1051/cocv/2012011}, DOI={10.1051/cocv/2012011}, number={2}, journal={ESAIM: Control, Optimisation and Calculus of Variations}, author={Altarovici, Albert and Bokanowski, Olivier and Zidani, Hasnaa}, year={2012}, month=Jun, pages={337-357} }

@article{Fisac_robotics, title={A General Safety Framework for Learning-Based Control in Uncertain Robotic Systems}, volume={64}, ISSN={2334-3303}, url={http://dx.doi.org/10.1109/TAC.2018.2876389}, DOI={10.1109/tac.2018.2876389}, number={7}, journal={IEEE Transactions on Automatic Control}, author={Fisac, Jaime F. and Akametalu, Anayo K. and Zeilinger, Melanie N. and Kaynama, Shahab and Gillula, Jeremy and Tomlin, Claire J.}, year={2019}, month=Jul,pages={2737-2752} }

@inproceedings{Pandya_2025, title={Robots that Learn to Safely Influence via Prediction-Informed Reach-Avoid Dynamic Games}, url={http://dx.doi.org/10.1109/ICRA55743.2025.11128803}, DOI={10.1109/icra55743.2025.11128803}, booktitle={2025 IEEE International Conference on Robotics and Automation (ICRA)}, author={Pandya, Ravi and Liu, Changliu and Bajcsy, Andrea}, year={2025}, month=May, pages={14330-14337} }

@ARTICLE{reach-and-stabilize-avoid,
  author={Li, Boyang and Gong, Zheng and Herbert, Sylvia},
  journal={IEEE Transactions on Automatic Control}, 
  title={Solving Reach- and Stabilize-Avoid Problems Using Discounted Reachability}, 
  year={2026},
  volume={},
  number={},
  pages={1-16},
  doi={10.1109/TAC.2026.3693989}}

@article{FasTrack, title={FaSTrack:A Modular Framework for Real-Time Motion Planning and Guaranteed Safe Tracking}, volume={66}, ISSN={2334-3303}, url={http://dx.doi.org/10.1109/TAC.2021.3059838}, DOI={10.1109/tac.2021.3059838}, number={12}, journal={IEEE Transactions on Automatic Control}, author={Chen, Mo and Herbert, Sylvia L. and Hu, Haimin and Pu, Ye and Fisac, Jaime Fernandez and Bansal, Somil and Han, SooJean and Tomlin, Claire J.}, year={2021}, month=Dec, pages={5861-5876} }

@InProceedings{pmlr-v305-tonkens25a,
  title = 	 {From Space to Time: Enabling Adaptive Safety with Learned Value Functions via Disturbance Recasting},
  author =       {Tonkens, Sander and Shinde, Nikhil Uday and Begzadi\'{c}, Azra and Yip, Michael C. and Cortes, Jorge and Herbert, Sylvia Lee},
  booktitle = 	 {Proceedings of The 9th Conference on Robot Learning},
  pages = 	 {4103-4122},
  year = 	 {2025},
  volume = 	 {305},
  month = 	 {Sep}
}

@article{efficient-state-constrained, title={Efficient Computation of State-Constrained Reachability Problems Using {H}opf–{L}ax Formulae}, volume={68}, ISSN={2334-3303}, url={http://dx.doi.org/10.1109/TAC.2023.3241180}, DOI={10.1109/tac.2023.3241180}, number={11}, journal={IEEE Transactions on Automatic Control}, author={Lee, Donggun and Tomlin, Claire J.}, year={2023}, month=Nov, pages={6481-6495} }

@book{baier2008principles,
  title={Principles of model checking},
  author={Baier, Christel and Katoen, Joost-Pieter and Larsen, Kim Guldstrand},
  volume={735},
  year={2008},
  publisher={MIT press Cambridge}
}

@inbook{STL, title={Monitoring Temporal Properties of Continuous Signals}, ISBN={9783540302063}, ISSN={1611-3349}, url={http://dx.doi.org/10.1007/978-3-540-30206-3_12}, DOI={10.1007/978-3-540-30206-3_12}, booktitle={Formal Techniques, Modelling and Analysis of Timed and Fault-Tolerant Systems}, publisher={Springer Berlin Heidelberg}, author={Maler, Oded and Nickovic, Dejan}, year={2004}, pages={152-166} }

@article{Gammoudi_2023, title={A differential game control problem with state constraints}, volume={13}, ISSN={2156-8499}, url={http://dx.doi.org/10.3934/mcrf.2022008}, DOI={10.3934/mcrf.2022008}, number={2}, journal={Mathematical Control and Related Fields}, publisher={American Institute of Mathematical Sciences (AIMS)}, author={Gammoudi, Nidhal and Zidani, Hasnaa}, year={2023}, pages={554-582} }

@article{Chen_2018, title={{H}amilton–{J}acobi Reachability: Some Recent Theoretical Advances and Applications in Unmanned Airspace Management}, volume={1}, ISSN={2573-5144}, url={http://dx.doi.org/10.1146/annurev-control-060117-104941}, DOI={10.1146/annurev-control-060117-104941}, number={1}, journal={Annual Review of Control, Robotics, and Autonomous Systems}, author={Chen, Mo and Tomlin, Claire J.}, year={2018}, month=May, pages={333-358} }

@book{Folland-Real-Analysis,
author = {Folland, Gerald B},
edition = {2},
publisher = {John Wiley \& Sons, Inc.},
title = {Real Analysis: Modern Techniques and Their Applications},
year = {1999},
}

@article{Ishii-uniqueness-unbounded-1984,
 author = {Hitoshi Ishii},
 title = {Uniqueness of Unbounded Viscosity Solution of {H}amilton-{J}acobi Equations},
 journal = {Indiana University Mathematics Journal},
 volume = {33},
 number = {5},
 pages = {721--748},
 year = {1984},
 doi = {10.1512/iumj.1984.33.33038}
}

@ARTICLE{Fialho-Georgiou-TAC-Worst-Case-Analysis-1999,
  author={Fialho, I.J. and Georgiou, T.T.},
  journal={IEEE Transactions on Automatic Control}, 
  title={Worst case analysis of nonlinear systems}, 
  year={1999},
  volume={44},
  number={6},
  pages={1180-1196},
  doi={10.1109/9.769372}}

@INPROCEEDINGS{deepreach,
  author={Bansal, Somil and Tomlin, Claire J.},
  booktitle={2021 IEEE International Conference on Robotics and Automation (ICRA)}, 
  title={{D}eep{R}each: A Deep Learning Approach to High-Dimensional Reachability}, 
  year={2021},
  volume={},
  number={},
  pages={1817-1824},
  doi={10.1109/ICRA48506.2021.9561949}}

@INPROCEEDINGS{lin-bansal-2023,
  author={Lin, Albert and Bansal, Somil},
  booktitle={2023 IEEE International Conference on Robotics and Automation (ICRA)}, 
  title={Generating Formal Safety Assurances for High-Dimensional Reachability}, 
  year={2023},
  volume={},
  number={},
  pages={10525-10531},
  doi={10.1109/ICRA48891.2023.10160600}}

\end{document}